\documentclass[12pt,reqno]{amsart}
\usepackage{main}

\graphicspath{{Figures/}}
\ifdefined\pdfinclusioncopyfonts\pdfinclusioncopyfonts=1\relax\fi
\newcounter{mainproposition}

\usepackage{csquotes}
\usepackage[style=apa,backend=biber,natbib=true]{biblatex}
\usepackage{titletoc}
\newcommand\DoToC{%
  \startcontents
  \printcontents{}{0}{\textbf{Contents}\vskip1em\hrule\vskip1em}
  \vskip1em\hrule\vskip5pt
}

\title{\Large Design-Conditional Prior Elicitation for Dirichlet Process Mixtures:\\
A Unified Framework for Cluster Counts and Weight Control}
\author{JoonHo Lee}
\date{\footnotesize November 17, 2026. \\[0.5em]
Lee: College of Education, The University of Alabama, jlee296@ua.edu.
\\[0.5em]
This research was supported by the Institute of Education Sciences, U.S. Department of Education, through Grant R305D240078 to the University of Alabama. The opinions expressed are those of the author and do not represent views of the Institute or the U.S. Department of Education.}

\hypersetup{
  pdftitle={Design-Conditional Prior Elicitation for Dirichlet Process Mixtures: A Unified Framework for Cluster Counts and Weight Control},
  pdfauthor={JoonHo Lee (jlee296@ua.edu)},
  pdfsubject={Article with Online Supplemental Materials},
  pdfkeywords={Dirichlet process mixture, concentration parameter, prior elicitation, multisite trials, item response theory}
}

\numberwithin{equation}{section}

\makeatletter

\renewcommand{\theHALG@line}{\thealgorithm.\arabic{ALG@line}}
\makeatother

\ifcsname c@recipebox\endcsname\else
  \floatstyle{plaintop}
  \newfloat{recipebox}{t}{lob}
  \floatname{recipebox}{Box}
  \floatstyle{plain}
  \mdfsetup{linewidth=0.6pt,innertopmargin=9pt,innerbottommargin=7pt,innerleftmargin=10pt,innerrightmargin=10pt,skipabove=4pt,skipbelow=4pt}
\fi
\newif\ifreplicationreleased
\replicationreleasedfalse
\newcommand{\replicationcommit}{827cf435}
\newcommand{\replicationcommiturl}{https://github.com/joonho112/DPprior-paper-replication/commit/827cf435addbd23aef656c4378f44ae67cc1e8df}
\algrenewcommand{\algorithmicrequire}{\textbf{Inputs:}}
\algrenewcommand{\algorithmicensure}{\textbf{Output:}}

\begin{document}

\begin{abstract}
Dirichlet process mixtures (DPMs) model variation in effects across sites or studies and in latent traits among examinees. The concentration parameter affects both the number of clusters in the sample and the population weights, yet its hyperprior is often chosen by default. We describe a method for choosing this hyperprior when the number of units is fixed by design. We calibrate a Gamma hyperprior to the analyst's expected cluster count and uncertainty about it, then examine the implied weights. If count and weight judgments conflict, Dual-Anchor selects a compromise and reports its implications for both. Across our simulations and applications, the choice of prior could change the inferred number of clusters considerably even when individual estimates changed little. A higher cluster count under an alternative prior warrants a closer look at the upper tail, where estimates can shift even when the median remains stable. The DPprior R package implements the calibration and diagnostics.
\end{abstract}

\maketitle


\pagestyle{plain}
\newpage

\section{Introduction}
\label{sec:intro}

Multisite trials and meta-analyses in education often ask how an intervention's effect varies across schools, districts, or studies \citep{bloom2017using, weiss2017much}. A related question in psychometrics concerns the distribution of a latent trait among examinees. In both cases, we may wish to allow a skewed or multimodal distribution rather than assume normality. Yet the data may provide little information about its shape. Site and study effects are estimated with error \citep{lee2025improving}, and each examinee answers a limited set of items. The prior distribution can therefore have a substantial influence on the estimated heterogeneity.

Dirichlet process mixture (DPM) models provide a flexible way to model these distributions. Dirichlet process priors have been used for site effects in multisite trials \citep{lee2025improving}, group-level effects in observational multilevel data \citep{antonelli_mitigating_2016}, and study effects in meta-analysis \citep{burr2005bayesian, cao2024bayesian, muthukumarana2016meta}. In item response models, \citet{duncan2008nonparametric} and \citet{sanmartin2011bayesian} placed a Dirichlet process prior on the ability distribution; the latter also studied identification. Related developments include flexible item characteristic curves \citep{miyazaki2009bayesian}, nonnormal latent traits in Rasch models \citep{FinchEdwards2016Rasch}, latent class models with mixed item types \citep{qiu2025bayesian}, and multiple-rater models \citep{MignemiManolopoulou2025bnpMultipleRaters}.

The concentration parameter governs the prior distributions of both the number of clusters in the sample and the population weights. It is commonly assigned a Gamma hyperprior \citep{escobar1995bayesian}, often with a default setting such as Gamma(1, 1) \citep[e.g.,][]{cao2024bayesian, ma2020bayesian, MignemiManolopoulou2025bnpMultipleRaters}. When the data are only moderately informative, posterior summaries of the partition can depend on this choice \citep{dorazio2009selecting, giordano2023evaluating, lee2025improving}. For 100 units, Gamma(1, 1) implies an expected cluster count of about 4.8 and a 59 percent prior probability that a randomly chosen unit belongs to a component holding more than half of the population mass. A prior that seems vague can still favor mixtures in which one component accounts for a large share of the population.

In a multisite trial, investigators who find smaller intervention effects in some schools may ask what distinguishes those schools and how common smaller effects might be in the wider population. Schools assigned to the same mixture component form a cluster. The cluster count tells us how many components of the population mixture are represented in the trial \citep{mccullagh2008many}. If one component has a smaller mean effect, its population weight tells us what share of schools the model assigns to that component. In a meta-analysis, such a finding may prompt a closer look at study-level moderators. In an item response model, a component centered at lower ability may help describe the lower tail, and its weight tells us how much it contributes to the ability distribution. We can use these summaries to guide further investigation without assuming that the components identify distinct types of schools, studies, or examinees.

The number of clusters we expect in a sample of this size can also guide our choice of prior for the concentration parameter \citep{murugiah2012selecting}. For example, \citet{bonafos2024dirichlet} used prior knowledge about infant vocalizations to choose an expected count, and \citet{kottas2006dirichlet} compared Gamma settings through their implied counts. Several authors have used the induced cluster-count distribution to guide prior choice \citep{greve2022spying, paganin2023computational}. \citet{dorazio2009selecting} selected Gamma parameters by minimizing the Kullback--Leibler divergence from a target count distribution. \citet{murugiah2012selecting} instead specified tail probabilities of the count and solved for the parameters at the design size. \citet{paganin2023computational} compared candidate settings through the prior mean and variance of the number of clusters at their sample sizes.

Two difficulties remain. First, translating a judgment about the count into Gamma parameters requires numerical calibration. A grid search for the Gamma parameters in \citeauthor{dorazio2009selecting}'s \citeyearpar{dorazio2009selecting} method can be costly because each candidate pair requires a new numerical integration. This cost matters when several judgments about the cluster count are plausible, because each judgment requires a separate calibration. Evaluating posterior sensitivity adds the cost of fitting alternative priors \citep{giordano2023evaluating}; empirical Bayes estimation of the concentration parameter can also require repeated MCMC runs \citep{antonelli_mitigating_2016, dorazio2008modeling}. Some applications instead use a fixed Gamma setting and assess sensitivity to alternative values \citep{leslie2007general, ma2020bayesian}. Second, matching the count does not ensure plausible population weights. For instance, expecting ten clusters among a hundred units does not imply that the population weights will be similar. \citet{vicentini2025prior} show that priors that appear diffuse over the count can be highly informative about the weights. We use the term \emph{unintended prior} for a discrepancy between these implications and the analyst's judgment.

Count calibration also depends on the number of units. \citet{vicentini2025prior} criticize this sample-size dependence and propose sample-size-independent (SSI) anchors based on the weights. Their objection is especially relevant when units accumulate over time or several data sets share a concentration parameter. Here we consider a fixed set of sites, a defined collection of studies, or a specified sample of examinees. We therefore ask how many clusters we expect among these units, conditioning the judgment on the number of units in the design. The distribution of their effects is a separate quantity of interest. We retain the weight diagnostics because even a plausible count calibration can imply implausible concentration of mass.

In this paper we develop design-conditional elicitation for these fixed designs. We translate a stated mean and variance of the cluster count into a Gamma hyperprior using a closed-form approximation followed by Newton refinement to the exact moments. We then report the probability that a randomly chosen unit belongs to a component holding more than half the population mass. We also give bounds on the probability that such a component exists. If the count and size judgments conflict, we provide a procedure for choosing a count/weight compromise and a separate sensitivity analysis for a required weight bound. We describe how to record the basis for a count judgment, express uncertainty about it and check the implications of the resulting prior. The \texttt{DPprior} R package implements these calculations.

When we interpret a cluster count or a component weight,\phantomsection\label{sec:intro-interpretation} we check whether the prior gives plausible values for both. When our interest lies in the effects themselves, we use calibration to select plausible priors with different implications for counts and weights. We then ask how much the estimated distribution of effects, or the estimate for a particular unit, changes across those priors.

Sections~\ref{sec:theory} and~\ref{sec:method} describe the count and weight distributions and show how we use them to choose a prior. In Section~\ref{sec:simulation} we compare posterior counts with the generating counts in simulations and assess how accurately we recover the site effects and their distribution. In Section~\ref{sec:applications} we apply the method to the Tennessee STAR class-size experiment \citep{krueger1999experimental,nye2000effects}, a writing-to-learn meta-analysis \citep{bangertdrowns2004effects}, and a Rasch analysis of vocabulary checklist data \citep{openpsychometrics2016gcbs,domingue2025irw}. Using our own illustrative judgments, we examine how the choice of prior affects the cluster counts, the distribution of latent effects among the observed units, the fitted population density, and the individual estimates. Section~\ref{sec:discussion} discusses the implications and limitations. The Online Supplemental Materials (OSM) give derivations, extended results, and software details.

\section{Two Interpretations of the Concentration Parameter}
\label{sec:theory}

For a sample of a given size, a larger concentration parameter gives more clusters on average. It also makes it less likely that a randomly chosen unit belongs to a component that accounts for more than half the population.

\subsection{Setup and Notation}
\label{subsec:setup}

We consider $J$ exchangeable units indexed by $j = 1, \ldots, J$, each with a latent effect $\tau_j$ (for instance a site-specific treatment effect) observed with known estimation error \citep{rubin1981estimation}. The model we fit throughout is a Dirichlet process mixture of normals on the unit effects \citep{escobar1995bayesian},
\begin{equation}\label{eq:model}
\begin{array}{@{}r@{\qquad}l@{}}
  \hat{\tau}_j \mid \tau_j
    \sim \Norm\!\left(\tau_j,\; \widehat{se}_j^{\,2}\right),
  & \tau_j \mid (\mu_j, \sigma_j^2)
    \sim \Norm\!\left(\mu_j,\; \sigma_j^2\right), \\[2pt]
  (\mu_j, \sigma_j^2) \mid G \iid G,
  & G \sim \DP(\alpha, G_0), \qquad \alpha \sim \mathrm{Gamma}(a, b).
\end{array}
\end{equation}
The first line separates estimation error in $\hat{\tau}_j$ from variation within an effect component. The second line assigns the component parameters to a Dirichlet process. The base measure $G_0$ is a normal distribution on the component mean times an inverse gamma distribution on the component variance, and $\alpha > 0$ controls how readily additional units form new clusters \citep{ferguson1973bayesian, escobar1995bayesian}.

Because $G$ is discrete \citep{sethuraman1994constructive}, some units share component parameters $(\mu_j, \sigma_j^2)$. These ties partition the $J$ units into $K_J$ clusters. Within each component, unit effects can vary continuously around the component mean. For item response models, we replace the observation model in Equation~\eqref{eq:model} by the Rasch likelihood $y_{ij} \sim \Bernoulli\{\mathrm{logit}^{-1}(\eta_j - \beta_i)\}$ for examinee $j$ and item $i$. The latent trait $\eta_j$ replaces $\tau_j$, and $J = N$ examinees \citep{paganin2023computational}.

The stick-breaking representation \citep{sethuraman1994constructive, ishwaran2000markov} is $G = \sum_{h \ge 1} w_h \delta_{\theta_h}$, where $\theta_h \iid G_0$ and the weights follow the $\GEM(\alpha)$ law. In particular, $w_1 \mid \alpha \sim \Beta(1, \alpha)$. The weights occur in size-biased order, not decreasing order \citep{vicentini2025prior}. We use the shape--rate parameterization $\alpha \sim \mathrm{Gamma}(a, b)$ throughout, so $\E(\alpha) = a/b$.

\subsection{Partition Lens and Weight Lens}
\label{subsec:lenses}

Given $\alpha$, the cluster count $K_J$ follows the Antoniak distribution \citep{antoniak1974mixtures}. The indicator that unit $i$ opens a new cluster has probability $\alpha/(\alpha + i - 1)$, independently of the preceding indicators. Summing their means and variances gives
\begin{equation}\label{eq:antoniak}
  \E(K_J \mid \alpha) = \sum_{i=1}^{J} \frac{\alpha}{\alpha + i - 1},
  \qquad
  \Var(K_J \mid \alpha) = \sum_{i=1}^{J} \frac{\alpha\,(i - 1)}{(\alpha + i - 1)^{2}}
\end{equation}
(the exact probability mass function and Gamma-mixed forms are in OSM Appendix~\ref{app:dp-crp}). Larger $\alpha$ gives more clusters. In Section~\ref{sec:method}, we match the Gamma-mixed mean and variance to the analyst's stated count judgment.

The mixture weights describe how population mass is distributed among mixture components. The first stick-breaking weight has distribution $w_1 \mid \alpha \sim \Beta(1, \alpha)$. A fresh size-biased pick $W_{\mathrm{SB}}$ from the ranked masses has the same marginal law \citep{sethuraman1994constructive, vicentini2025prior}. \citet{vicentini2025prior} derive its marginal distribution under a Gamma hyperprior. Equivalently, integrating the conditional tail $(1 - t)^{\alpha}$ over $\alpha \sim \mathrm{Gamma}(a, b)$ gives:
\begin{equation}\label{eq:w1-tail}
  \Pr(W_{\mathrm{SB}} > t \mid \alpha) = (1 - t)^{\alpha},
  \qquad
  \Pr(W_{\mathrm{SB}} > t) = \left( \frac{b}{b - \log(1 - t)} \right)^{\!a},
  \qquad 0 < t < 1 .
\end{equation}
This is the prior probability that a new unit belongs to a mixture component with more than a fraction $t$ of the total mass. By exchangeability, it is also the probability for any given unit. We use this size-biased tail as the diagnostic and second anchor in Section~\ref{subsec:dual-anchor}.

Two related quantities answer different questions. The largest-atom probability $\Pr(W_{\max} > t)$ concerns whether any mixture component exceeds $t$; the largest observed share is the fraction of the $J$ units assigned to the biggest cluster in a fitted model. For the three applications, OSM Table~\ref{tab:H-all-posteriors} gives posterior means for both the largest observed share and the largest population weight. Proposition~\ref{prop:largest-weight} bounds the largest-atom probability using the size-biased tail; the software also computes it exactly.

\medskip
\refstepcounter{mainproposition}\label{prop:largest-weight}
\noindent\textbf{Proposition~\themainproposition{} (Bounds on the largest weight).}
Let $W_{\max}$ be the largest ranked mass, let $W_{\mathrm{SB}}$ be a fresh size-biased pick from those masses, and let $t \ge 1/2$. Then, for every fixed $\alpha$, $\Pr(W_{\mathrm{SB}} > t \mid \alpha) \le \Pr(W_{\max} > t \mid \alpha) \le \Pr(W_{\mathrm{SB}} > t \mid \alpha)/t$. The same inequalities hold after averaging over $\alpha \sim \mathrm{Gamma}(a, b)$.
\medskip

To prove the result, condition on the ranked masses (OSM Appendix~\ref{app:weights}). At most one mass can exceed $t \ge 1/2$, so $\Pr(W_{\mathrm{SB}} > t) = \E\{W_{\max}\,\mathbf{1}(W_{\max} > t)\}$ lies between $t\Pr(W_{\max} > t)$ and $\Pr(W_{\max} > t)$. At $t = 0.9$ the size-biased tail determines the largest-weight probability to within 11 percent; at $t = 0.5$, to within a factor of two. Exact values of $\Pr(W_{\max} > t)$ require only one-dimensional integration (OSM Appendix~\ref{app:weights}).

We also use the co-clustering index $\rho = \sum_h w_h^2$, the probability that two randomly selected units share a cluster, with $\E(\rho \mid \alpha) = 1/(1 + \alpha)$ \citep{antoniak1974mixtures}. Table~\ref{tab:functionals} summarizes the interpretation and reporting of each quantity.

\begin{table}[!tp]
\centering
\caption{Partition and weight functionals of the concentration parameter, with the question each one answers for the analyst and the form in which it is reported}
\label{tab:functionals}
\begingroup
\normalsize
\setstretch{1.08}
\renewcommand{\arraystretch}{1.18}
\setlength{\tabcolsep}{5pt}
\begin{tabularx}{\textwidth}{@{}>{\raggedright\arraybackslash}p{0.25\textwidth} >{\raggedright\arraybackslash}X >{\raggedright\arraybackslash}X@{}}
\toprule
Quantity & Question for the analyst & How to report \\
\midrule
$\E(K_J)$, $\Var(K_J)$ & How many clusters do I expect among these $J$ units, and how sure am I? & ``About $X$ clusters; 80\% interval $[L, U]$.'' \\ \addlinespace[5pt]
$\Pr(K_J = k)$ & What is the chance of exactly $k$ clusters? & Mode and central interval of $K_J$. \\ \addlinespace[5pt]
\makecell[l]{$\Pr(W_{\mathrm{SB}} > 0.5)$\\$\Pr(W_{\mathrm{SB}} > 0.9)$} & If I pick a unit at random, could it fall in a component holding more than half (more than 90\%) of the mass? & ``$X$\% chance of belonging to a majority-mass (near-universal-mass) population component.'' \\ \addlinespace[5pt]
$\Pr(W_{\max} > t)$ & Could the largest component hold more than a fraction $t$ of the mass? & Bounds from Proposition~\ref{prop:largest-weight} or the exact value. \\ \addlinespace[5pt]
$\rho = \sum_h w_h^2$ & If I pick two units at random, will they share a mixture component? & ``$X$\% chance of same-cluster membership.'' \\
\bottomrule
\end{tabularx}
\endgroup

\medskip
\noindent\parbox{\textwidth}{\small\textit{Note.} All quantities have closed forms or one-dimensional integrals under $\alpha \sim \mathrm{Gamma}(a, b)$ (OSM Appendices~\ref{app:dp-crp} and~\ref{app:weights}).}
\end{table}

\subsection{Fixed Designs and the Unintended-Prior Phenomenon}
\label{subsec:unintended}

In multisite trials, meta-analyses and psychometric studies, the design fixes the number of units. When interest concerns the distribution of effects or traits across these $J$ units, $K_J$ provides a natural starting point for elicitation. Figure~\ref{fig:1} shows why we also examine the weights. Under this hyperprior the expected count is about five for $J=50$ (Panel~A), and the weight probabilities are $\Pr(W_{\mathrm{SB}} > 0.5) = 0.49$ and $\Pr(W_{\mathrm{SB}} > 0.9) = 0.18$ (Panel~B). Thus, a randomly chosen unit has nearly even odds of belonging to a component with more than half of the mass. This implication would be surprising if by ``about five clusters'' the analyst meant components of comparable size.

\begin{figure}[!tp]
  \centering
  \includegraphics[width=\textwidth]{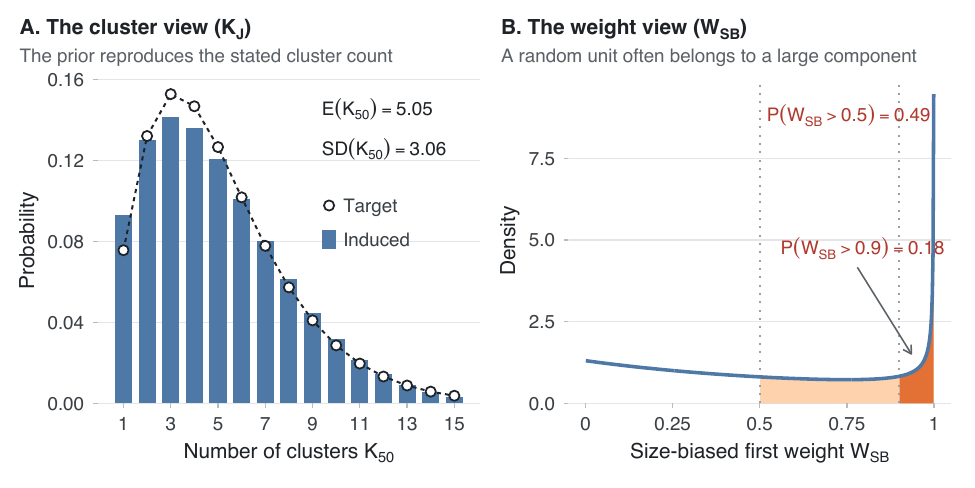}
  \caption{The unintended prior for $J = 50$ under $\alpha \sim \mathrm{Gamma}(1.60, 1.22)$. In Panel~A the bars give the distribution of $K_{50}$ (mean 5.05, SD 3.06). The circles give the probabilities for a shifted negative binomial distribution matched to these moments and normalized over counts 1 to 50. Panel~B shows the marginal density of the size-biased first weight, with tails above 0.5 and 0.9 shaded.}
  \label{fig:1}
\end{figure}

Table~\ref{tab:vif} compares four calibrations with $J=100$ and $\mu_K=5$ (Section~\ref{subsec:elicit}). Raising the target variance from 5 to 20 increases $\Pr(W_{\mathrm{SB}}>0.5)$ from 0.53 to 0.61 and $\Pr(W_{\mathrm{SB}}>0.9)$ from 0.14 to 0.36. The prior probability of a small concentration parameter, $\Pr(\alpha<0.1)$, rises from less than 0.001 to 0.182.

A one-parameter Dirichlet process cannot freely specify both few clusters and little dominance. If the count moments identify the Gamma shape and rate locally, specifying these moments also fixes the weight probabilities locally. We therefore pair count elicitation with weight diagnostics at fixed thresholds. If the implied dominance is unacceptable, the analyst must choose and report the count/weight trade-off produced by Dual-Anchor refinement.

Design size also changes the interpretation of a default. In item response models, $J=N$ may range from a few hundred to several thousand examinees. For fixed $\alpha$ and large $J$, the expected count grows approximately as $\alpha\log J$ (Equation~\eqref{eq:antoniak}). It also increases with $J$ under a fixed Gamma hyperprior. The size-biased majority probability, whether conditional on $\alpha$ or averaged over that hyperprior, is independent of $J$ (Equation~\eqref{eq:w1-tail}).

Figure~\ref{fig:defaults} shows both quantities for $J = 10$ to 5,000 under three Gamma settings in use, $\alpha = 1$, and two forms of a count-calibrated prior. Under $\mathrm{Gamma}(1, 1)$, the expected count rises from about 4 at $J = 50$ to about 8 at $J = 2{,}000$ (Panel~A), while the majority tail remains constant (Panel~B). Calibrating once at $J = 500$ and transporting the prior produces the same upward drift. Under the default we would thus expect about twice as many clusters among 2,000 examinees as among 50 sites. In either design, a randomly selected unit has the same chance of belonging to a component holding a majority of the population mass. Section~\ref{sec:method} gives a procedure for choosing the prior at the intended design size and checking its weight implications.

\begin{figure}[!tp]
  \centering
  \includegraphics[width=\textwidth]{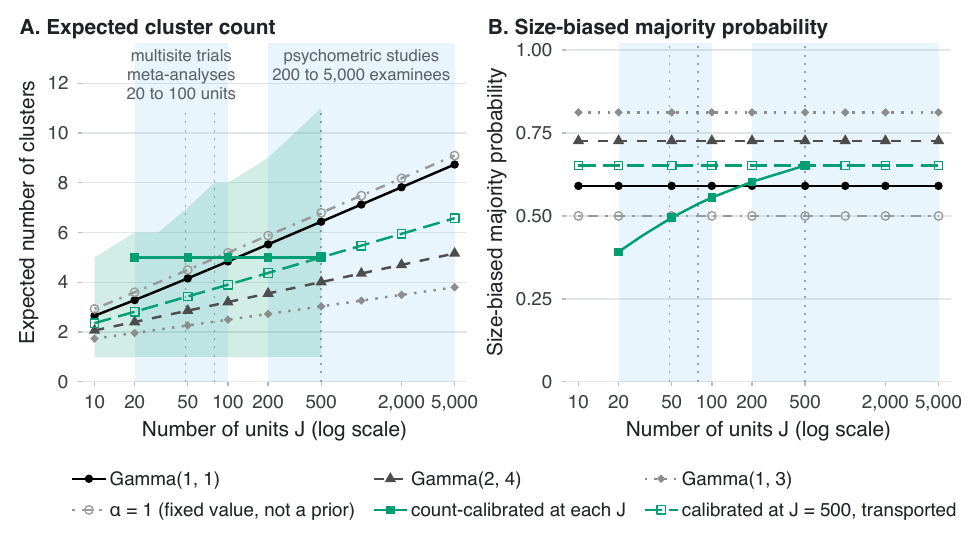}
  \caption{Cluster counts and population weights across sample sizes. Panel~B gives the probability that a randomly selected unit belongs to a component holding more than half the population mass. The two green curves use the same judgment about the cluster count (about five, with an 80\% interval from two to ten). We calibrate the prior at each sample size for the solid curve and once at $J=500$ for the dashed curve. \textit{Note.} In Panel~A, the green band spans the 5th to 95th percentiles of the cluster count under the prior calibrated at $J=500$. Dotted lines mark the three applications in Section~\ref{sec:applications}.}
  \label{fig:defaults}
\end{figure}


\section{Design-Conditional Elicitation}
\label{sec:method}

We ask analysts to state how many clusters they expect among their sites, schools, studies or examinees, and how uncertain they are. We use this statement to calibrate $\alpha \sim \mathrm{Gamma}(a, b)$. Because the same hyperprior implies different counts at different $J$, we condition the calibration on the design size. We call this procedure design-conditional elicitation (DCE). It matches the stated mean and spread of $K_J$, then checks the implied population weights (Section~\ref{subsec:unintended}).

\subsection{Overview}
\label{subsec:overview}

Figure~\ref{fig:2} follows one statement through four stages: state, calibrate, diagnose and decide. The analyst specifies $J$, an expected count $\mu_K = \E(K_J)$ and an uncertainty judgment, which we convert to $\sigma_K^2 = \Var(K_J)$ (Section~\ref{subsec:elicit}). Two-stage moment matching (TSMM, Section~\ref{subsec:tsmm}) obtains $(a, b)$ with these count moments at this $J$. We then calculate the distribution of $K_J$ and the weight summaries in Table~\ref{tab:functionals}. If the size-biased majority probability exceeds a threshold fixed in advance, we apply the Dual-Anchor refinement of Section~\ref{subsec:dual-anchor}. Algorithm~\ref{alg:tsmm} specifies the inputs and outputs.

We fix five policy choices before fitting: the size-biased threshold $t = 0.5$, a trigger of 0.40, a soft target of 0.25, a grid of Dual-Anchor tuning values $\lambda$, and a separate hard bound of 0.25 for sensitivity analysis. Section~\ref{subsec:dual-anchor} defines their roles. We report failed calibrations explicitly.

\begin{figure}[t]
  \centering
  \includegraphics{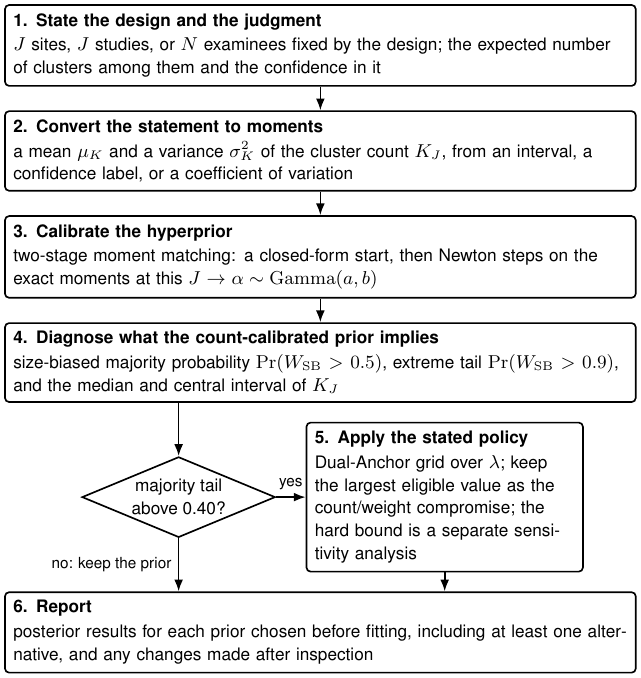}
  \caption{Design-conditional elicitation in the six steps of Box~\ref{box:recipe}. We convert the analyst's statement about $J$ units to target count moments and calibrate a Gamma hyperprior using two-stage moment matching. We then compare the size-biased diagnostic with the 0.40 trigger and either retain the prior or apply Dual-Anchor refinement. The 0.25 hard bound is evaluated separately as a sensitivity analysis.}
  \label{fig:2}
\end{figure}

\begin{algorithm}[t]
  \caption{Design-conditional elicitation}
  \label{alg:tsmm}
  \begin{algorithmic}[1]
    \Require the design size $J$; the target mean $\mu_K$ and variance $\sigma_K^2$ of $K_J$ from Section~\ref{subsec:elicit}; the policy constants $t = 0.5$, trigger $0.40$, soft target $0.25$, the fixed $\lambda$ grid, and hard bound $0.25$.
    \Ensure a Gamma hyperprior $(a^\star, b^\star)$ for $\alpha$ with its prior-predictive summaries, or a report that the calibration did not pass its checks.
    \State \textsc{State.} Record $J$, the statement about $K_J$, the route by which it was converted, and $(\mu_K, \sigma_K^2)$.
    \State \textsc{Calibrate.} Find $(a_K, b_K)$ by two-stage moment matching. Start from the large-$J$ approximation and match the target mean and variance at this $J$ to within numerical tolerances (Section~\ref{subsec:tsmm}). OSM Appendices~\ref{app:stage1} and~\ref{app:stage2} describe the starting values, the solver and its checks.
    \State \textsc{Diagnose.} Recompute the two moments at a higher quadrature order, and report the median and central interval of $K_J$, $\Pr(W_{\mathrm{SB}} > 0.5)$, $\Pr(W_{\mathrm{SB}} > 0.9)$, the bounds on $\Pr(W_{\max} > t)$ from Proposition~\ref{prop:largest-weight}, and optionally $\E(\rho)$.
    \State \textsc{Decide.} If $\Pr(W_{\mathrm{SB}} > 0.5) \le 0.40$, keep $(a_K, b_K)$; otherwise apply Algorithm~\ref{alg:dual-anchor}. A calibration that does not pass its checks at any stage is reported as such and is never replaced by an approximate one.
  \end{algorithmic}
\end{algorithm}

\subsection{Eliciting the Inputs in Practice}
\label{subsec:elicit}

We ask analysts to specify their expected number of clusters, $\mu_K$, and their uncertainty about the count. Such judgments may be difficult and may differ across analysts. We offer three ways to state the uncertainty: a central interval, a qualitative confidence label, or a coefficient of variation. We convert an interval either with the normal proxy below or with a discrete maximum-entropy distribution; confidence labels and coefficients of variation map directly to count moments. \texttt{DPprior} implements the maximum-entropy and direct-moment routes, and the accompanying analysis code implements the normal proxy.

We first specify the effect or latent trait of interest, its scale and the likelihood that relates it to the observed data. We then ask how many clusters to expect among the $J$ units. A plausible count may differ across studies because the effects of interest or the numbers of units differ. The clusters need not correspond to substantive groups. We recommend documenting what informed the judgment and whether it was made before or after examining the data. We supplied the count judgments for the three applications to illustrate the method. OSM Section~\ref{subsec:E6-guidelines} describes their basis and what is known about their timing.

An interval statement specifies that $K_J$ lies in $[k_L, k_U]$ with probability $q$. The normal approximation gives $\sigma_K \approx (k_U - k_L)/(2 z_{(1+q)/2})$. For $\mu_K$ we use the stated mean if available, or the midpoint otherwise. Alternatively, we use the discrete maximum-entropy distribution on $\{1, \ldots, J\}$ with the stated interval mass. Either conversion specifies moment targets. After calibration we check how much probability the prior assigns to the interval, since this can differ from $q$. In each of the three applications, we specified an 80 percent probability for the stated count interval. Under the count-calibrated priors, the corresponding probabilities ranged from 85 to 93 percent (OSM Table~\ref{tab:E-elicitation-worksheet}). For $J = 100$, ``between 3 and 10 clusters with 80 percent probability'' gives $\mu_K = 6.5$ and $\sigma_K^2 \approx 7.5$ under the normal proxy. TSMM returns $\alpha \sim \mathrm{Gamma}(6.82, 4.82)$, with $\Pr(W_{\mathrm{SB}} > 0.5) = 0.40$ and a probability of $0.07$ that a unit belongs to a component holding more than 90 percent of the mass.

At the stated mean, a small target variance for the count can be unattainable. Interpreting the same interval as a 95 percent statement gives $\sigma_K^2 \approx 3.2$, for which the calibration finds no solution. We can see why by holding the concentration parameter fixed. If $\alpha^*$ solves $\E(K_J \mid \alpha^*) = \mu_K$, every prior on $\alpha$ with expected count $\mu_K$ satisfies $\Var(K_J) \ge \Var(K_J \mid \alpha^*)$ (OSM Proposition~\ref{prop:var-lower-bound}). Here $\alpha^* = 1.38597$ and the minimum variance is $4.5224$. Only a point mass at $\alpha^*$ attains the minimum, so every Gamma hyperprior with this expected count has a larger variance. Equation~\eqref{eq:antoniak} gives both the required concentration and the bound before calibration. The software also rejects fits that miss the requested moments. Such conflicts usually arise from tight intervals around means small relative to $J$. With the expected count held fixed, widening the interval or lowering its stated probability sufficiently can give a variance above the bound. The change requires a substantive justification, and the revised moments still have to be checked by calibration.

For a qualitative confidence judgment, we use a variance inflation factor relative to the shifted Poisson proxy of Section~\ref{subsec:tsmm}: $\sigma_K^2 = \mathrm{VIF}\,(\mu_K - 1)$. A larger VIF represents greater uncertainty about the count, without changing its stated mean. As working defaults, we use variance inflation factors of $1.5$, $2.5$ and $5$ for high, medium and low confidence, respectively. At $\mu_K = 5$, these imply $\sigma_K^2 = 6$, $10$ and $20$, or standard deviations of about 2.4, 3.2 and 4.5 clusters. Analysts who can specify relative uncertainty may instead supply $\mathrm{CV}_K = \sigma_K/\mu_K$, giving $\sigma_K^2 = (\mathrm{CV}_K\,\mu_K)^2$.

In the simulations, posterior counts are sensitive to the stated mean (Section~\ref{subsec:sim-count-sensitivity}). At a fixed mean, the prior calculations below show that the stated variance also changes the probability of extreme dominance. Table~\ref{tab:vif} holds $J = 100$ and $\mu_K = 5$ fixed while varying VIF from 1.25 to 5, with $\mathrm{Gamma}(1, 1)$ for comparison. The majority tail changes from 0.53 to 0.61, the extreme tail more than doubles, and the mass below $\alpha = 0.1$ grows from nil to 0.18. In this example, allowing more uncertainty about the count makes it more likely that a randomly chosen unit belongs to a component holding a majority of the population mass. Both count inputs therefore require justification and sensitivity analysis (Section~\ref{subsec:sim-count-sensitivity}). We also fix the soft target $\delta$, threshold $t$, trigger and rule for selecting $\lambda$ before fitting, and report them with the full frontier described in Section~\ref{subsec:dual-anchor}.

\begin{table}[t]
\centering
\caption{Hyperpriors calibrated at $J = 100$ to a mean count of five under four variance inflation factors, compared with the $\mathrm{Gamma}(1, 1)$ default}
\label{tab:vif}
\begingroup
\normalsize
\setstretch{1.08}
\renewcommand{\arraystretch}{1.18}
\setlength{\tabcolsep}{9pt}
\noindent\textit{Panel A. Count calibration}\par\smallskip
\begin{tabularx}{\textwidth}{@{}>{\raggedright\arraybackslash}X c c c@{}}
\toprule
Setting & $\sigma_K^2$ & $(a, b)$ & $\E(\alpha)$ \\
\midrule
VIF $=1.25$ & 5  & (7.13, 7.41) & 0.96 \\ \addlinespace[3pt]
VIF $=1.5$  & 6  & (4.35, 4.47) & 0.97 \\ \addlinespace[3pt]
VIF $=2.5$  & 10 & (1.67, 1.65) & 1.01 \\ \addlinespace[3pt]
VIF $=5$    & 20 & (0.62, 0.56) & 1.11 \\ \addlinespace[3pt]
$\mathrm{Gamma}(1, 1)$ default & {--} & (1, 1) & 1.00 \\
\bottomrule
\end{tabularx}

\medskip
\noindent\textit{Panel B. Weight and concentration diagnostics}\par\smallskip
\begin{tabularx}{\textwidth}{@{}>{\raggedright\arraybackslash}X c c c@{}}
\toprule
Setting & $\Pr(W_{\mathrm{SB}} > 0.5)$ & $\Pr(W_{\mathrm{SB}} > 0.9)$ & $\Pr(\alpha < 0.1)$ \\
\midrule
VIF $=1.25$ & 0.53 & 0.14 & 0.000 \\ \addlinespace[3pt]
VIF $=1.5$  & 0.53 & 0.16 & 0.001 \\ \addlinespace[3pt]
VIF $=2.5$  & 0.56 & 0.23 & 0.030 \\ \addlinespace[3pt]
VIF $=5$    & 0.61 & 0.36 & 0.182 \\ \addlinespace[3pt]
$\mathrm{Gamma}(1, 1)$ default & 0.59 & 0.30 & 0.095 \\
\bottomrule
\end{tabularx}
\endgroup

\medskip
\noindent\parbox{\textwidth}{\small\textit{Note.} Each VIF row is the exact-moment TSMM solution for the stated variance. We computed the diagnostics using the unrounded $(a,b)$. At $J = 100$ the $\mathrm{Gamma}(1,1)$ default gives a mean cluster count of 4.84, close to the calibration target of five. Its size-biased tail probabilities at 0.5 and 0.9 are 0.59 and 0.30, respectively, both within the ranges for the four calibrated priors.}
\end{table}

\subsection{Two-Stage Moment Matching}
\label{subsec:tsmm}

There is no closed-form solution for the Gamma parameters that match a specified mean and variance of $K_J$. We obtain starting values in closed form and refine them to match the moments at the study's sample size. This makes it practical to calibrate several priors for sensitivity analysis before fitting the model. In Table~\ref{tab:vif}, for example, we held the expected count at five and calibrated a separate prior for each of four target variances. OSM Section~\ref{subsec:G2-overview} reports calibration times for a worked example.

Large-$J$ Poisson approximations for $K_J$ motivate the initializer \citep{arratia2000poisson, vicentini2025prior}. We use the shifted proxy $K_J - 1 \mid \alpha \approx \Poisson(\alpha \log J)$ so that it respects $K_J \ge 1$. Gamma mixing gives a negative binomial proxy. Matching its mean $\mu_0 = \mu_K - 1$ and variance $\sigma_K^2$ yields $a_0 = \mu_0^2/(\sigma_K^2 - \mu_0)$ and $b_0 = \mu_0 \log J/(\sigma_K^2 - \mu_0)$, provided $\mu_0 > 0$ and $\sigma_K^2 > \mu_0$. OSM Appendix~\ref{app:stage1} gives the derivation and the projection used near this boundary.

At moderate $J$, this initializer is biased. The conditional count is a sum of independent Bernoulli indicators and is underdispersed relative to the Poisson proxy. Moreover, $\E(K_J \mid \alpha) \approx 1 + \alpha \log J$ omits a term of order $\alpha$. In Stage~2 we use damped Newton iteration to solve $\E_{a,b}(K_J) = \mu_K$ and $\Var_{a,b}(K_J) = \sigma_K^2$ in $(\log a, \log b)$. We evaluate the Gamma expectations by Gauss--Laguerre quadrature. Score-function identities give the Newton Jacobian (OSM Appendix~\ref{app:stage2}). If the residuals exceed the tolerances after Newton iteration, the solver tries Nelder--Mead.

The numerical refinement is needed throughout the regimes studied here; OSM Appendix~\ref{app:error-bounds} examines when it can be omitted. We check the resulting moments at a higher quadrature order. If the requested moments are not attained or the recomputation disagrees, we report calibration failure.

\subsection{Dual-Anchor Protocol}
\label{subsec:dual-anchor}

The TSMM solution $(a_K, b_K)$ is the \emph{count-calibrated prior} for the \emph{stated count judgment}. We check its weight implications using $\Pr(W_{\mathrm{SB}} > 0.5)$ and $\Pr(W_{\mathrm{SB}} > 0.9)$ from Equation~\eqref{eq:w1-tail}, the bounds on $\Pr(W_{\max} > t)$ in Proposition~\ref{prop:largest-weight}, and optionally $\E(\rho)$, a one-dimensional Gamma integral (OSM Appendix~\ref{app:weights}). Our rule, fixed in advance, applies soft Dual-Anchor refinement when the majority probability exceeds 0.40. We use 0.25 as the soft target and as the bound in a separate hard analysis. We chose these probabilities for illustration. For a particular application, both choices require justification before fitting. Section~\ref{subsec:worked-example} discusses how to choose the trigger with limited prior information.

A small change in $(a,b)$ can lower the majority probability while preserving both the mean and the variance of the count only if these moments fail to identify $(a,b)$ locally. For instance, when $J=2$, $K_2-1$ is Bernoulli and its mean determines its variance. We use fixed scales to measure how far a Dual-Anchor solution departs from the count and weight targets. With $s_\mu = \max(|\mu_K|, 1)$ and $s_v = \max(\sigma_K^2, 1)$, let
\begin{equation}\label{eq:dual-losses}
\begin{aligned}
  D_K(a,b) &=
  \left\{\frac{\E_{a,b}(K_J)-\mu_K}{s_\mu}\right\}^{\!2}
  {}+\left\{\frac{\Var_{a,b}(K_J)-\sigma_K^2}{s_v}\right\}^{\!2},\\
  D_w(a,b) &= \{\Pr(W_{\mathrm{SB}}>t\mid a,b)-\delta\}^2 .
\end{aligned}
\end{equation}
The refined hyperprior minimizes
\begin{equation}\label{eq:dual-objective}
  \mathcal{L}_\lambda(a,b)=\lambda D_K(a,b)+(1-\lambda)D_w(a,b),
  \qquad \lambda \in (0, 1],
\end{equation}
where $\lambda$ controls the relative weight given to the count discrepancy. It neither imposes a constraint nor gives a probability or guarantee that the target is attained.

Before posterior simulation, we fix the grid
\[
  \Lambda = \{0.01, \ldots, 0.20, 0.25, 0.30, 0.40, \ldots, 0.90, 1.00\}.
\] If the count-calibrated tail exceeds 0.40, we minimize Equation~\eqref{eq:dual-objective} at each grid value and verify each solution at a higher quadrature order. Among finite, interior solutions with $\Pr(W_{\mathrm{SB}} > 0.5) \le 0.40$, we select the largest $\lambda$. This gives the greatest weight to the count moments among eligible grid points. We call the resulting Dual-Anchor prior the \emph{count/weight compromise}.

If no grid point is eligible, we report the unresolved conflict and fit no model. The soft target $\delta = 0.25$ guides the optimization, whereas 0.40 determines eligibility. The selected Dual-Anchor compromise can therefore satisfy the trigger without reaching the soft target.

We evaluate a strict bound separately by minimizing $D_K$ subject to $\Pr(W_{\mathrm{SB}} > t) \le \delta$ (\texttt{DPprior\_dual\_hard()}). This is the \emph{weight-bound sensitivity}. We record whether the constraint is active, the solution lies on the solver boundary, or the hyperprior approaches a point mass. This constraint applies only to the prior size-biased tail. It does not constrain posterior dominance or impose the same bound on $\Pr(W_{\max}>t)$. Proposition~\ref{prop:largest-weight} gives the corresponding bounds for the ranked weights. These results supplement the Dual-Anchor compromise. Algorithm~\ref{alg:dual-anchor} gives the full protocol, and OSM Appendix~\ref{app:software} plots the frontier over $\Lambda$ for the worked example. We recommend reporting this frontier.

\begin{algorithm}[t]
  \caption{Dual-Anchor refinement}
  \label{alg:dual-anchor}
  \begin{algorithmic}[1]
    \Require the count-calibrated prior $(a_K, b_K)$ and its targets $(\mu_K, \sigma_K^2)$; $t = 0.5$; soft target $\delta = 0.25$; trigger $\eta = 0.40$; the fixed grid $\Lambda$.
    \Ensure the Dual-Anchor count/weight compromise, or a report that no grid value passed the checks; the weight-bound sensitivity.
    \State If $\Pr(W_{\mathrm{SB}} > t \mid a_K, b_K) \le \eta$, keep the count-calibrated prior.
    \State Otherwise, for every $\lambda \in \Lambda$, minimize Equation~\eqref{eq:dual-objective} in $(\log a, \log b)$ and recompute the moments and the tail probability of the solution at a higher quadrature order.
    \State Discard solutions that are not finite, whose moments at the selection and verification quadrature orders disagree beyond tolerance, or that lie on the solver boundary or near a point mass. Among the remaining solutions with $\Pr(W_{\mathrm{SB}} > t) \le \eta$, select the largest $\lambda$; if there is none, report that the calibration did not pass the checks and fit no model.
    \State Separately solve $\min D_K(a, b)$ subject to $\Pr(W_{\mathrm{SB}} > t) \le \delta$, and report whether the constraint is active and whether the solution lies on a boundary.
    \State Report $(a^\star, b^\star)$, $\E(K_J)$, $\Var(K_J)$, both size-biased tail probabilities, the ranked-weight bounds, $\lambda$, and the four status fields recorded by the software (convergence, whether the solution is interior, agreement of the independent recomputation, and the boundary classification).
  \end{algorithmic}
\end{algorithm}

\subsection{Worked Example, Software, and Reporting}
\label{subsec:worked-example}

Suppose that an analyst expects ``about five clusters, with moderate uncertainty'' among $J = 50$ units. We set $\mu_K = 5$ and use the medium confidence label to obtain $\sigma_K^2 = 10$. Stage~1 gives $(a_0, b_0) \approx (2.67, 2.61)$. Stage~2 gives $\mathrm{Gamma}(1.4082, 1.0770)$, which reproduces the count moments (5.000 and 10.000), with median 4 and a 90 percent interval of $[1, 11]$ for $K_{50}$. Following \citet{vicentini2025prior}, we examine what the count-calibrated prior implies for the population weights. We report both tail probabilities whether or not we refine the prior. Here, the probabilities that a randomly chosen unit belongs to a component holding more than half or more than 90 percent of the population mass are 0.497 and 0.200, respectively.

Because 0.497 exceeds our 0.40 trigger, we compare the refined priors on the prespecified grid. Among the verified solutions that bring the majority probability to 0.40 or below, we choose the one that gives the most weight to the original count moments. This gives $\lambda = 0.30$. The resulting Dual-Anchor solution is the count/weight compromise $\mathrm{Gamma}(2.3158, 1.4204)$. It gives $\E(K_{50}) = 5.884$, $\Var(K_{50}) = 9.738$ and tails 0.398 and 0.107. Reducing the majority probability below the trigger thus adds about one expected cluster, while leaving the 0.25 soft target unmet. The weight-bound sensitivity, $\mathrm{Gamma}(5.2745, 2.3059)$, gives $\E(K_{50}) = 7.488$, $\Var(K_{50}) = 9.051$ and majority probability 0.250, with the constraint active. It shows the further count discrepancy required by the strict bound. OSM Appendix~\ref{app:software} reports all 29 grid values.

The following calls in \texttt{DPprior} (version 2.0.0) return the count-calibrated prior and the Dual-Anchor prior at the selected $\lambda$. The accompanying analysis code evaluates the grid, applies the selection rule, fits the hard variant and verifies the quadrature (OSM Appendix~\ref{app:software}). These operations, including selection of $\lambda$, are separate from the two calls. Box~\ref{box:recipe} summarizes the workflow.
{\small
\begin{verbatim}
fit  <- DPprior_fit(J = 50, mu_K = 5, var_K = 10, method = "A2-MN")
soft <- DPprior_dual_soft(fit, lambda = 0.30, target = list(
          metric = "wsb_tail", threshold = 0.5,
          relation = "target", value = 0.25))
\end{verbatim}}

\begin{recipebox}
\begin{minipage}{\textwidth}
\caption{Specifying and reporting a prior in six steps}
\label{box:recipe}
\end{minipage}
\begin{mdframed}
\begin{minipage}{\linewidth}
\normalsize\setstretch{1.10}
For illustration, we use 0.40 as the trigger and 0.25 as the target for the majority probability $\Pr(W_{\mathrm{SB}} > 0.5)$. Choose and justify both values for the application before fitting the model. The procedure need not attain the target.
\begin{enumerate}[label=(\arabic*),leftmargin=2em,itemsep=5pt,topsep=4pt]
\item Give the number of sites or studies $J$ (or $J=N$ examinees in an item response model) and explain why this set of units is fixed for the analysis. State the expected number of clusters and the uncertainty about this count. Record who or what informed the judgment and when it was made.
\item Convert the statement to a mean $\mu_K$ and a variance $\sigma_K^2$ of $K_J$ by one of the three routes of Section~\ref{subsec:elicit}, and record the route.
\item Calibrate the Gamma hyperprior with \texttt{DPprior\_fit} (\texttt{DPprior} 2.0.0), which returns $(a, b)$ matching the two moments at this $J$ or reports that the calibration did not pass its checks.
\item Check both weight tails, $\Pr(W_{\mathrm{SB}} > 0.5)$ and $\Pr(W_{\mathrm{SB}} > 0.9)$. For the count, report the achieved moments, median, central interval and probability assigned to the stated interval.
\item Keep the count-calibrated prior if the majority probability is at or below the chosen trigger. If it exceeds the trigger, use \texttt{DPprior\_dual\_soft} over the prespecified grid and select the largest eligible $\lambda$. In a separate sensitivity analysis, impose an upper bound on the majority probability and report how far the count moments depart from their targets.
\item Before fitting the model, choose at least one alternative prior for comparison. If several count judgments seem plausible, specify two or three. Identify any stress tests, report the posterior results for each prior and record any revisions made after inspecting those results.
\end{enumerate}
\end{minipage}
\end{mdframed}
\end{recipebox}

When there is little prior information about the count, three questions can guide the choice of prior. (a) Among these $J$ units, what are the smallest and largest counts that would be unsurprising? This range and the probability the analyst assigns to it express the uncertainty used in calibration. (b) How does the expected count under the calibrated prior compare with a single cluster and with the expected count under $\mathrm{Gamma}(1,1)$ at this $J$? A large departure from both calls for a reason to expect a different count in this study. (c) Would it be surprising for a randomly chosen unit to belong to a component holding more than half the population mass? The largest probability the analyst would accept for this event defines the trigger. We retain the count-calibrated prior if its probability is at or below the trigger. Otherwise, a lower probability the analyst considers acceptable serves as the soft target for refinement (Box~\ref{box:recipe}).

When plausible priors give materially different partitions, there is little basis for giving any one partition a substantive interpretation. The fitted distribution of effects or the estimate for a particular unit may nevertheless agree across priors. To show how sensitive the cluster count is, we recommend reporting the range across fits for both the posterior mean count and the posterior probability of a single cluster. A larger count alone does not indicate better recovery of heterogeneity.

OSM Section~\ref{subsec:E6-guidelines} gives a reporting worksheet, an example statement and the completed records for our three applications. In Section~\ref{sec:simulation}, we examine the consequences of these choices when the data provide little information about the partition.

\FloatBarrier

\section{Simulation Study}
\label{sec:simulation}

\subsection{Design}
\label{subsec:sim-design}

We examined how the concentration prior affected the inferred number of clusters and the accuracy with which we estimated the site effects and their distribution. In the main comparison we simulated 100 sites from $K^*=5$, 10, or 30 components. For each generating count we considered three levels of informativeness, $I=0.2$, 0.5, and 0.8. A larger $I$ indicates less sampling error relative to the variation in effects across sites. We generated 200 data sets for each of the nine designs and fitted the same mixture model to each data set under five priors.

We divided the sites as evenly as possible among components with equally spaced means. The true site effects varied around their respective component means. We rescaled the effects in each data set to give a sample standard deviation of 0.25, corresponding to a finite-site SD of 0.24875 with denominator $J$. For each site we generated a normally distributed estimate centered on its true effect, with sampling variance $4/n_j$ for a site of size $n_j$. We varied site sizes to target these levels of $I$ \citep{lee2025improving}. Since sampling error could obscure differences between adjacent component means, we used $K^*$ as a reference count when assessing the fitted partitions. OSM Section~\ref{subsec:F2-design} gives the generating parameters and the calculation of $I$.

We use \emph{Vague} to denote the fixed $\mathrm{Gamma}(1,1)$ prior. For \emph{TSMM}, we set the prior mean of the cluster count to the generating count, $\mu_K=K^*$, and used a variance inflation factor of 2.5. \emph{Dual-Anchor} used the TSMM prior when its majority probability was at or below 0.40. Above this threshold we adjusted the prior using Algorithm~\ref{alg:dual-anchor}. For \emph{SSI-Weight}, we specified two probabilities for the population weights without setting a target for the count (OSM Section~\ref{subsec:F3-hyperpriors}). We also fitted \emph{DORO-Unif}, which uses \citeauthor{dorazio2009selecting}'s \citeyearpar{dorazio2009selecting} calibration to a uniform count distribution. We treated it as an external reference because its discrepancy from that target exceeded our calibration tolerance.

For the partition, we recorded the posterior mean cluster count, the posterior probability of a single cluster (the \emph{collapse probability}), and the coverage and width of 95 percent posterior intervals for the count. We compared the posterior mean of the between-site SD with the SD of the 100 true effects. To assess recovery of the full distribution, we compared the empirical CDF of the constrained Bayes (CB) estimates \citep{louis1984estimating,ghosh1992constrained} with the empirical CDF of the true effects. The integrated squared error loss (ISEL) is the area under the squared difference between the two CDFs. We also compared matching percentiles of the 100 CB site estimates under $\mathrm{Gamma}(1,1)$ and Dual-Anchor. At each percentile, we averaged the absolute difference between the two estimates over data sets in which both fits converged. To assess recovery under each prior, we compared its estimated percentile with the corresponding percentile of the 100 true site effects in the same data set. For individual sites we calculated the root mean squared error (RMSE) of the posterior means. We report summaries for converged fits, with $n/N$ giving the number that converged out of those attempted (OSM Section~\ref{subsec:F4-mcmc} and Table~\ref{tab:F6-population}).

We also examined whether the inferred cluster count depended on the prior for the grand mean, the common center of the component means. The main comparison used a grand-mean prior with mean zero and variance 1. In each of the 45 conditions we refitted the same data with the variance set to 0.01 and to 100. In separate experiments we held the variance at 100 to examine the effects of understating or overstating the expected count, varying the number of sites, and assigning most sites to one component (Section~\ref{subsec:sim-count-sensitivity}).

\subsection{Main Results}
\label{subsec:sim-main}

At low informativeness ($I=0.2$), sampling error obscured differences between component means, and the inferred cluster count depended strongly on the prior. Under $\mathrm{Gamma}(1,1)$ the posterior mean count was about 1.5 whether the data came from 5, 10, or 30 components (Table~\ref{tab:sim-results}). In the 30-component design, setting the TSMM prior mean count to 30 gave a posterior mean of about 25 clusters, although the 95 percent posterior intervals were wide. OSM Figure~\ref{fig:F-cluster-counts} gives the estimates at all three levels of informativeness.

\begin{table}[!tp]
\centering
\caption{Posterior cluster-count summaries at low informativeness ($I=0.2$) under five hyperpriors ($J=100$, $N(0,1)$ grand-mean prior)}
\label{tab:sim-results}
\begingroup
\normalsize
\setstretch{1.05}
\renewcommand{\arraystretch}{1.22}
\setlength{\tabcolsep}{4pt}
\begin{tabularx}{\textwidth}{@{}c >{\raggedright\arraybackslash}X ccccc@{}}
\toprule
$K^*$ & Hyperprior & \makecell{Posterior mean\\of $K_J$ (SD)} & \makecell{$\Pr(K_J{=}1$\\$\mid\,\text{data})$} & Coverage & \makecell{Mean\\width} & $n/N$ \\
\midrule
5 & Vague & 1.50 (0.11) & 0.725 & 0.535 & 3.60 & 200/200 \\
 & DORO-Unif$^{\dagger}$ & 1.94 (0.17) & 0.536 & 1.000 & 5.07 & 200/200 \\
 & TSMM & 1.75 (0.14) & 0.616 & 0.965 & 4.57 & 200/200 \\
 & Dual-Anchor & 3.00 (0.27) & 0.315 & 1.000 & 8.29 & 200/200 \\
 & SSI-Weight & 3.41 (0.37) & 0.383 & 1.000 & 12.79 & 200/200 \\
\addlinespace
10 & Vague & 1.51 (0.12) & 0.723 & 0.000 & 3.67 & 200/200 \\
 & DORO-Unif$^{\dagger}$ & 2.60 (0.31) & 0.516 & 0.980 & 10.85 & 200/200 \\
 & TSMM & 4.26 (0.40) & 0.269 & 1.000 & 13.70 & 200/200 \\
 & Dual-Anchor & 4.26 (0.39) & 0.269 & 1.000 & 13.68 & 200/200 \\
 & SSI-Weight & 3.42 (0.38) & 0.382 & 1.000 & 12.83 & 200/200 \\
\addlinespace
30 & Vague & 1.51 (0.13) & 0.723 & 0.000 & 3.70 & 200/200 \\
 & DORO-Unif$^{\dagger}$ & 3.57 (0.57) & 0.600 & 0.097 & 24.89 & 195/200 \\
 & TSMM & 25.17 (0.32) & 0.001 & 1.000 & 32.94 & 167/200 \\
 & Dual-Anchor & 25.18 (0.31) & 0.001 & 1.000 & 32.92 & 169/200 \\
 & SSI-Weight & 3.39 (0.41) & 0.386 & 0.000 & 12.79 & 200/200 \\
\bottomrule
\end{tabularx}
\endgroup

\medskip
\noindent\parbox{\textwidth}{\small \textit{Note.} All summaries use the converged fits. For each cell, $n/N$ gives their number out of the 200 replications. Parentheses give the between-replication SD of the posterior mean of $K_J$. Coverage is the proportion of 95 percent posterior intervals containing the generating count $K^*$. Mean width is the average difference between their upper and lower limits. OSM Table~\ref{tab:F-main-full} gives bias and RMSE. $^\dagger$DORO-Unif is an external reference.}
\end{table}

With five generating components, the majority probability under TSMM exceeded 0.40, so Dual-Anchor adjusted the prior. For the same 200 low-information data sets, the adjustment raised the posterior mean count from 1.75 to 3.00 and lowered the posterior probability of a single cluster from 62 to 32 percent. The estimated between-site SD was about 0.30 under both priors, compared with the true value of 0.249 (Figure~\ref{fig:3}A). Mean CDF ISEL and mean site RMSE were also similar under the two priors (OSM Table~\ref{tab:F-finite-site-spread}). With 10 or 30 generating components, the majority probability under TSMM was below the 0.40 threshold, so the two methods used the same prior for $\alpha$.

The larger gains in accuracy came from more precise site estimates. At each of the three generating counts, increasing informativeness reduced the errors in both the estimated distribution of the 100 site effects and the individual site estimates under all five priors. From $I=0.2$ to $I=0.8$, the mean integrated squared error of the estimated CDF fell by 66--81 percent and the mean RMSE of the individual site estimates fell by about half.

When site estimates were imprecise, the relative differences across priors were larger for mean CB ISEL than for mean site RMSE (Figure~\ref{fig:3}B). At $I=0.2$ the largest mean CB ISEL exceeded the smallest by 2.30, 4.31 and 6.88 percent for $K^*=5$, 10 and 30, respectively. For mean site RMSE, the largest value exceeded the smallest by 0.56, 0.71 and 1.15 percent for the same three generating counts. For the maximum vertical error of the CDF, the relative range did not consistently exceed that for site RMSE (OSM Section~\ref{subsec:F-finite-site-spread}).

\needspace{6\baselineskip}
Which parts of the estimated distribution changed when we switched from $\mathrm{Gamma}(1,1)$ to Dual-Anchor? The 99th percentile moved more than the median on average in all nine designs (Figure~\ref{fig:3}C). Across the three generating counts at $I=0.2$, the mean absolute change was 0.025--0.037 at the 99th percentile and 0.002--0.004 at the median. At $I=0.8$ the mean absolute change at the 99th percentile was only 0.003--0.004. At $I=0.2$ and $0.5$, both priors overestimated the 99th percentile of the 100 true site effects on average. Dual-Anchor gave the larger RMSE for this percentile in all six designs (OSM Section~\ref{subsubsec:F-cb-percentiles}).

\begin{figure}[!tp]
  \centering
  \includegraphics[width=\textwidth,height=0.77\textheight,keepaspectratio]{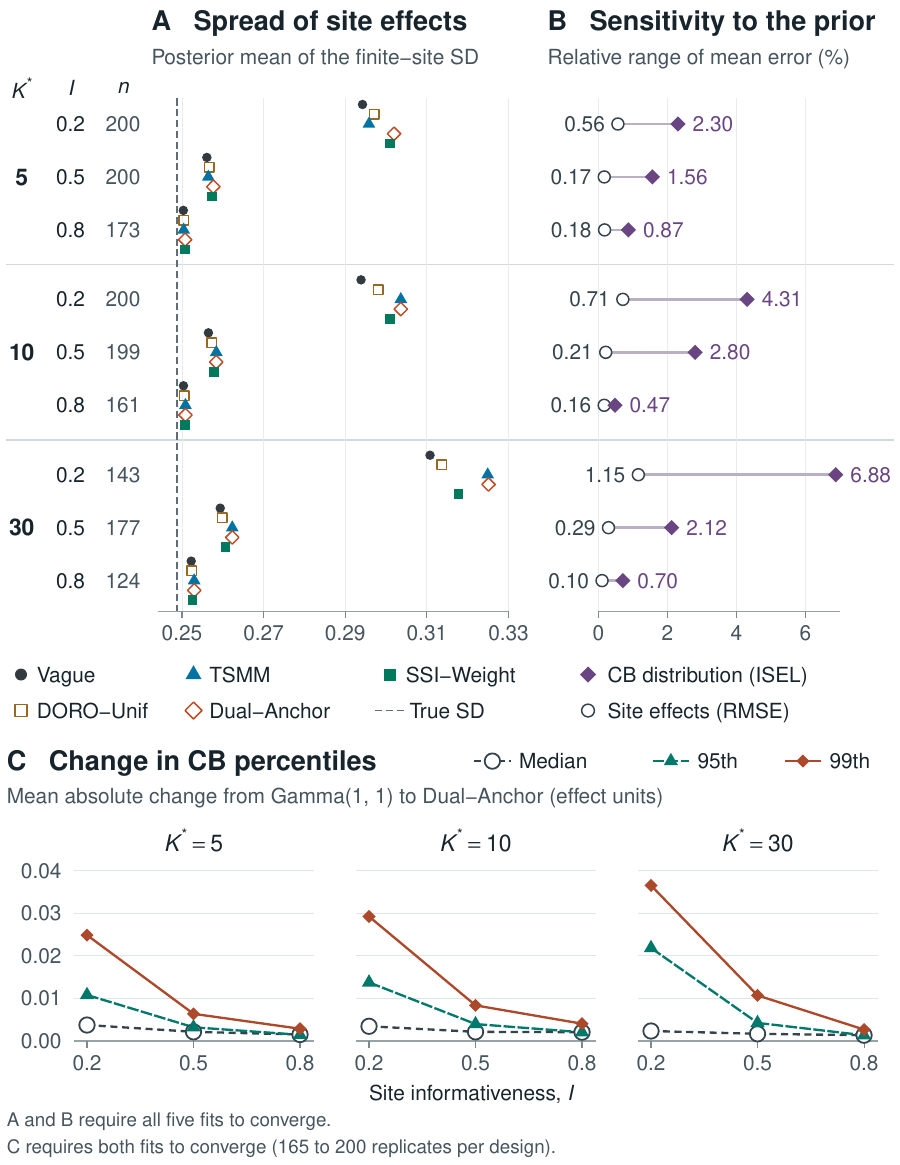}
  \caption{Sensitivity of site-effect estimates to the concentration prior. A compares the estimated SD across the 100 site effects under five priors with the true SD. B shows the percentage increase from the smallest to the largest mean error across those priors, separately for the estimated CDF and for individual site effects. C shows the mean absolute change in the median and upper percentiles of the constrained Bayes site estimates when switching from $\mathrm{Gamma}(1,1)$, labeled Vague in A, to Dual-Anchor.}
  \label{fig:3}
\end{figure}

\subsection{Sensitivity of the Posterior Cluster Count}
\label{subsec:sim-count-sensitivity}

The inferred number of clusters depended both on the prior for the grand mean and on the number of clusters we expected before seeing the data. When we fitted the same data with a wider prior for the grand mean, the posterior mean cluster count fell in all 45 conditions relative to the tightest prior, although the RMSE of the individual site estimates changed little on average (OSM Section~\ref{subsec:F12-original}). In a separate comparison with ten generating components and $I=0.2$, increasing the expected count used for TSMM from five to twenty raised the posterior mean count from 1.51 to 15.03 (OSM Table~\ref{tab:F11-A}).

The posterior counts continued to differ across concentration priors as the number of sites increased from 25 to 200. The comparison at 200 sites is less certain at $I=0.8$, where only 30 to 93 of the 200 fits converged per condition (OSM Figure~\ref{fig:6}). The location of a large component also mattered. When one component contained 60 of the 100 sites, placing its mean at the center rather than varying its position across data sets made little difference to the posterior count at $I=0.2$. At $I=0.8$, however, central placement gave fewer clusters on average and a higher posterior probability of a single cluster (OSM Table~\ref{tab:F11-D}).

\subsection{Practical Implications}
\label{subsec:sim-summary}

Inferences about the partition can be sensitive to the concentration prior even when estimates of between-site variation are similar. In the five-component design with imprecise site estimates, Dual-Anchor gave a posterior mean count closer to five than TSMM did. The two priors gave nearly the same estimated between-site SD and similar RMSEs for the individual site effects. In an application the generating count is unknown. Before interpreting a partition, we would specify which counts are plausible and examine the population weights implied by each calibration. We would then compare posterior estimates of the count and the largest population weight across priors.

When the site estimates were imprecise, the distribution and the individual effects differed in their sensitivity to the prior. At all three generating counts, the relative range across priors was larger for the mean integrated squared error of the estimated CDF than for the mean RMSE of the individual site estimates. For both measures, the reduction in error with more precise site estimates was larger than the differences across priors. We observed this improvement under every prior we considered.

If the question concerns a particular school, we would compare estimates of that school's effect under plausible priors. If it concerns how large effects can be among the schools studied, we would compare an upper percentile of their effect distribution. In the low-information simulations, comparing the medians alone would have missed the larger shift at the 99th percentile (Figure~\ref{fig:3}C). In Section~\ref{sec:applications} we compare estimates across priors on the effect or ability scale, since the true values are unknown.

\section{Applications}
\label{sec:applications}

We consider 79 schools in the Tennessee STAR class-size experiment \citep{krueger1999experimental,nye2000effects}, 48 studies of writing-to-learn interventions \citep{bangertdrowns2004effects}, and 500 respondents to a vocabulary checklist \citep{openpsychometrics2016gcbs,domingue2025irw}. We address two questions in each application. First, we examine how the choice of concentration prior affects the prior and posterior distributions of the cluster count and the population component weights. Second, we assess how much replacing $\mathrm{Gamma}(1,1)$ with the selected Dual-Anchor prior changes the estimated distribution of latent effects among the observed units and the estimated effect for each unit.

In all three applications, we use $\mathrm{Gamma}(1,1)$ as the reference and compare it with the TSMM and Dual-Anchor priors calibrated for that application. As external reference priors, we use DORO-Unif for the schools and the meta-analysis and $\mathrm{Gamma}(1,3)$ for the vocabulary checklist. In each application we also fit the hard-bound version of Dual-Anchor as a sensitivity check.

Equation~\eqref{eq:model} relates the estimated small-class effects for the schools and the standardized mean differences for the studies to their respective latent effects using the corresponding standard errors. In both applications we model the latent effects with a DPM and fit a normal random-effects model as a benchmark. The DPM uses a standard normal prior for the grand mean, a uniform prior on $(0,1)$ for the variance of the component means in the base measure, and separate inverse-gamma priors for the within-component variances. We widen the first two priors together by increasing the normal variance to 100 and the upper limit of the uniform prior to 10. In STAR, we also examine sensitivity to the first-stage analysis (Section~\ref{subsec:app-star}).

For the vocabulary checklist, we fit a Rasch model to each respondent's binary item responses and use a DPM for the population distribution of ability. As a benchmark, we fit a Rasch model with normally distributed abilities (Section~\ref{subsec:app-irt}). We fit the DPMs by Markov chain Monte Carlo in NIMBLE, using \texttt{DPMirt} for the Rasch analysis, and check convergence across four chains for each reported fit (OSM Sections~\ref{subsec:H4-mcmc} and~\ref{subsec:H-irt}). Table~\ref{tab:applications} gives prior and posterior summaries of cluster counts and weights under the three main priors, along with estimates of the SD of effects among the observed units. Figure~\ref{fig:app-effect-changes}A compares $\mathrm{Gamma}(1,1)$ with Dual-Anchor in two ways. We calculate separate root mean square differences for posterior means matched by unit and for effect values matched by percentile of the posterior mean CDF. The curves in Figure~\ref{fig:app-effect-changes}B show where Dual-Anchor gives higher or lower estimates across that distribution. OSM Tables~\ref{tab:H-all-priors} and~\ref{tab:H-all-posteriors} include the external references and hard-bound sensitivity analyses.

\begin{table}[!htbp]
\centering
\caption{Cluster counts, weights, and effect SD under three priors}
\label{tab:applications}
\smallskip
\begingroup
\small
\renewcommand{\arraystretch}{1.12}
\setlength{\tabcolsep}{2pt}
\begin{tabular*}{\linewidth}{@{\extracolsep{\fill}}lcccccc@{}}
\toprule
 & \multicolumn{2}{c}{Prior implications} & \multicolumn{4}{c}{Posterior estimates}\\
\cmidrule(lr){2-3}\cmidrule(l){4-7}
Prior on $\alpha$ & Mean & Majority & Mean & One cluster & Largest & Effect\\
 & count & prob. (\%) & count & prob. (\%) & weight (\%) & SD\\
\midrule
\multicolumn{7}{@{}l}{\textbf{Project STAR} (79 schools, stated 80\% count range 2--8)} \\
$\mathrm{Gamma}(1,1)$ & 4.61 & 59.1 & 1.96 & 56.0 & 91.7 & 0.226 \\
TSMM & 5.00 & 51.0 & 3.12 & 20.0 & 81.8 & 0.234 \\
Dual-Anchor & 6.24 & 38.5 & 4.47 & 6.4 & 71.5 & 0.236 \\
\addlinespace[5pt]
\multicolumn{7}{@{}l}{\textbf{Writing-to-learn} (48 studies, stated 80\% count range 1--8)} \\
$\mathrm{Gamma}(1,1)$ & 4.12 & 59.1 & 1.89 & 54.8 & 90.4 & 0.252 \\
TSMM & 4.00 & 59.7 & 1.96 & 52.5 & 90.1 & 0.250 \\
Dual-Anchor & 5.66 & 39.5 & 3.31 & 19.0 & 77.3 & 0.256 \\
\addlinespace[5pt]
\multicolumn{7}{@{}l}{\textbf{Vocabulary checklist} (500 respondents, stated 80\% count range 2--10)} \\
$\mathrm{Gamma}(1,1)$ & 6.44 & 59.1 & 4.43 & 0.6 & 85.9 & 2.03 \\
TSMM & 5.00 & 65.2 & 4.30 & 0.3 & 86.6 & 2.02 \\
Dual-Anchor & 8.75 & 38.8 & 6.73 & 0.0 & 79.6 & 2.02 \\
\bottomrule
\end{tabular*}
\endgroup
\smallskip

\begin{minipage}{\linewidth}\small
We use $\mathrm{Gamma}(1,1)$ as the reference. TSMM calibrates the cluster count.
We applied the Dual-Anchor adjustment when the TSMM prior put more than 40\% probability on a randomly chosen unit belonging to a component with over half the population mass.
The last two columns give posterior means of the largest population component weight and the SD across the observed units in each draw. The SD is in achievement SDs for STAR, standardized mean differences for writing-to-learn, and logits for vocabulary.
\end{minipage}
\end{table}

\subsection{Small-Class Effects Across Schools in Project STAR}
\label{subsec:app-star}

Did smaller classes have much the same effect on kindergarten achievement across schools? In Project STAR, children were randomly assigned to small or regular-size classes within each school, allowing us to estimate a separate effect for each of the 79 schools \citep{krueger1999experimental,nye2000effects}. We examine whether priors that yield different cluster counts also give different estimates of how much these effects vary across schools or of the effect in a given school.

We include the 5,786 kindergarten students for whom the public STAR data give both test scores, class type, and school and teacher identifiers \citep{bates2026mlmrev}. Within this sample, we standardize the reading and mathematics scores separately, average the two standardized scores for each student, and standardize that average. We pool regular-size classes with and without an aide, so a positive school effect indicates higher mean achievement in small classes than in the pooled regular-size classes at that school.

We compare two ways of estimating the small-class effect and its standard error for each school. In the primary analysis, we fit a single mixed model by restricted maximum likelihood, with a separate fixed intercept and small-class coefficient for each school. A random classroom intercept accounts for dependence among students in the same classroom. As a sensitivity check, we fit ordinary least squares separately within each school and use HC3 standard errors to allow for heteroskedasticity. We then fit the same Bayesian model to each set of school estimates in turn. In Equation~\eqref{eq:model}, each estimated small-class coefficient is approximately normally distributed about its underlying school effect, with the sampling standard deviation fixed at its first-stage standard error. We use a DPM to model the distribution of the underlying effects across schools.

In the primary analysis, sampling error is large relative to the estimated variation in effects across schools ($I=0.107$). An individual school's estimate therefore tells us little about how far its effect lies from the average. We can estimate that average more precisely by pooling the school estimates. A normal random-effects model fitted by REML gives a mean small-class effect of 0.194 achievement standard deviations (95 percent confidence interval 0.097 to 0.291). This model describes variation across schools with a single normal distribution, providing a benchmark for the mean and spread of the effects without assigning schools to clusters.

For illustration, we expected five clusters among the 79 schools and assigned an 80 percent probability to a count between two and eight. OSM Table~\ref{tab:E-elicitation-worksheet} gives the probability that each calibrated prior assigned to this range. TSMM matched the expected count of five, but the resulting prior gave a randomly chosen school a 51 percent chance of belonging to a component with more than half the population weight. This exceeded the 40 percent trigger we had set before posterior simulation. We therefore adjusted the prior using Dual-Anchor, lowering that probability to 38.5 percent at the cost of raising the expected count from 5 to 6.24 (Table~\ref{tab:applications}).

Under $\mathrm{Gamma}(1,1)$, the posterior mean cluster count was 1.96 and the posterior mean weight of the largest component was 91.7 percent. Under Dual-Anchor, the mean count rose to 4.47 and the mean weight fell to 71.5 percent. The corresponding TSMM estimates were 3.12 clusters and 81.8 percent. Switching from $\mathrm{Gamma}(1,1)$ to Dual-Anchor also raised the posterior mean SD of effects across the 79 schools from 0.226 to 0.236 achievement standard deviations, an increase of about 4.5 percent (Table~\ref{tab:applications}).

\begin{figure}[!htbp]
\centering
\includegraphics[width=\textwidth]{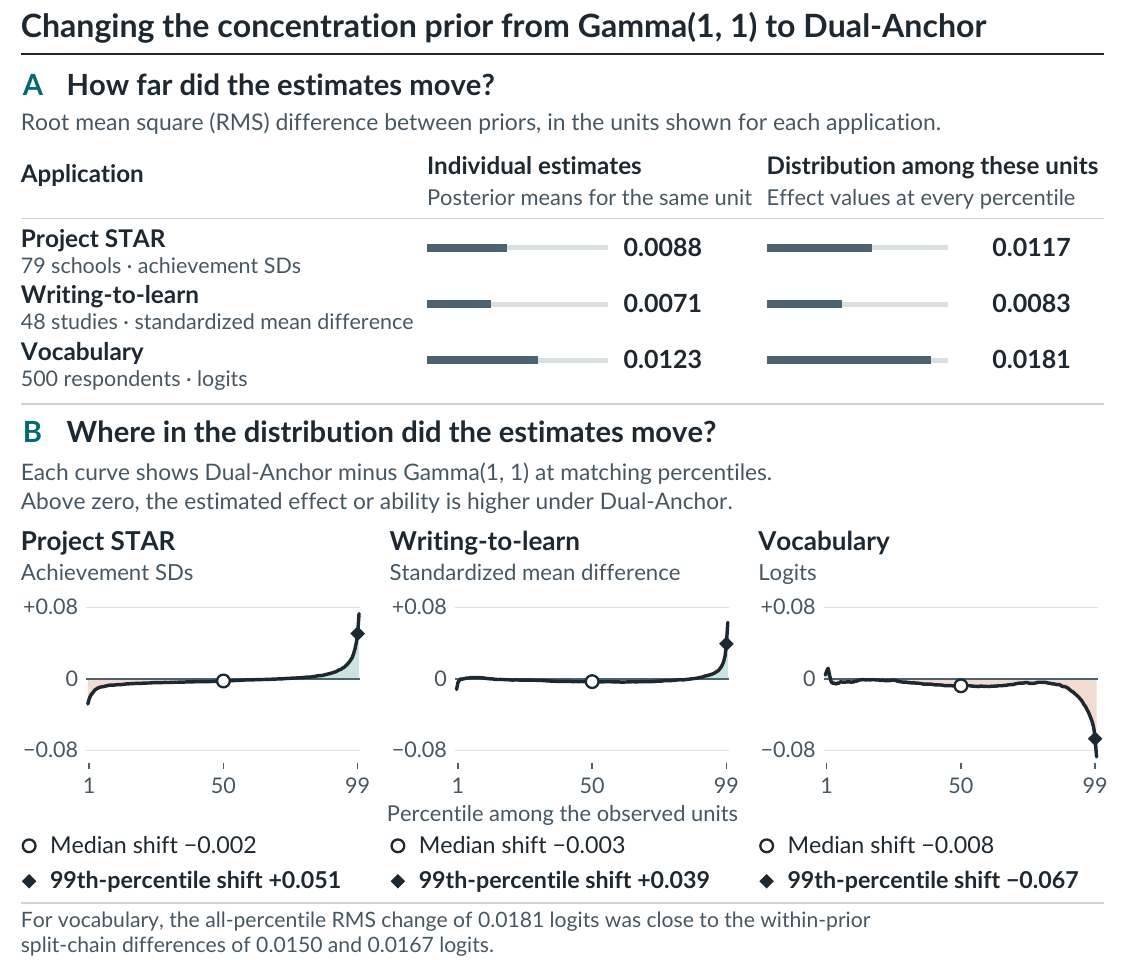}
\caption{Prior sensitivity of individual estimates and distribution percentiles. A compares two root mean square differences between priors. We match posterior means by school, study or respondent and match values from the estimated distributions of effects or abilities by percentile. In B, replacing $\mathrm{Gamma}(1,1)$ with Dual-Anchor changed the median little in all three applications. The 99th percentile rose for school and study effects and fell for vocabulary ability.}
\label{fig:app-effect-changes}
\end{figure}

For the 79 schools, the root mean square difference between priors was 0.0088 achievement standard deviations for the school posterior means and 0.0117 for the effect values at matching percentiles (Figure~\ref{fig:app-effect-changes}A). The estimated median school effect fell from 0.189 under $\mathrm{Gamma}(1,1)$ to 0.186 under Dual-Anchor, while the 99th percentile rose from 0.803 to 0.853 (Figure~\ref{fig:app-effect-changes}B).

With the classroom random-intercept model, the 95 percent posterior intervals excluded zero for the same three schools under $\mathrm{Gamma}(1,1)$ and Dual-Anchor. When we used schoolwise OLS estimates and HC3 standard errors instead, the intervals excluded zero for 22 of the 79 schools under each prior. The HC3 standard errors allow for heteroskedasticity but do not account for dependence among pupils in the same classroom. The cluster count was also sensitive to the scale priors. Using the original first-stage estimates and the Dual-Anchor concentration prior, we widened the priors for the grand mean and the variance of the component means together. The posterior mean cluster count fell from 4.47 to 2.85, and the posterior probability of a single cluster rose from 6.4 to 31.8 percent (OSM Table~\ref{tab:H4-sens}).

\FloatBarrier

\subsection{A Meta-Analysis of Writing-to-Learn Interventions}
\label{subsec:app-bd}

\citet{bangertdrowns2004effects} synthesized 48 studies comparing instruction that placed greater emphasis on writing with conventional instruction. We ask how much the effect on achievement varied from study to study. When the concentration prior yields a larger posterior cluster count, does the estimated distribution of the 48 study effects also change? We compare the distributions at matching percentiles and follow each study's estimate across priors.

The 48 studies included 5,576 participants. For each study, the bias-corrected standardized mean difference $y_j$ compares mean achievement in the writing group with mean achievement in the comparison group. A positive value indicates higher achievement in the writing group. We use $y_j$ as the observed effect estimate in Equation~\eqref{eq:model} and the square root of its reported sampling variance $v_j$ as its standard error. We treat $v_j$ as fixed and use the DPM to estimate how the underlying effects vary across studies. We obtained these summaries from \texttt{metadat::}\allowbreak\texttt{dat.bangertdrowns2004} in version 1.6-0 of the \texttt{metadat} package \citep{viechtbauer2026metadat}, as described in the \hyperref[subsec:H2-data]{OSM}.

We use the normal random-effects model as a benchmark for the mean effect of writing instruction on achievement and for the variation in effects across studies. Fitting the model by REML with \texttt{metafor} \citep{viechtbauer2010conducting} gave a pooled standardized mean difference of 0.222 (95 percent confidence interval 0.132 to 0.312) and an estimated between-study variance of 0.0499 ($I^2 = 58.4\%$). Within the DPM, we examine how the fitted mixture changes when we replace the external concentration prior DORO-Unif with TSMM or Dual-Anchor.

For illustration, we expected four clusters among the 48 studies and assigned an 80 percent probability to a count between one and eight. TSMM matched the expected count of four. Under this prior, however, a randomly chosen study had a 59.7 percent chance of belonging to a component with more than half the population weight. This exceeded the 40 percent threshold we had set before posterior simulation. We therefore used the Dual-Anchor adjustment ($\lambda = 0.13$), which lowered that probability to 39.5 percent at the cost of raising the expected count to 5.66 (Table~\ref{tab:applications}).

The posterior probability that all 48 studies formed a single cluster was 54.8 percent under $\mathrm{Gamma}(1,1)$ and 52.5 percent under TSMM. Under Dual-Anchor, it fell to 19.0 percent. The largest component also carried less population weight under Dual-Anchor than under $\mathrm{Gamma}(1,1)$, with posterior means of 77.3 and 90.4 percent, respectively. The estimated spread of effects across the 48 studies changed little. The posterior mean SD ranged from 0.250 to 0.256 across the three priors (Table~\ref{tab:applications}).

Switching from $\mathrm{Gamma}(1,1)$ to Dual-Anchor made little difference to individual study estimates on average. The root mean square change in their posterior means was 0.0071 on the standardized mean difference scale (Figure~\ref{fig:app-effect-changes}A), and the 95 percent posterior intervals excluded zero for the same 13 studies. Two additional intervals excluded zero under TSMM, though both lower endpoints were within one Monte Carlo standard error of zero (OSM Figure~\ref{fig:H-unit-changes}). The estimated median study effect also changed little, falling by 0.003. The change was larger at the 99th percentile, which rose by 0.039 (Figure~\ref{fig:app-effect-changes}B).

The cluster count was also sensitive to the scale priors. With Dual-Anchor fixed, widening the priors for the grand mean and the variance of the component means raised the posterior probability of a single cluster from 19.0 to 55.3 percent. The intervals still excluded zero for the same 13 studies (OSM Table~\ref{tab:H4-sens}). We examined grade and intervention length after fitting to describe how the estimated effects varied with these study characteristics (OSM Section~\ref{subsec:H5-results}).

\subsection{Latent Vocabulary Traits in a Rasch Model}
\label{subsec:app-irt}

In 2016, the Open-Source Psychometrics Project asked 2,495 online respondents to mark the words whose definitions they were sure they knew. The checklist contained 13 real words and three nonwords included as validity checks \citep{openpsychometrics2016gcbs}. From the data distributed by the Item Response Warehouse \citep{domingue2025irw}, we drew a random sample of 500 respondents, the largest sample size covered by the calibration tables used here. We use their responses to the 13 real words to ask whether a prior that yields more clusters also changes the estimated distribution of vocabulary ability among these respondents or the ability estimate for a given respondent.

The 500 respondents covered the full score range from 0 to 13, with 63 marking all 13 words. We fit the Rasch model to their item responses. A respondent with ability $\eta_j$ marks a word of difficulty $\beta_i$ with log odds $\eta_j-\beta_i$. We center the word difficulties at zero \citep{paganin2023computational}. Under this model, two respondents who mark the same number of words provide the same information about their abilities, even if they mark different words \citep{sanmartin2011bayesian}.

We assessed individual score precision using \citeauthor{warm1989weighted}'s \citeyearpar{warm1989weighted} weighted-likelihood estimates. With word difficulties fixed at their Dual-Anchor posterior means, score reliability was 0.71 \citep{adams2005reliability}, and the standard error at a perfect score was 1.73 logits (OSM Section~\ref{subsubsec:H-irt-wle}). At the population level, 13 items provide only partial information about the shape of the ability distribution \citep[Theorems~5--6]{sanmartin2011bayesian}. We model this distribution as a mixture of normal components. The count $K_{500}$ records how many of these components are represented in the sample, and abilities can vary within each one.

\citet{paganin2023computational} examined the prior mean and variance of the cluster count at the sample sizes in their studies. They used $\mathrm{Gamma}(2,4)$ in their simulations and $\mathrm{Gamma}(1,3)$ in their two real-data analyses. We included $\mathrm{Gamma}(1,3)$ as an external reference. It implies an average of 3.03 clusters among 500 respondents. For this checklist, after inspecting the data and before fitting the models, we expected five clusters and assigned an 80 percent probability to a count between two and ten.

Under the $\mathrm{Gamma}(1,1)$ prior, a randomly chosen respondent had a 59.1 percent chance of belonging to a component with more than half the population weight. Matching the expected count of five with TSMM raised this probability to 65.2 percent. We therefore selected the largest value of $\lambda$ on our grid that brought this probability below 40 percent. At $\lambda=0.05$, the probability fell to 38.8 percent and the expected count rose to 8.75 (Table~\ref{tab:applications}). We used this Dual-Anchor prior for the main analysis. The separate hard-bound analysis gave a concentration prior close to a point mass at $\alpha=2$ (OSM Section~\ref{subsec:E6-guidelines}).

We fitted the Rasch model under the five priors in OSM Table~\ref{tab:H-all-priors} and under a sixth prior with an expected count of 25 to assess sensitivity to a much higher count judgment. Only the fit under TSMM met the convergence criteria after the planned rerun. The other five required a third run. For each final fit we retained 4,000 draws from each of four chains (OSM Section~\ref{subsec:H-irt}).

Under $\mathrm{Gamma}(1,1)$, the posterior mean cluster count was 4.43 and the largest component accounted for 85.9 percent of the population weight on average. Under Dual-Anchor, the corresponding estimates were 6.73 clusters and 79.6 percent (Table~\ref{tab:applications}). Thus the prior affected how many clusters represented the sample, although one component continued to carry most of the population weight. The estimated spread of ability across the 500 respondents changed little. The posterior mean SD was 2.03 logits under $\mathrm{Gamma}(1,1)$ and 2.02 under Dual-Anchor.

To estimate the ability distribution among these respondents, we averaged the fraction at or below each ability level over posterior draws \citep{lee2025improving}. Over the 521 thresholds examined, the largest difference between any two of the five estimated CDFs was 0.26 percentage points (OSM Figure~\ref{fig:H-finite-cdf}). When we replaced $\mathrm{Gamma}(1,1)$ with Dual-Anchor, the 99th percentile fell by 0.067 logits, compared with a decrease of 0.008 logits at the median (Figure~\ref{fig:app-effect-changes}B).

Individual ability estimates changed little under the same two priors. The largest change in posterior mean ability was 0.031 logits for the 63 respondents who marked all 13 words, with a Monte Carlo standard error of 0.016 logits (OSM Section~\ref{subsubsec:H-irt-posterior}). Across all 500 respondents, the root mean square change was 0.0123 logits, compared with 0.0181 logits across percentiles (Figure~\ref{fig:app-effect-changes}A). The change across percentiles was close to the differences we obtained by splitting the chains within each prior (OSM Table~\ref{tab:H-effect-splits}).

\subsection{What the Hyperprior Changes}
\label{subsec:app-summary}

In all three applications, replacing $\mathrm{Gamma}(1,1)$ with Dual-Anchor raised the posterior mean cluster count and lowered the posterior mean weight of the largest population component (Table~\ref{tab:applications}). Individual posterior means changed little on average (Figure~\ref{fig:app-effect-changes}A). Within each application, the main priors also gave similar population densities for a new school, study or respondent (OSM Figure~\ref{fig:H-population-density}).

Among the observed units, the comparison depended on which part of the distribution we examined. In all three applications, the root mean square change across all percentiles was larger than the root mean square change in the individual posterior means. When we restricted the comparison to the 5th through 95th percentiles, the change across percentiles was smaller than the change in individual posterior means (Figure~\ref{fig:app-effect-changes}A and OSM Table~\ref{tab:H-effect-changes}). The estimated medians changed little. At the 99th percentile, Dual-Anchor gave higher estimates for the effects of smaller classes and of writing activities, but a lower estimate for vocabulary ability (Figure~\ref{fig:app-effect-changes}B).

We followed the same schools and studies across the two fits, identifying the upper tenth and middle half from their observed effects. Figure~\ref{fig:app-tail-unit-pooling} shows the change for each unit under Dual-Anchor relative to $\mathrm{Gamma}(1,1)$. For all eight upper-tail schools and all five upper-tail studies, the probability of sharing a cluster with a unit in the middle fell by more than the average decrease for units in the middle (dashed line). In STAR, the distance between a school's posterior mean and its observed effect decreased by an average of 0.019 achievement SDs for upper-tail schools and 0.002 for schools in the middle. In writing-to-learn, the corresponding decreases were 0.012 and less than 0.001 on the standardized mean difference scale. One study accounted for 71 percent of the total decrease among the five upper-tail studies (OSM Section~\ref{subsubsec:H-tail-pooling}).

\begin{figure}[!htbp]
\centering
\includegraphics[width=\textwidth]{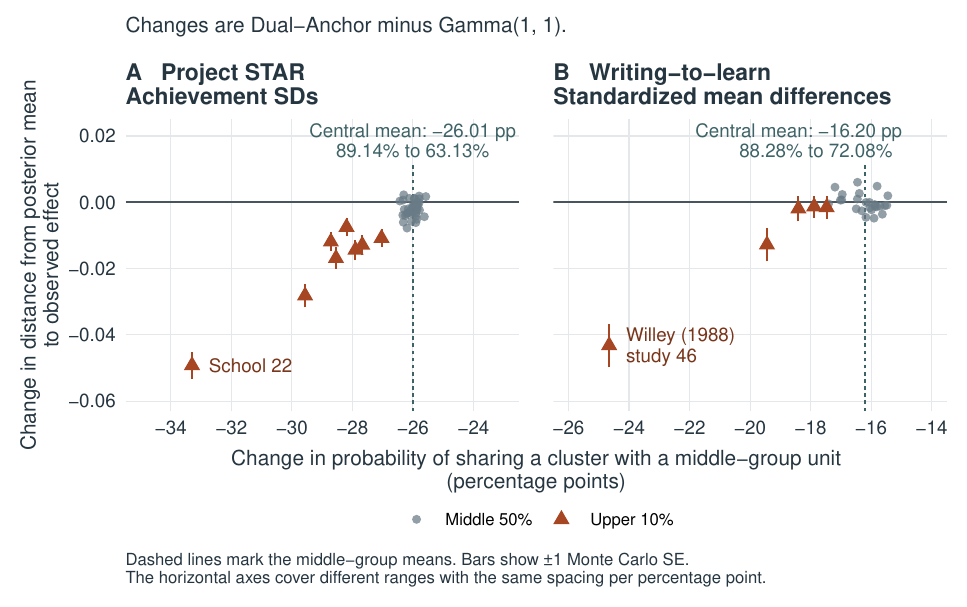}
\caption{Co-clustering and shrinkage in STAR and writing-to-learn. For each school or study shown, we calculated the probability of sharing a cluster with another unit drawn at random from the middle half, as ranked by observed effect. A point below zero indicates that the unit's posterior mean was closer to its observed effect under Dual-Anchor.}
\label{fig:app-tail-unit-pooling}
\end{figure}

\FloatBarrier

\section{Discussion}
\label{sec:discussion}

\subsection{Choosing and Reporting the Hyperprior}
\label{subsec:contributions}

How we assess a concentration prior depends on what we intend to interpret. If either the sample cluster count or a component's share of the population is of substantive interest, we need to consider what the prior implies for both. For a study with a fixed number of units, we calibrate a Gamma prior to the analyst's expected cluster count and uncertainty about that count. We then examine the probability that a randomly chosen unit belongs to a component containing more than half the population. If this probability seems too high, we show how closely an adjusted prior meets both judgments.

When the aim is to estimate effects, different partitions can give similar answers. Effects vary within components, and components with similar means can represent much the same range of effects. We therefore compare estimates across plausible priors on the scale of interest, whether we are estimating a percentile of the effect distribution or the effect for a particular unit. For an individual unit, the partition can matter because it determines which other units contribute directly to estimating its component mean.

In STAR and the writing-to-learn meta-analysis, schools and studies in the upper tenth of the observed effects were less often grouped with units in the middle under Dual-Anchor than under $\mathrm{Gamma}(1,1)$. On average, their posterior means also lay slightly closer to their observed effects. Among draws with the same cluster count, the two priors gave similar probabilities of an upper-tail unit sharing a cluster with a unit in the middle. Dual-Anchor gave more weight to draws with several clusters, in which upper-tail and middle units were less often grouped together. This offers one explanation for the modest difference in shrinkage. In a Dirichlet process mixture, a unit may form a new cluster when the existing clusters predict its data poorly \citep{neal2000markov}. \citet{lee2025improving} observed a related pattern in the Conditional Cash Transfer trial, where a prior allowing more clusters preserved some of the more extreme site estimates.

Despite the modest difference in shrinkage, the 95 percent intervals under $\mathrm{Gamma}(1,1)$ and Dual-Anchor excluded zero for the same three schools and the same 13 studies. In STAR, changing how we estimated the school effects and their standard errors had a much larger influence on which intervals excluded zero (Section~\ref{subsec:app-star}). For an effect chosen before fitting,\phantomsection\label{subsec:contributions-sensitivity} we would report the range of estimates across plausible concentration priors. If the estimates support the same substantive conclusion and Monte Carlo error is small enough to resolve a difference that would matter in practice, this comparison is sufficient. If the conclusion changes, we would identify the prior judgments on which it depends.

\subsection{Alternatives to the Dirichlet Process}
\label{subsec:why-dp}

We retained the Dirichlet process because it is used in the applied analyses that motivated this work. A finite mixture with a prior on the number of components \citep{miller2018mixture, fruhwirth2021generalized} or a Pitman--Yor process \citep{pitman1997two} gives more control over the relation between counts and weights. For a finite mixture and a given sample size, the prior distribution of the number of clusters depends on both the number of components allowed by the model and the weights assigned to them.

Count and weight judgments can conflict because the same concentration parameter governs both features of the Dirichlet process \citep{antoniak1974mixtures,sethuraman1994constructive}. If the analyst has firm requirements for both and no Gamma hyperprior satisfies them, a model with additional flexibility is appropriate. Section~\ref{subsec:dual-anchor} shows how closely each of the Gamma priors we compare meets the count and weight targets. A Pitman--Yor process adds a discount parameter that changes the partition and weight distributions \citep{pitman1997two}. In a mixture of finite mixtures, the number of components itself has a prior distribution \citep{miller2018mixture,fruhwirth2021generalized}. Extending the calibration to either model would allow a broader range of judgments.

\subsection{Design-Conditional Elicitation and Sample Size}
\label{subsec:ssd-ssi}

\citet{vicentini2025prior} prefer weight-based anchors because their distributions do not depend on sample size. We agree that these anchors are appropriate when observations accumulate over time or several data sets share one concentration parameter. Recalibrating a count prior at each sample size would then define a sequence of priors in place of a common specification.

Our focus is a fixed set of sites, studies, or examinees. The expected cluster count depends on the number of units in the study. \citet{vicentini2025prior} also allow for sample-size-dependent priors in some settings. We regard a study with a fixed design as one such setting. We start from the count and require the analyst to examine the implied weights.

A prior chosen for its weights also has implications for the cluster count. In our simulations we calibrated a hyperprior to $\Pr(W_{\mathrm{SB}}>0.5)=0.25$ and $\Pr(W_{\mathrm{SB}}>0.9)=0.05$. With the $N(0,1)$ grand-mean prior and $I=0.2$, the posterior mean of $K_J$, averaged over replications, ranged from 3.39 to 3.42 across the designs with $K^*=5,10,$ and 30. All fits under this hyperprior converged in these three designs. With 30 generating components, the weight-calibrated prior understated the count on average, and none of the 95 percent posterior intervals contained 30 (Section~\ref{subsec:sim-main}). The suitability of the weight targets depends on the design and the quantity to be estimated. An analyst with firmer judgments about the weights can calibrate to those judgments first and then examine the implied count.

\subsection{Reporting in Education and Psychometrics}
\label{subsec:reporting-practice}

In reporting a concentration prior, we would explain why we expected a given number of clusters among the schools, studies or examinees in the sample, and how uncertain we were about that count. Earlier studies or practitioners familiar with the setting may inform these judgments. In our applications we supplied them for illustration. For the vocabulary checklist, we chose the expected count after examining the responses and before fitting the Rasch models. We would likewise explain why we chose the weight thresholds in Box~\ref{box:recipe} for the study at hand. OSM Section~\ref{subsec:E6-guidelines} gives our record for each application.

\citet{paganin2023computational} compared Gamma priors for item response models by calculating the mean and variance of the cluster count under each prior at the sample sizes in their studies. Our procedure automates the inverse calculation. Starting from a proposed mean and variance for the count, TSMM obtains initial values for the Gamma shape and rate in closed form and refines them numerically at the study's sample size. The calculation is fast enough to prepare several priors for sensitivity analysis. For example, we can hold the expected count fixed and vary our uncertainty about it, as in Table~\ref{tab:vif}. For each prior, we would report the implied count distribution alongside the probability that a randomly chosen unit belongs to a component containing more than half the population.

When reporting the fitted model, we would give the cluster count for the units studied and the estimated population share of the largest component. In a multisite trial, differences in effects across clusters may suggest school characteristics worth examining. In a meta-analysis, they may suggest study-level moderators for further investigation. For the vocabulary checklist, Dual-Anchor gave a posterior mean of 6.73 clusters among the 500 respondents, yet the largest component accounted for an estimated 79.6 percent of the population.

\subsection{Limitations and Future Directions}
\label{subsec:limitations}

We assessed recovery using one family of normal mixtures with equally spaced component means. We considered Gamma priors for the concentration parameter and did not compare the Stirling--gamma prior of \citet{zito2024bayesian}. We summarized only the fits that met the convergence criteria. This restriction matters most in the 200-site design with the most precise site estimates, where fewer than half the fits converged in each condition (Section~\ref{subsec:sim-count-sensitivity}).

With the true effects and abilities unknown, our application results describe how much the estimates change across priors. In STAR, conclusions about individual schools also depended on how we estimated the school effects and their standard errors. For the 13-item vocabulary checklist, estimated ability percentiles differed about as much overall between $\mathrm{Gamma}(1,1)$ and Dual-Anchor as between sets of chains run under the same prior (OSM Table~\ref{tab:H-effect-splits}). This makes the small shift in the upper tail difficult to attribute to the prior.

We supplied the count judgments in our examples. A next step is to ask practitioners considering the same study how many clusters they expect, how uncertain they are, and whether they would revise those judgments after seeing the implied population weights. We could then fit the same data under the priors they chose and examine whether their judgments lead to different conclusions about the effects of interest.

\subsection{Conclusion}
\label{subsec:conclusion}

For a fixed number of units, the concentration prior governs both how many clusters we expect in the sample and how much of the population a single component might represent. We calibrate a Gamma prior to the expected count and its uncertainty, then show how the count distribution changes if we adjust the prior to bring the population weights closer to the analyst's judgment. In all three applications, Dual-Anchor gave higher posterior mean cluster counts than $\mathrm{Gamma}(1,1)$, while individual estimates changed little on average. When the substantive question concerns a particular unit or a percentile of the effect distribution, we would compare estimates for that unit or percentile directly across plausible priors and report the difference on the effect scale.


\printbibliography[title={References}]

\clearpage
\appendix
\setcounter{section}{0}
\renewcommand{\thesection}{\Alph{section}}
\numberwithin{equation}{section}
\numberwithin{figure}{section}
\numberwithin{table}{section}

\begin{center}
{\Large\bfseries Online Supplemental Materials}\\[0.75em]
{\large Design-Conditional Prior Elicitation for Dirichlet Process Mixtures:\\
A Unified Framework for Cluster Counts and Weight Control}\\[0.75em]
{\normalsize JoonHo Lee}
\end{center}

\vspace{1em}
\DoToC

\section{Exact Distribution and Moments of the Cluster Count}
\label{app:dp-crp}

\subsection{The Conditional Count Distribution}
\label{subsec:A1-purpose}
\label{subsec:A2-dist-theory}

We use the exact distribution of $K_J$ to calibrate the prior and report its count probabilities. Here $J$ is fixed, $G\sim\mathrm{DP}(\alpha,G_0)$ has a non-atomic base measure, and $K_J$ counts the clusters among $J$ draws from $G$, one for each distinct value of the component parameters. In a normal mixture, it does not count distinct unit effects: effects within a component remain continuous. Throughout, the hyperprior has shape--rate density
\begin{equation}
p(\alpha\mid a,b)=\frac{b^a}{\Gamma(a)}\alpha^{a-1}e^{-b\alpha},\qquad a,b,\alpha>0.
\label{eq:A1}
\end{equation}

Under the Chinese restaurant process, unit $i$ starts a new cluster with probability $\alpha/(\alpha+i-1)$ \citep{blackwell1973ferguson}. This probability depends only on $i$ and $\alpha$, so the new-cluster indicators are mutually independent given $\alpha$. Thus the classical Bernoulli-sum representation is
\begin{equation}
K_J\overset{d}{=}\sum_{i=1}^J B_i,\qquad
B_i\mid\alpha\ \text{independent},\qquad
B_i\mid\alpha\sim\mathrm{Bernoulli}\!\left(\frac{\alpha}{\alpha+i-1}\right),
\label{eq:A5}
\end{equation}
with $B_1=1$ \citep[Sec.~5.2]{arratia2003logarithmic}. The equivalent Antoniak probability mass function is
\begin{equation}
\Pr(K_J=k\mid\alpha)
=|s(J,k)|\frac{\alpha^k\Gamma(\alpha)}{\Gamma(\alpha+J)},
\qquad k=1,\ldots,J,
\label{eq:A8}
\end{equation}
where $|s(J,k)|$ is an unsigned Stirling number of the first kind \citep{antoniak1974mixtures}. As $\alpha$ tends to zero, $K_J$ tends to one; as $\alpha$ increases without bound, $K_J$ tends to $J$.

\subsection{Exact Moments and Gamma Mixing}
\label{subsubsec:A23-moments}

Summing the Bernoulli means and variances gives
\begin{align}
\kappa_J(\alpha):=\E(K_J\mid\alpha)
&=\sum_{r=0}^{J-1}\frac{\alpha}{\alpha+r}
=\alpha\{\psi(\alpha+J)-\psi(\alpha)\},
\label{eq:A12}\\
v_J(\alpha):=\Var(K_J\mid\alpha)
&=\sum_{r=0}^{J-1}\frac{\alpha r}{(\alpha+r)^2}
=\kappa_J(\alpha)-\alpha^2\{\psi_1(\alpha)-\psi_1(\alpha+J)\},
\label{eq:A13}
\end{align}
where $\psi$ and $\psi_1$ are the digamma and trigamma functions. The finite sums follow directly from~\eqref{eq:A5}; the special-function forms allow evaluation without summing over $J$ terms. For $J\ge2$, $0<v_J(\alpha)<\kappa_J(\alpha)-1$: the nonconstant Bernoulli sum $K_J-1$ is underdispersed relative to a Poisson variable with the same mean. This is one reason to refine the Poisson-based initializer.

Integrating~\eqref{eq:A8} against~\eqref{eq:A1} gives the prior-predictive count distribution \citep{dorazio2009selecting,murugiah2012selecting}:
\begin{equation}
p_{a,b}(k)=\frac{b^a|s(J,k)|}{\Gamma(a)}
\int_0^\infty\alpha^{k+a-1}e^{-b\alpha}
\frac{\Gamma(\alpha)}{\Gamma(\alpha+J)}\,d\alpha.
\label{eq:A21}
\end{equation}
We evaluate these one-dimensional integrals when reporting a median, interval or individual count probability. For moment matching, we avoid computing the full PMF and instead use total expectation and total variance:
\begin{align}
M_1(a,b):=\E_{a,b}(K_J)&=\E_\alpha\{\kappa_J(\alpha)\},
\label{eq:A22}\\
V(a,b):=\Var_{a,b}(K_J)&=\E_\alpha\{v_J(\alpha)+\kappa_J(\alpha)^2\}-M_1(a,b)^2.
\label{eq:A23}
\end{align}
All expectations on the right are over $\alpha\sim\mathrm{Gamma}(a,b)$.

At a given count mean, this decomposition gives a lower bound on the variance for any prior on $\alpha$.

\begin{proposition}[Lower bound on the count variance]
\label{prop:var-lower-bound}
Let $J\ge2$ and $1<\mu_K<J$, and let $\alpha^*$ solve $\kappa_J(\alpha^*)=\mu_K$. For every prior $\Pi$ on $\alpha>0$ with $\E_\Pi\{\kappa_J(\alpha)\}=\mu_K$, the induced count satisfies $\Var_\Pi(K_J)\ge v_J(\alpha^*)$. For $J\ge3$ the inequality is strict unless $\Pi$ is a point mass at $\alpha^*$.
\end{proposition}

\begin{proof}
For $r=1,\ldots,J-1$ write $p_r=\alpha/(\alpha+r)$, $u_r=1/(\alpha+r)$ and $w_r=\partial p_r/\partial\alpha=r\,u_r^2$, and let $W=\sum_r w_r=\kappa_J'(\alpha)>0$. Since $\kappa_J$ is strictly increasing, we may write $v_J$ as a function $g$ of $m=\kappa_J(\alpha)$. With $M=\kappa_J(\alpha)$ and $h(m)=g(m)+m^2$, total variance gives $\Var_\Pi(K_J)=\E_\Pi\{h(M)\}-\mu_K^2$. Since $v_J'(\alpha)=\sum_r w_r(1-2p_r)$, we have $g'(m)=1-2\rho$ with $\rho=\sum_r w_rp_r/W$, and hence $h''(m)=2(1-\rho'/W)$, where $\rho'$ is the derivative with respect to $\alpha$. Using $w_r'=-2w_ru_r$ and $p_r=\alpha u_r$,
\[
\rho'=\frac{\sum_r w_r^2}{W}
-\frac{2\alpha}{W^2}\Bigl\{W\sum_r w_ru_r^2-\Bigl(\sum_r w_ru_r\Bigr)^{2}\Bigr\}
\le\frac{\sum_r w_r^2}{W}\le W .
\]
The first inequality follows from Cauchy--Schwarz and the second from the nonnegativity of the $w_r$. Thus $h$ is convex on $(1,J)$, and Jensen's inequality gives $\E_\Pi\{h(M)\}\ge h(\mu_K)=v_J(\alpha^*)+\mu_K^2$. For $J\ge3$ at least two of the $w_r$ are positive, so $\sum_r w_r^2<W^2$ and $h$ is strictly convex; equality then requires $M$ to be constant and hence the prior to put all its mass at $\alpha^*$.
\end{proof}

No hyperprior can give a count variance below $v_J(\alpha^*)$ at the stated mean. For $J\ge3$, every Gamma hyperprior gives a count variance strictly above the bound. Even if the revised target exceeds the bound, we must still check that the chosen hyperprior matches both moments. For $J=100$ and $\mu_K=6.5$ the bound is $4.5224$ at $\alpha^*=1.38597$, the example in main-text Section~\ref{subsec:elicit}. For the stated means in the three applications it is 3.36 (79 schools, mean 5), 2.51 (48 studies, mean 4), and 3.64 (500 respondents, mean 5).

\subsection{Numerical Evaluation}
\label{subsubsec:A26-computation}

We compute the Stirling numbers in log space to avoid overflow. With $L_{n,k}=\log|s(n,k)|$ and $\log0=-\infty$, the standard Stirling recurrence becomes
\begin{equation}
L_{n,k}=\mathrm{logaddexp}\{L_{n-1,k-1},\log(n-1)+L_{n-1,k}\}.
\label{eq:A31}
\end{equation}
We initialize $L_{1,1}=0$, set entries outside the triangular support to $-\infty$, and use $\mathrm{logaddexp}(u,v)=\max(u,v)+\log\{1+e^{-|u-v|}\}$. The conditional log-PMF is then
\begin{equation}
\log p_{J,k}(\alpha)=L_{J,k}+k\log\alpha+\log\Gamma(\alpha)-\log\Gamma(\alpha+J).
\label{eq:A32}
\end{equation}
We normalize by log-sum-exp and check that probabilities are nonnegative and sum to one.

For the Gamma expectations, set $x=b\alpha$. Generalized Gauss--Laguerre quadrature for the weight $x^{a-1}e^{-x}$ gives
\begin{equation}
\E_{a,b}\{h(\alpha)\}
=\frac{1}{\Gamma(a)}\int_0^\infty h(x/b)x^{a-1}e^{-x}\,dx
\approx\sum_{m=1}^M\tilde w_m h(x_m/b),
\label{eq:A33}
\end{equation}
where $x_m$ are the nodes and $\tilde w_m$ are weights normalized by $\Gamma(a)$ (\href{https://dlmf.nist.gov/3.5#v}{DLMF \S3.5(v)}). At each quadrature order we use the same nodes for the moments and the score expectations. Moment evaluation costs $O(M)$; the full count PMF costs $O(JM)$ after the log-Stirling table has been cached.

The package default is $M=80$. We use 160 nodes for fitting and independently recompute at 320 nodes for the reported analyses. Higher-order checks are part of accepting a calibration, not an optional accuracy comparison (Appendix~\ref{app:stage2}).

For main-text Figure~\ref{fig:defaults}, the recalibration passed our checks at the plotted sample sizes from $J=20$ through $J=500$. It failed those checks at $J=10$, the only smaller size on the grid. Above $J=500$, we computed count moments for the fixed priors but did not compute count quantiles or recalibrate.

\section{Stage 1: A Closed-Form Initializer}
\label{app:stage1}

\subsection{The Shifted Poisson--Gamma Approximation}
\label{subsec:B1-purpose}
\label{subsec:B2-proxy}

Stage~1 converts the target count moments $(\mu_K,\sigma_K^2)$ into starting values for the exact numerical calibration. We use
\begin{equation}
K_J-1\mid\alpha\approx\mathrm{Poisson}(\alpha c_J),\qquad c_J=\log J,
\label{eq:B1}
\end{equation}
so that the approximation respects the minimum count of one. Poisson approximations for the number of clusters are classical \citep{arratia2000poisson}; \citet[Sec.~3.1.2]{vicentini2025prior} use an unshifted negative-binomial approximation under Gamma mixing. Our shifted approximation is a computational starting point, not an exact finite-$J$ model or a limit to a fixed distribution.

If $X\mid\alpha\sim\mathrm{Poisson}(\alpha c_J)$ and $\alpha\sim\mathrm{Gamma}(a,b)$, direct integration gives
\begin{equation}
X\sim\mathrm{NegBin}\!\left(a,\frac{b}{b+c_J}\right),\qquad
\E(X)=\frac{ac_J}{b},\qquad
\Var(X)=\frac{ac_J}{b}+\frac{ac_J^2}{b^2}.
\label{eq:B5}
\end{equation}
Here the negative binomial has support $\{0,1,\ldots\}$. This identity is exact for the Poisson proxy; the approximation lies in replacing $K_J-1$ by $X$.

\subsection{Inversion and the Projection Used for Initialization}
\label{subsec:B3-inversion}
\label{subsubsec:B31-mapping}

Let $\mu_0=\mu_K-1$. Matching~\eqref{eq:B5} to $(\mu_0,\sigma_K^2)$ gives
\begin{equation}
a_0=\frac{\mu_0^2}{\sigma_K^2-\mu_0},\qquad
b_0=\frac{\mu_0\log J}{\sigma_K^2-\mu_0},
\qquad \mu_0>0,\quad\sigma_K^2>\mu_0.
\label{eq:B9}
\end{equation}
To obtain the result, the mean equation gives $b=ac_J/\mu_0$; substituting it into the variance equation gives $\sigma_K^2=\mu_0+\mu_0^2/a$. The initializer therefore has $\E(\alpha)=\mu_0/\log J$ and $\mathrm{CV}(\alpha)=1/\sqrt{a_0}$.

The proxy requires $\sigma_K^2>\mu_0$. When $\sigma_K^2\le\mu_0$, the initializer uses a variance just above this lower bound. Let $U=(\mu_K-1)(J-\mu_K)$ be the largest count variance at the stated mean. When the gap $U-\mu_0$ is positive, the default increment and its capped value are
\begin{equation}
\varepsilon_0=10^{-6}(1+\mu_0^2),\qquad
\varepsilon_1=\min\{\varepsilon_0,(U-\mu_0)/2\}.
\label{eq:B12}
\end{equation}
The code uses $\mu_0+\varepsilon_1$ in~\eqref{eq:B9}, adjusting the buffer at machine precision so the starting variance lies in $(\mu_0,U]$. If the Stage~1 call fails, including when this interval contains no representable value, Stage~2 tries a fixed grid of starting points. In either case Stage~2 uses the original targets $(\mu_K,\sigma_K^2)$.

For $J=50$ and targets $(5,10)$,~\eqref{eq:B9} gives $(a_0,b_0)=(2.667,2.608)$. Its exact count moments are 4.415 and 5.618, well below the targets. Stage~2 gives $(a,b)=(1.4082,1.0770)$ and moments 5.000 and 10.000. Appendix~\ref{app:error-bounds} explains why increasing $J$ alone does not make the initializer reliable for every target.

\subsection{Converting an Uncertainty Judgment}
\label{subsec:B4-confidence}

The qualitative confidence interface uses
\begin{equation}
\sigma_K^2=\mathrm{VIF}(\mu_K-1),
\label{eq:B13}
\end{equation}
with VIF values 1.5, 2.5 and 5 for high, medium and low confidence. At $\mu_K=5$, these correspond to variances 6, 10 and 20. The labels are starting points for elicitation: the analyst should examine the resulting count distribution and decide whether it expresses the intended uncertainty. A coefficient of variation gives $\sigma_K^2=(\mathrm{CV}_K\mu_K)^2$ directly.

For an interval $[k_L,k_U]$ assigned probability $q$, the analysis code can use the normal proxy
\begin{equation}
\sigma_K^2\approx\left\{\frac{k_U-k_L}{2z_{(1+q)/2}}\right\}^{2},
\label{eq:B14}
\end{equation}
taking the midpoint as $\mu_K$ unless a mean is supplied separately. The package also offers a discrete maximum-entropy conversion on $\{1,\ldots,J\}$ (main text \S3.2). After either conversion, we report the target moments, the central interval under the calibrated prior, and the probability that the count lies in the interval the analyst originally specified. This probability can differ from $q$. Tight targets may be incompatible with the Gamma family; neither conversion nor the Stage~1 projection establishes feasibility.

\section{Stage 2: Exact-Moment Newton Refinement and Verification}
\label{app:stage2}

\subsection{Target Equations and Jacobian}
\label{subsec:C1-purpose}
\label{subsec:C2-problem}
\label{subsec:C3-jacobian}

Stage~2 solves
\begin{equation}
F(a,b)=\begin{pmatrix}M_1(a,b)-\mu_K\\V(a,b)-\sigma_K^2\end{pmatrix}=\mathbf0,
\label{eq:C6}
\end{equation}
where the exact moments are given by~\eqref{eq:A22}--\eqref{eq:A23}. We differentiate the Gamma expectations to obtain the Jacobian. For a function $h(\alpha)$ that does not itself depend on $(a,b)$, differentiation under the integral gives
\begin{equation}
\partial_\theta\E_{a,b}\{h(\alpha)\}
=\E_{a,b}\{h(\alpha)s_\theta(\alpha)\},\qquad\theta\in\{a,b\},
\label{eq:C7}
\end{equation}
with Gamma scores
\begin{equation}
s_a(\alpha)=\log b-\psi(a)+\log\alpha,\qquad
s_b(\alpha)=a/b-\alpha.
\label{eq:C8}
\end{equation}
This follows from $\partial_\theta p=p\,\partial_\theta\log p$. The conditional count moments are bounded for fixed $J$, and the Gamma law has the log and polynomial moments needed to justify the differentiation.

Writing $M_1=\E(\kappa_J)$ and $V=\E(v_J+\kappa_J^2)-M_1^2$, we obtain
\begin{align}
\partial_\theta M_1&=\E\{\kappa_J(\alpha)s_\theta(\alpha)\},
\label{eq:C10}\\
\partial_\theta V&=\E\{[v_J(\alpha)+\kappa_J(\alpha)^2]s_\theta(\alpha)\}
-2M_1\partial_\theta M_1.
\label{eq:C11}
\end{align}
Both conditional moments $\kappa_J$ and $v_J+\kappa_J^2$ tend to one as $\alpha\downarrow0$. Since the Gamma score has expectation zero, we subtract this limit from each conditional moment before multiplying by the score. When two quadrature orders are available, we evaluate the derivatives at both orders. We use the higher order for the Jacobian and the fitting order for the moments in the residuals.

To keep both parameters positive, we work with $\eta=(\log a,\log b)$ and $G(\eta)=F(e^{\eta_1},e^{\eta_2})$. The Jacobian in these coordinates is
\begin{equation}
\mathbf J_G(\eta)=
\begin{pmatrix}\partial_aM_1&\partial_bM_1\\\partial_aV&\partial_bV\end{pmatrix}
\mathrm{diag}(a,b).
\label{eq:C14}
\end{equation}

\subsection{Algorithm and Numerical Settings}
\label{subsec:C4-algorithm}

\noindent\textbf{Algorithm C.1: exact-moment Newton refinement (A2-MN).}
The algorithm takes $J$, the target $y=(\mu_K,\sigma_K^2)$ and the quadrature and solver settings as inputs. We give the defaults for \texttt{DPprior} 2.0.0.

\begin{enumerate}[leftmargin=2em,itemsep=5pt,topsep=4pt]
\item Check the necessary support conditions
\begin{equation}
1<\mu_K<J,\qquad0<\sigma_K^2\le(\mu_K-1)(J-\mu_K).
\label{eq:C15}
\end{equation}
These do not establish that a Gamma solution exists. Use the starting values from Stage~1. If this fails, take the finite point on the fixed ten-point log-parameter grid with the smallest objective in~\eqref{eq:C19}.

\item At each iteration, divide each residual $G_i(\eta)$ by its tolerance $d_i=10^{-8}+10^{-8}\max(|y_i|,1)$. Let $D=\mathrm{diag}(d_1,d_2)$. If both standardized residuals have absolute value at most one, proceed to verification. Otherwise solve
\begin{equation}
D^{-1}\mathbf J_G(\eta)\,\Delta\eta=-D^{-1}G(\eta).
\label{eq:C18}
\end{equation}
Before solving, compute the ratio of the smallest to the largest singular value of the scaled Jacobian. Stop the Newton iteration if the ratio is at most the square root of machine precision. Also stop if the full step has $\|\Delta\eta\|_2<10^{-10}$.

\item Try up to 20 step lengths, starting at $u=1$ and halving the length after each unsuccessful trial. Accept the first finite candidate in the log-parameter bounds $[-15,15]^2$ for which
\begin{equation}
Q(\eta+u\Delta\eta)<Q(\eta),\qquad Q(\eta)=\|D^{-1}G(\eta)\|_2^2.
\label{eq:C19}
\end{equation}
If no step is accepted, stop the Newton iteration. If either residual is outside tolerance when Newton stops (at the latest after 20 iterations), run Nelder--Mead from the best Newton candidate. Accept its solution only if the algorithm improves $Q$, terminates successfully and meets both residual tolerances.

\item Recompute the selected candidate's moments at the fitting order $M$ and independently at the higher order $M'$. Both moment vectors must satisfy the componentwise target tolerances. For each moment, their absolute difference must also be at most $10^{-10}+10^{-8}\max(1,|m_M|,|m_{M'}|)$, where $m_M$ and $m_{M'}$ are its two estimates. We retain finite, interior count calibrations that pass both the solver and the verification checks.
\end{enumerate}

By default, \texttt{DPprior} fits at order $M=80$ and verifies at order $M'=\max(2M,M+40)$ when available. We used 160 nodes for fitting and 320 for verification. The software records each solver attempt, the residuals, the verification result and the final status. Local Newton convergence requires a nonsingular root and a sufficiently close starting point.

\subsection{Distributional Projection as an Optional Alternative}
\label{subsec:C7-KL}

For a full target PMF $p^\star(k)$, the software also provides the distributional projection of \citet{dorazio2009selecting}:
\begin{equation}
(a^\dagger,b^\dagger)\in\arg\min_{a,b>0}
\sum_{k=1}^Jp^\star(k)\log\frac{p^\star(k)}{p_{a,b}(k)},
\label{eq:C22}
\end{equation}
where $p_{a,b}$ is~\eqref{eq:A21}. The full PMF and its score-based derivatives require $O(JM)$ work per evaluation. A stable optimum need not reproduce the target adequately: the software checks numerical stability and target discrepancy separately. We match the exact count moments for TSMM before any Dual-Anchor adjustment. In the simulations and the school and study analyses, we report DORO-Unif as an external reference because its fitted count distribution differed from the uniform target by more than the prespecified tolerance.

\section{Why the Finite-\texorpdfstring{$J$}{J} Refinement Matters}
\label{app:error-bounds}

\subsection{The Two Approximation Errors}
\label{subsec:D1-purpose}
\label{subsec:D2-decomposition}

Stage~1 first approximates the Bernoulli sum $S_J=K_J-1$ by a Poisson variable with the same mean, then replaces that mean by $\alpha\log J$. These are different approximations. Write
\begin{equation}
\lambda_J(\alpha)=\kappa_J(\alpha)-1
=\alpha\{\psi(\alpha+J)-\psi(\alpha+1)\}.
\label{eq:D6}
\end{equation}
For total variation distance $d_{\mathrm{TV}}$, define
\begin{align}
E_1(\alpha)&=d_{\mathrm{TV}}\{\mathcal L(S_J\mid\alpha),\mathrm{Pois}(\lambda_J(\alpha))\},
\label{eq:D7}\\
E_2(\alpha)&=d_{\mathrm{TV}}\{\mathrm{Pois}(\lambda_J(\alpha)),\mathrm{Pois}(\alpha\log J)\}.
\label{eq:D8}
\end{align}
For $S_J=K_J-1$, write $q_{a,b}(s)=p_{a,b}(s+1)$, where $p_{a,b}$ is the count probability mass function in~\eqref{eq:A21}. Let $\tilde q_{a,b}$ denote its negative-binomial approximation. Then
\begin{equation}
d_{\mathrm{TV}}(q_{a,b},\tilde q_{a,b})\le\E_\alpha E_1(\alpha)+\E_\alpha E_2(\alpha).
\label{eq:D9}
\end{equation}
For each $\alpha$, the triangle inequality gives the sum of the two errors. Integrating the absolute difference of the conditional PMFs bounds the absolute difference of their integrals, giving~\eqref{eq:D9}. Thus Gamma mixing propagates the conditional errors without adding a third approximation. The exact PMF is extended by zero outside $\{0,\ldots,J-1\}$; the proxy has unbounded support.

\subsection{Conditional Bounds}
\label{subsec:D3-conditional}

The Stein--Chen bound for an independent Bernoulli sum \citep[Sec.~3.4]{arratia2003logarithmic} gives
\begin{equation}
E_1(\alpha)\le\min\{1,1/\lambda_J(\alpha)\}
\alpha^2\{\psi_1(\alpha+1)-\psi_1(\alpha+J)\}.
\label{eq:D10}
\end{equation}
The last factor is $\sum_{j=2}^Jp_j(\alpha)^2=\lambda_J(\alpha)-v_J(\alpha)$, with $p_j(\alpha)=\alpha/(\alpha+j-1)$. For fixed $\alpha$, this factor remains bounded while $\lambda_J(\alpha)$ grows as $\alpha\log J$. The resulting $O(1/\log J)$ rate is slow.

The mean approximation has a different error. The digamma expansion and its first-omitted-term bound (\href{https://dlmf.nist.gov/5.11#ii}{DLMF \S5.11(ii)}) yield
\begin{equation}
\left|\lambda_J(\alpha)-\alpha\log J+\alpha\psi(\alpha+1)\right|
\le\frac{\alpha^2}{J}+\frac{\alpha}{2J}+\frac{\alpha}{12J^2}.
\label{eq:D12}
\end{equation}
To obtain this bound, write the expression inside the absolute value as $\alpha\{\psi(\alpha+J)-\log J\}$, then use
$|\psi(x)-\log x|\le1/(2x)+1/(12x^2)$ and $\log(1+\alpha/J)\le\alpha/J$.
Consequently,
\begin{equation}
\lambda_J(\alpha)=\alpha\log J-\alpha\psi(\alpha+1)
+O\{\alpha(1+\alpha)/J\}.
\label{eq:D13}
\end{equation}
The residual $-\alpha\psi(\alpha+1)$ persists as $J$ increases for fixed $\alpha$, except at its zero. The approximation can therefore retain an absolute mean bias even when its relative error decreases.

Pinsker's inequality gives the following bound for the second error
\citep[Lemma~11.6.1]{cover2006elements}, using natural logarithms:
\begin{equation}
E_2(\alpha)\le
\left[\frac12\left\{\lambda_J(\alpha)\log\frac{\lambda_J(\alpha)}{\alpha\log J}
+\alpha\log J-\lambda_J(\alpha)\right\}\right]^{1/2}.
\label{eq:D15}
\end{equation}
The expression in braces is the KL divergence between the two Poisson laws. Equations~\eqref{eq:D10} and~\eqref{eq:D15}, averaged over the specified Gamma prior, bound the finite-$J$ error without relying on an asymptotic rate alone.

\subsection{Mixing, Moment Accuracy and Practical Use}
\label{subsec:D4-marginal}
\label{subsec:D5-thresholds}

For fixed $a,b>0$, the bounds imply $\E E_1=O(1/\log J)$ and $\E E_2=O(1/\sqrt{\log J})$. The constants depend on the Gamma hyperparameters. In particular, for $J\ge2$,
\begin{align*}
(\log J)E_1(\alpha)&\le\frac{\pi^2}{6}\alpha(1+\alpha),\\
\sqrt{\log J}\,E_2(\alpha)&\le\sqrt{\alpha/2}
\{|\psi(\alpha+1)|+\alpha/2+1/4+1/48\}.
\end{align*}
The first bound uses $\sum p_j^2\le\alpha^2\pi^2/6$ and $\lambda_J\ge\alpha\log J/(1+\alpha)$. For the second, apply
$\mathrm{KL}\{\mathrm{Pois}(\lambda)\Vert\mathrm{Pois}(\lambda')\}\le(\lambda-\lambda')^2/\lambda'$
to~\eqref{eq:D15}, then use~\eqref{eq:D12}. Both envelopes have finite expectation under every fixed Gamma law. This justifies mixing the rates; it does not provide uniform guarantees for sequences of hyperpriors that change with $J$.

Distributional accuracy also does not establish moment accuracy for the unbounded negative-binomial proxy. Write $\tilde K=1+\tilde S$, where
$\tilde S\sim\mathrm{NegBin}\{a,b/(b+\log J)\}$. The bounded count $\min(\tilde K,J)$ can be compared with $K_J$ using total variation, but the unbounded proxy requires tail corrections, including $\E\{(\tilde K-J)_+\}$ for the mean. We therefore check the exact moments themselves before accepting Stage~1 alone.

Design size alone cannot establish the accuracy of the initializer. We used Stage~2 for all reported calibrations. We would use Stage~1 alone only if the count mean and variance meet their targets at both quadrature orders and the two calculations agree within the tolerances in Section~\ref{subsec:C4-algorithm}.

\section{Weight Diagnostics and Dual-Anchor Calibration}
\label{app:weights}

\subsection{Size-Biased and Ranked Weights}
\label{subsec:E1-purpose}
\label{subsec:E2-stickbreaking}

The stick-breaking representation is
\begin{equation}
G=\sum_{h\ge1}w_h\delta_{\theta_h^*},\qquad
w_h=v_h\prod_{\ell<h}(1-v_\ell),\qquad
v_h\stackrel{iid}{\sim}\mathrm{Beta}(1,\alpha),
\label{eq:E2}
\end{equation}
with atoms drawn independently from $G_0$ \citep{sethuraman1994constructive}. The weights are in size-biased order. In particular, $w_1$ has the same marginal law as a fresh size-biased pick $W_{\mathrm{SB}}$ from the ranked masses: conditional on those masses, we select a component with probability equal to its weight. Equivalently, this is the weight of the component to which a randomly chosen unit belongs. That component need not have the largest weight \citep{vicentini2025prior}.

Thus $\Pr(W_{\mathrm{SB}}>0.5)$ is the probability that a unit belongs to a component holding more than half the population mass. The probability that any such component exists, $\Pr(W_{\max}>0.5)$, is generally larger. Both differ from the largest observed share among a finite set of $J$ units. We use the size-biased tail for calibration and report the largest-weight diagnostic separately.

\subsection{The Gamma-Mixed Size-Biased Distribution}
\label{subsec:E3-marginal}

Because $w_1\mid\alpha\sim\mathrm{Beta}(1,\alpha)$, its conditional tail is $(1-t)^\alpha$. Integrating this tail by the Gamma Laplace transform gives
\begin{equation}
S(t;a,b):=\Pr(W_{\mathrm{SB}}>t\mid a,b)
=\E_\alpha\{e^{-c_t\alpha}\}
=\left(\frac{b}{b+c_t}\right)^a.
\label{eq:E9}
\end{equation}
Here $c_t=-\log(1-t)$ and $0<t<1$. The CDF is $1-S(t;a,b)$,
and differentiation gives the density
\begin{equation}
f_{\mathrm{SB}}(w\mid a,b)
=\frac{ab^a}{(1-w)\{b-\log(1-w)\}^{a+1}},\qquad0<w<1.
\label{eq:E7}
\end{equation}
\citet[Appendix~A]{vicentini2025prior} derive this density and CDF. We use~\eqref{eq:E9} at $t=0.5$ and $0.9$ to report majority and near-universal mass probabilities. No simulation or quadrature is needed for these tails.

For optimization, the derivatives are
\begin{equation}
\partial_aS=S\log\frac{b}{b+c_t},\qquad
\partial_bS=S\frac{ac_t}{b(b+c_t)}.
\label{eq:E14}
\end{equation}
At fixed $b$, increasing $a$ lowers the tail; at fixed $a$, increasing $b$ shifts $\alpha$ toward zero and raises it.

\subsection{Bounds on the Largest Weight: Proof of Main-Text Proposition~\ref{prop:largest-weight}}
\label{subsec:E2b-wmax}

\begin{proposition}[Bounds on the largest weight]
\label{prop:wmax-bounds}
Let $(p_h)_{h\ge1}$ be the ranked masses of a random discrete probability measure and $W_{\max}=p_1$. Draw an index $H$ with $\Pr(H=h\mid(p_\ell))=p_h$ and set $W_{\mathrm{SB}}=p_H$. For $1/2\le t<1$,
\[
\Pr(W_{\mathrm{SB}}>t)\le\Pr(W_{\max}>t)
\le\frac{\Pr(W_{\mathrm{SB}}>t)}{t}.
\]
For the Dirichlet process these bounds hold conditionally on $\alpha$ and after averaging over any hyperprior on $\alpha$.
\end{proposition}

\begin{proof}
At most one mass can exceed $t\ge1/2$. Conditioning on the ranked masses therefore gives
\[
\Pr(W_{\mathrm{SB}}>t)
=\E\left\{\sum_{h\ge1}p_h\mathbf1(p_h>t)\right\}
=\E\{W_{\max}\mathbf1(W_{\max}>t)\}.
\]
On this event, $t<W_{\max}\le1$, so the final expectation lies between
$t\Pr(W_{\max}>t)$ and $\Pr(W_{\max}>t)$. Rearranging proves the result. The same argument can be conditioned on $\alpha$; integration then preserves the inequalities.
\end{proof}

At $t=0.9$ the upper bound is within a factor $1/0.9$ of the lower bound; at $t=0.5$ the factor is two. We cap the upper bound at one. The same conditioning argument also gives an exact integral for $t\ge1/2$:
\begin{equation}
\Pr(W_{\max}>t)
=\E\{W_{\mathrm{SB}}^{-1}\mathbf1(W_{\mathrm{SB}}>t)\}
=\int_t^1\frac{f_{\mathrm{SB}}(w)}{w}\,dw.
\label{eq:E-wmax-exact}
\end{equation}
Indeed, the conditional expectation of the reciprocal-weight expression is $\sum_h\mathbf1(p_h>t)$, which equals the indicator that the largest mass exceeds $t$. Substituting the Beta density gives the conditional probability; substituting~\eqref{eq:E7} gives the Gamma-mixed probability by one-dimensional integration. This establishes the calculation stated in the main text without identifying $w_1$ with the largest weight.

\subsection{Expected Co-Clustering Probability}
\label{subsec:E4-coclustering}

For two independent units drawn from the population, the probability of belonging to the same component conditional on the weights is
\begin{equation}
\rho=\sum_{h\ge1}w_h^2.
\label{eq:E16}
\end{equation}
By the predictive rule for the second unit, $\E(\rho\mid\alpha)=1/(1+\alpha)$. Hence the optional diagnostic in the main text is
\begin{equation}
\E(\rho\mid a,b)=\frac{b^a}{\Gamma(a)}
\int_0^\infty\frac{\alpha^{a-1}e^{-b\alpha}}{1+\alpha}\,d\alpha.
\label{eq:E23}
\end{equation}
We evaluate this integral on the same Gamma quadrature nodes as the count moments. Its value also equals $\E(W_{\mathrm{SB}}\mid a,b)$, but a mean co-clustering judgment does not determine the weight tail targeted below.

\subsection{The Dual-Anchor Objective and Selection Rule}
\label{subsec:E5-dual-anchor}

The count-calibrated prior $(a_K,b_K)$ matches two moments using two Gamma parameters. If its implied majority probability is unacceptable, reducing that probability generally changes the count moments. Dual-Anchor calibration makes this trade-off explicit. For fixed scales $s_\mu=\max(|\mu_K|,1)$ and $s_v=\max(\sigma_K^2,1)$, define
\begin{equation}
\begin{aligned}
D_K(a,b)&=\left\{\frac{M_1(a,b)-\mu_K}{s_\mu}\right\}^2
+\left\{\frac{V(a,b)-\sigma_K^2}{s_v}\right\}^2,\\
D_w(a,b)&=\{S(t;a,b)-\delta\}^2,
\end{aligned}
\label{eq:E31}
\end{equation}
and minimize
\begin{equation}
\mathcal L_\lambda(a,b)=\lambda D_K(a,b)+(1-\lambda)D_w(a,b),
\qquad0<\lambda\le1.
\label{eq:E32}
\end{equation}
The scales depend only on the stated targets and remain fixed across the grid. The parameter $\lambda$ weights the discrepancies; it does not impose the weight target as a constraint.

We fix the policy before posterior fitting: threshold $t=0.5$, trigger $\eta=0.40$, soft target $\delta=0.25$, and
$\Lambda=\{0.01,\ldots,0.20,0.25,0.30,0.40,\ldots,0.90,1.00\}$.
The trigger determines whether refinement is used and which fitted grid points are eligible. The distinct soft target enters the loss. The procedure is:

\begin{enumerate}[leftmargin=2em,itemsep=5pt,topsep=4pt]
\item Retain the verified count-calibrated prior if $S(0.5;a_K,b_K)\le0.40$.
\item Otherwise, minimize~\eqref{eq:E32} at all 29 grid values in a bounded log-parameter domain. Use bounded L-BFGS-B, with the documented bounded Nelder--Mead fallback if convergence or fresh-objective checks fail. A verified solution at the preceding larger $\lambda$ may provide a warm start. Keep every grid point in the diagnostic output, including unsuccessful fits.
\item For each candidate, recompute the count moments at a quadrature order at least twice the fitting order. Evaluate the size-biased tail from~\eqref{eq:E9}, then recompute the objective using these values. Among verified finite interior solutions with $S(0.5;a,b)\le0.40$, select the largest $\lambda$. Exclude unverified or approximate solutions, solver-boundary fits and hyperpriors close to a point mass. If no candidate is eligible, report the unresolved conflict and fit no model.
\item Separately minimize $D_K$ subject to $S(0.5;a,b)\le0.25$. Report this hard-bound sensitivity with its constraint residual, constraint activity, solver-boundary status and $\mathrm{CV}(\alpha)=1/\sqrt a$. A coefficient of variation below 0.01 is flagged as close to a point mass.
\end{enumerate}

A selected soft solution can pass the 0.40 eligibility rule while missing the 0.25 target. The hard analysis enforces the inequality, possibly at a larger cost to count fit. An active constraint is distinct from reaching the edge of the solver domain. The numerical search does not establish a global optimum; reporting the complete grid shows the available verified trade-offs. Appendix~\ref{app:software} gives the worked frontier and all 29 grid values.

\subsection{Reporting the Specification}
\label{subsec:E6-guidelines}
\label{subsubsec:E62-reporting}

Use this worksheet to document the choices in Box~\ref{box:recipe} for each prior. Before posterior fitting, record the inputs, decision rules and calibration results. After fitting, add the posterior comparisons and record when and why any choices changed.
\begin{enumerate}[leftmargin=2em,itemsep=4pt,topsep=4pt]
\item \textbf{Target, evidence and timing.} State the units, outcome scale, likelihood and fixed design size $J$. Define $K_J$ as the cluster count among these units. Give the basis for the count judgment (external evidence, documented expert elicitation or the authors' own illustrative judgment) and say which data or posterior results had already been seen. These clusters need not correspond to substantive groups.
\item \textbf{Inputs and their implications.} Give the supplied mean and either the interval with its stated probability or the VIF/CV. Explain how these give $(\mu_K,\sigma_K^2)$. For the fitted prior, report the count mean, variance, median, central interval and probability assigned to the stated interval. Indicate whether the interval width supplied a variance approximation or the stated probability was imposed when converting the judgment to moments. A confidence label or a wider variance alone cannot justify the prior. Examine its weight tails and the priors on the component parameters.
\item \textbf{Calibration and verification.} Give the Gamma shape--rate parameters, software version, fitting and verification quadrature orders, tolerance and verification result. Include failed calibrations and incompatible targets. If the probability of the stated interval is unavailable, report it as missing. A probability reported for a different central interval cannot be substituted.
\item \textbf{Weight policy and prior roles.} Give the size-biased tail probabilities at 0.5 and 0.9 and the corresponding bounds for the largest weight. The expected co-clustering probability may also be of interest. If a prior was refined, give the trigger, soft target, grid, selection rule and selected $\lambda$. Give the changes in the count moments and say whether the soft target was attained. Retain the full grid with the convergence, usability, verification and boundary fields. For the separate hard-bound analysis, give the constraint residual, whether the constraint was active, the Gamma coefficient of variation and the solver-boundary classification.
\item \textbf{Scenario set.} Explain the choice of a few alternative counts, degrees of uncertainty or weight policies. Identify any stress test for which substantive plausibility has yet to be established. Choose the comparison rule before inspecting the posterior partitions. Avoid selecting a prior for the partition it produces, and report any changes made after inspection.
\item \textbf{Posterior reporting.} Identify the analysis code. For each of the partition, distributional and unit-level summaries, report uncertainty and diagnostics for the functional being compared. Distinguish the finite-$J$ distribution from the population density. For each, distinguish the posterior mean from posterior uncertainty about the distribution itself. Include convergence problems and limits on Monte Carlo precision. State which conclusions change across scenarios. Similar mean curves or a larger cluster count alone cannot establish posterior equivalence or improved recovery of heterogeneity.
\end{enumerate}

\paragraph{Application records.}
We chose the count inputs in the three applications for illustration. For STAR ($J=79$), the supplied mean is 5 and the stated 80\% interval is $[2,8]$; for Bangert--Drowns ($J=48$), they are 4 and $[1,8]$; for the Rasch example ($J=500$), they are 5 and $[2,10]$. The normal interval-width proxy gives target variances 5.480, 7.459 and 9.742, respectively. We supplied the means separately. In the Bangert--Drowns and Rasch examples they differ from the interval midpoints. The Gamma-mixed count distribution need not assign 80\% probability to the stated interval.

Table~\ref{tab:E-elicitation-worksheet} records the achieved count moments and stated-interval probabilities. TSMM matches the count moments. Dual-Anchor allows the count moments to depart from their targets to reduce the size-biased majority probability. In a separate sensitivity analysis, Dual-Anchor (hard) places an upper bound on this probability. For illustration, we used 0.40 as the trigger for refinement and 0.25 as the soft target. Only the hard analysis imposes the prior bound $\Pr(W_{\mathrm{SB}}>0.5)\leq0.25$. These choices have no universal status and give no bound on posterior dominance. The selected soft $\lambda$ values of 0.20, 0.13 and 0.05 gave majority-tail probabilities of 0.385, 0.395 and 0.388, respectively. All three fell below the trigger and remained above the soft target. The STAR and Rasch hard results approach a point-mass limit near $\alpha=2$; the Bangert--Drowns hard result is a finite Gamma solution with an active constraint.

\begin{table}[tbp]
\centering
\small
\caption{Application worksheet: achieved count summaries under the three calibrated prior roles.}
\label{tab:E-elicitation-worksheet}
\begin{tabular}{@{}llrrr@{}}
\toprule
Application & Role & Mean & Variance & Stated-interval mass \\
\midrule
STAR & TSMM & 5.000 & 5.480 & 0.884 \\
 & Dual-Anchor & 6.237 & 5.370 & 0.833 \\
 & Dual-Anchor (hard) & 7.931 & 5.401 & 0.614 \\
Bangert--Drowns & TSMM & 4.000 & 7.459 & 0.927 \\
 & Dual-Anchor & 5.658 & 7.008 & 0.856 \\
 & Dual-Anchor (hard) & 7.154 & 6.346 & 0.720 \\
Rasch & TSMM & 5.000 & 9.742 & 0.849 \\
 & Dual-Anchor & 8.752 & 8.986 & 0.736 \\
 & Dual-Anchor (hard) & 11.590 & 9.018 & 0.373 \\
\bottomrule
\end{tabular}
\begin{minipage}{0.98\linewidth}
\footnotesize
Stated intervals are $[2,8]$, $[1,8]$ and $[2,10]$ for STAR, Bangert--Drowns and Rasch, respectively. Each probability is the mass that the calibrated prior assigns to the stated interval. It need not equal 0.80. We used the original calibrated parameters at full precision, without recalibration or refitting. The count PMFs from \texttt{DPprior} 2.0.0 at 160 and 320 quadrature nodes agreed with those from an independent adaptive integration.
\end{minipage}
\end{table}

For STAR and Bangert--Drowns, we fixed the weight policy before posterior simulation. The records do not show whether we had seen the data when we chose the count inputs. For Rasch, we chose the count inputs after inspecting the data, selected Dual-Anchor from its prior count and weight summaries before fitting, and retained it after fitting.

We used a separate stress test for the Rasch example. A count mean of 25 and a VIF of 2.5 give a target variance of 60. We specified no interval probability. The fitted prior gives a mean of 25, a variance of 60 and a majority-tail probability of 0.042. This tail is below the trigger for soft refinement and the hard bound is inactive. TSMM, Dual-Anchor and Dual-Anchor (hard) thus give the same prior to numerical tolerance. A substantive interpretation of the count of 25 would require further evidence, as would the $\times0.5$ and $\times2$ perturbations in the simulation. We have no independent record of when the count was chosen relative to inspection of the data.

\noindent\textit{Illustrative example statement.}
\begin{quote}
For 50 sites, we expected five clusters with moderate uncertainty, giving target mean 5 and variance 10. The TSMM shape and rate were 1.4082 and 1.0770, with median 4 and a 90 percent count interval of $[1,11]$. The size-biased tails at 0.5 and 0.9 were 0.497 and 0.200. The majority probability exceeded our prespecified trigger of 0.40. Over the fixed grid, the largest eligible value was $\lambda=0.30$, giving the Dual-Anchor prior $\mathrm{Gamma}(2.3158,1.4204)$. Its count mean and variance were 5.884 and 9.738, and its tails were 0.398 and 0.107. It passed the trigger rule but did not attain the 0.25 soft target. The separate hard-bound analysis gave $\mathrm{Gamma}(5.2745,2.3059)$, with count mean 7.488, variance 9.051 and majority probability 0.250 at the active constraint. We used \texttt{DPprior} 2.0.0, fitting with 160 quadrature nodes and verifying with 320. The full grid and diagnostics are reported in Appendix~\ref{app:software}.
\end{quote}

The two size-biased tails and largest-weight bounds remain useful when no refinement is triggered: they show the weight concentration implied by the count judgment \citep{vicentini2025prior}.

\section{Extended Simulation Study Results}
\label{app:simulation}

\subsection{Overview}
\label{subsec:F1-purpose}

This appendix gives the design, convergence criteria, and results supporting
Section~\ref{sec:simulation}. We planned 200 replications for each of 322 conditions. The 64{,}400 planned fits were divided into two batches of 50{,}200 and 14{,}200. The conditions form targeted comparisons rather than a
full factorial design. Fits under different priors use the same generated
data, with distinct chain seeds. The estimation error is comparable to or larger than the spacing between generating component means. We therefore use $K^*$ as a reference count without assuming that the components can be recovered.

\subsection{Simulation Design}
\label{subsec:F2-design}

\subsubsection{Design Grid}
\label{subsubsec:F21-grid}

\Cref{tab:F1-design} lists the simulation suites.

\begin{table}[!htbp]
\centering
\caption{Simulation Suites}
\label{tab:F1-design}
\small
\begin{tabular}{@{}lp{0.50\textwidth}rr@{}}
\toprule
Suite & Purpose & Cells & Fits \\
\midrule
Initial set & Prespecified count and weight sensitivity comparisons and design extensions under $N(0,100)$ & 145 & 29{,}000 \\
Matched $N(0,0.01)$ & The 45 conditions for the grand-mean-prior comparison & 45 & 9{,}000 \\
Main $N(0,1)$ & All conditions in Figure~\ref{fig:3}, including the 15 low-information conditions in Table~\ref{tab:sim-results} & 45 & 9{,}000 \\
Central dominant & Location sensitivity with an exactly 60\% component & 16 & 3{,}200 \\
Additional cells & Remaining misspecification, design-size and rotating-dominant-location comparisons under $N(0,100)$ & 71 & 14{,}200 \\
\addlinespace
Total & All planned conditions & 322 & 64{,}400 \\
\bottomrule
\end{tabular}
\end{table}

The comparisons cover $K^*\in\{5,10,30\}$,
$I\in\{0.2,0.5,0.8\}$, $J\in\{25,50,100,200\}$ where specified, and five priors: Vague, TSMM, the selected
Dual-Anchor policy, SSI-Weight, and DORO-Unif as an
\emph{external-reference} projection. Its discrepancy exceeds the calibration tolerance, so its results are
reported separately from successfully calibrated priors.

\subsubsection{Data-Generating Process (DGP)}
\label{subsubsec:F22-dgp}

We simulate multisite effect estimates using the normal sampling model of \citet{rubin1981estimation}:
\begin{equation}
\hat\tau_j \mid \tau_j, \widehat{se}_j^2 \sim \mathcal{N}(\tau_j, \widehat{se}_j^2),
\quad j=1,\ldots,J.
\label{eq:F1}
\end{equation}
We generate true site effects from a finite mixture with exactly $K^*$ components:
\begin{equation}
\tau_j \mid z_j \sim \mathcal{N}(\mu_{z_j}, s_{\text{within}}^2),
\qquad z_j \in \{1,\ldots,K^*\},
\label{eq:F2}
\end{equation}
where the component means are evenly spaced on $[-d, d]$ (with $d=2.5$). Within each replication, we centered the generated site effects and rescaled them to a sample SD of 0.25, using denominator $J-1$. Specifically, we used $\tau_j^{\text{final}} = \sigma_\tau \cdot (\tau_j^{\text{raw}} - \bar{\tau}^{\text{raw}})/\mathrm{SD}(\tau^{\text{raw}})$.
Their finite-site SD, using denominator $J$, was therefore $0.25\sqrt{(J-1)/J}$, or 0.24875 in the 100-site design.

For the equal-weight design, we assign sites to the generating components so that their allocation counts differ by at most one, then shuffle site labels. For the dominant design, exactly $0.60J$ sites occupy one component and the remainder are divided as evenly as possible among the other generating components. The primary dominant location is balanced across component positions in blocks of $K^*$ replications; the separate location sensitivity fixes the dominant component at the central component mean.

\subsubsection{Fixed Design Parameters}
\label{subsubsec:F23-fixed}

The remaining generating parameters are in \Cref{tab:F2-fixed}.

\begin{table}[H]
\centering
\caption{Fixed DGP and Calibration Parameters}
\label{tab:F2-fixed}
\small
\begin{tabular}{@{}p{0.17\textwidth}p{0.15\textwidth}p{0.20\textwidth}p{0.38\textwidth}@{}}
\toprule
\raggedright Component & \raggedright Parameter & \raggedright Value & \raggedright Role \tabularnewline
\midrule
\raggedright Number of sites & \raggedright $J$ & \raggedright 25, 50, 100, 200 & \raggedright Fixed at 100 in the main comparison and varied in the design-size comparison. \tabularnewline
\raggedright Between-site SD & \raggedright $\sigma_\tau$ & \raggedright 0.25 & \raggedright Sample SD of the generated site effects (denominator $J-1$) \tabularnewline
\raggedright Within-component SD & \raggedright $s_{\text{within}}$ & \raggedright 0.15 & \raggedright Controls component overlap \tabularnewline
\raggedright Component separation span & \raggedright $d$ & \raggedright 2.5 & \raggedright Means span $[-d, d]$ \tabularnewline
\raggedright SE scaling & \raggedright $\kappa$ & \raggedright 4 & \raggedright Defines $\widehat{se}_j^2 = \kappa / n_j$ \tabularnewline
\raggedright Gamma draw & \raggedright $\mathrm{CV}(X_j)$ & \raggedright 0.5 & \raggedright Before applying the minimum and rounding \tabularnewline
\raggedright Minimum site size & \raggedright $n_{\min}$ & \raggedright 5 & \raggedright Avoids degenerate SEs \tabularnewline
\raggedright Component allocation & \raggedright --- & \raggedright exact integer counts & \raggedright Equal or exactly 60\% dominant, as declared by the cell \tabularnewline
\bottomrule
\end{tabular}
\end{table}

\subsubsection{Informativeness Index and Calibration Targets}
\label{subsubsec:F24-informativeness}

Following \citet{lee2025improving}, we define informativeness as
\begin{equation}
I \equiv \frac{\sigma_\tau^2}{\sigma_\tau^2 + \mathrm{GM}(\widehat{se}_j^2)},
\qquad \mathrm{GM}(\widehat{se}_j^2)=\exp\!\left(\frac{1}{J}\sum_{j=1}^J\log \widehat{se}_j^2\right).
\label{eq:F3}
\end{equation}
To generate the unit sizes, we draw $X_j$ from a Gamma distribution with mean $\bar n$ and coefficient of variation $c=0.5$ (shape $1/c^2$, rate $1/(c^2\bar n)$). We then set $n_j=\mathrm{round}(\max\{5,X_j\})$ and $\widehat{se}_j^2=4/n_j$. To attain each informativeness level in $\{0.2,0.5,0.8\}$, we vary the Gamma mean $\bar n$ over a grid, estimating $\E[I]$ from 2,000 simulated designs at each value. We stop when the estimate is within 0.01 of the target.

\Cref{tab:F3-informativeness} reports the calibration targets used in the completed run.

\begin{table}[H]
\centering
\caption{Informativeness Calibration Inputs}
\label{tab:F3-informativeness}
\small
\begin{tabular}{@{}rrl@{}}
\toprule
Target $I$ & Calibrated Gamma mean $\bar n$ & Status \\
\midrule
0.2 & 17.401486  & Frozen input \\
0.5 & 72.897896  & Frozen input \\
0.8 & 278.423775 & Frozen input \\
\bottomrule
\end{tabular}
\end{table}

\subsubsection{Replications and Reproducibility}
\label{subsubsec:F25-seeds}

We use common generated data for paired prior comparisons and distinct
random seeds for each chain and rerun. The code records these seeds so that
each replication can be reproduced independently.

\subsubsection{Extended Designs}
\label{subsubsec:F26-extensions}

We extend the design in four ways. Section~\ref{subsec:F11-extended}
reports the full comparisons and the number of fits retained in each cell.
\begin{enumerate}[nosep, label=(\alph*)]
\item \textit{Elicited-mean misspecification}: $r=\mu_K/K^*\in\{0.5,1,2\}$. The grid holds 56 of its 60 conceptual cells; the remaining four (Dual-Anchor and DORO-Unif at $\mu_K=2.5$) were rejected because calibration failed.
\item \textit{Design size}: $J\in\{25,50,100,200\}$. The five-prior grid covers all 80 prior-by-$K^*$-by-$I$-by-$J$ combinations, giving 20 complete trajectories.
\item \textit{Dominant component}: exactly 60 percent of sites, with balanced-random and central component-location rules. The four-prior location grid covers all 32 cells, so every condition has both location rules.
\item \textit{Trade-off weight}: the selected Dual-Anchor policy is compared with fixed $\lambda\in\{0.5,0.9\}$ for $K^*=5$, $I\in\{0.2,0.8\}$, and $N(0,100)$.
\end{enumerate}
The fixed-$\lambda$ cells are sensitivity analyses; they do not replace the
selected policy.

Calibration rejected 14 planned cells (2{,}800 fits). For Dual-Anchor at
$\mu_K=2.5$, the high-confidence search reached the solver boundary and the
medium-confidence search approached a point-mass limit (three
informativeness levels each). At $\mu_K=60$ and low confidence, count
calibration returned approximate or unverified results for TSMM and
Dual-Anchor (six cells). DORO-Unif at $\mu_K=2.5$ failed multistart
verification at $I=0.2$ and $0.8$ (two cells). These are failures of the
specified calibration procedure, not proofs that no suitable prior exists.
The four missing cells in the medium-confidence misspecification comparison
are the Dual-Anchor and DORO-Unif cells at $\mu_K=2.5$ and $I\in\{0.2,0.8\}$.

\subsection{Hyperprior Methods and Exact Parameters}
\label{subsec:F3-hyperpriors}

\subsubsection{Methods Compared}
\label{subsubsec:F31-methods}

We compare four primary $\alpha\sim\mathrm{Gamma}(a,b)$ hyperpriors:
\begin{enumerate}[label=(\arabic*),nosep,leftmargin=2em]
\item \textit{Vague}: the fixed $\mathrm{Gamma}(1,1)$ baseline;
\item \textit{TSMM}: A2-MN calibration to $(\mu_K,\sigma_K^2)$;
\item \textit{Dual-Anchor}: the fixed-grid solution when the TSMM majority tail
exceeds 0.40, and TSMM otherwise; and
\item \textit{SSI-Weight}: the $J$-invariant solution to
$\Pr(W_{\mathrm{SB}}>0.5)=0.25$ and $\Pr(W_{\mathrm{SB}}>0.9)=0.05$.
\end{enumerate}
We verify each solution independently. The hard 0.25 inequality is a separate
sensitivity analysis.

DORO-Unif projects onto a uniform count target following
\citet{dorazio2009selecting}. For the retained fits, the projection is stable
across starting values but fails our KL-discrepancy and count-distribution checks. We therefore use
it only as an external reference.

\subsubsection{Confidence-to-Variance Mapping (TSMM / Dual-Anchor)}
\label{subsubsec:F32-vif}

We use a variance inflation factor (VIF) mapping:
\begin{equation}
\sigma_K^2 \equiv \Var(K_J) = \mathrm{VIF}\times(\mu_K-1),
\qquad \mathrm{VIF}\in\{1.5,\;2.5,\;5.0\}
\label{eq:F4}
\end{equation}
corresponding to high, medium, and low confidence, respectively. The main-text results use medium confidence unless stated otherwise.

\subsubsection{Exact \texorpdfstring{$(a,b)$}{(a,b)} Hyperparameters Used}
\label{subsubsec:F33-params}

\Cref{tab:F4-params} gives the fitted hyperparameters, count means and
size-biased majority probabilities. The selected trade-off weights are
shown for each confidence level.

\begin{table}[!htbp]
\centering
\caption{Calibrated Priors and External References at $J=100$}
\label{tab:F4-params}
\small
\setlength{\tabcolsep}{3pt}
\begin{tabular}{@{}lrrrlrrr@{}}
\toprule
Method & VIF & $\mu_K$ & $\sigma_K^2$ & $(a,b)$ & $\E(K_J)$ &
$P(W_{\mathrm{SB}}{>}.5)$ & $\lambda$ \\
\midrule
Vague & --- & --- & --- & $(1.000,\,1.000)$ & 4.84 & 0.591 & --- \\
SSI-Weight & --- & --- & --- & $(2.179,\,0.780)$ & 10.05 & 0.250 & --- \\
TSMM & 1.5 & 5.0 & 6 & $(4.351,\,4.471)$ & 5.00 & 0.534 & --- \\
Dual-Anchor & 1.5 & 5.0 & 6 & $(13.733,\,10.045)$ & 6.39 & 0.400 & 0.19 \\
TSMM & 2.5 & 5.0 & 10 & $(1.666,\,1.646)$ & 5.00 & 0.557 & --- \\
Dual-Anchor & 2.5 & 5.0 & 10 & $(4.232,\,2.859)$ & 6.65 & 0.399 & 0.16 \\
TSMM & 5.0 & 5.0 & 20 & $(0.623,\,0.561)$ & 5.00 & 0.606 & --- \\
Dual-Anchor & 5.0 & 5.0 & 20 & $(1.841,\,1.070)$ & 7.17 & 0.399 & 0.12 \\
TSMM & 1.5 & 10.0 & 13.5 & $(7.185,\,2.728)$ & 10.00 & 0.197 & --- \\
TSMM & 2.5 & 10.0 & 22.5 & $(2.992,\,1.100)$ & 10.00 & 0.232 & --- \\
TSMM & 5.0 & 10.0 & 45 & $(1.179,\,0.400)$ & 10.00 & 0.306 & --- \\
TSMM & 1.5 & 30.0 & 43.5 & $(10.931,\,0.750)$ & 30.00 & 0.001 & --- \\
TSMM & 2.5 & 30.0 & 72.5 & $(5.178,\,0.344)$ & 30.00 & 0.003 & --- \\
TSMM & 5.0 & 30.0 & 145 & $(2.198,\,0.135)$ & 30.00 & 0.018 & --- \\
DORO-Unif & --- & 5.0 & --- & $(2.434,\,2.461)$ & 4.99 & 0.547 & --- \\
DORO-Unif & --- & 10.0 & --- & $(1.531,\,0.543)$ & 9.91 & 0.284 & --- \\
DORO-Unif & --- & 30.0 & --- & $(0.943,\,0.050)$ & 29.50 & 0.078 & --- \\

\bottomrule
\end{tabular}
\par\medskip
\begin{minipage}{\linewidth}\normalsize \textit{Notes.} The soft row is selected only when verified TSMM has
$\Pr(W_{\mathrm{SB}}>0.5)>0.40$; it is the largest verified interior
$\lambda$ with tail at most 0.40. Thus the displayed $J=100,\mu_K=5$
values use $\lambda=0.19,0.16,0.12$ for VIF $1.5,2.5,5$, respectively.
For $\mu_K\in\{10,30\}$ the trigger is inactive, so Dual-Anchor equals TSMM
and duplicate rows are omitted. DORO-Unif fails the calibration checks
and remains an external reference.\end{minipage}
\end{table}

\subsection{Inference Model and MCMC Configuration}
\label{subsec:F4-mcmc}

In each replication we fit a DP mixture model with a CRP prior on cluster assignments:
\begin{equation}
\hat\tau_j \sim \mathcal{N}(\tau_j, \widehat{se}_j^2),\qquad
\tau_j \sim \mathcal{N}(\theta_{z_j}, \sigma^2_{z_j}),\qquad
z \sim \mathrm{CRP}(\alpha).
\label{eq:F5}
\end{equation}
Let $v_\theta$ denote the variance of the component means in the fitted model. We use $\theta_k \sim \mathcal{N}(\tau,v_\theta)$ for the means, $\sigma_k^2 \sim \mathrm{InvGamma}(s_1,s_2)$ for the within-component variances, and the hyperpriors
\begin{equation}
\alpha \sim \mathrm{Gamma}(a,b),\quad
\tau \sim \mathcal{N}(0,\, v_\tau),\quad
v_\theta \sim \mathrm{Unif}(0,10),\quad
s_1, s_2 \sim \mathrm{InvGamma}(2,1),
\label{eq:F6}
\end{equation}
where $v_\tau=1$ in the main comparison and $v_\tau\in\{0.01,100\}$
in the paired sensitivities. The likelihood and remaining priors are fixed.
We pass the variance explicitly to NIMBLE. In the original simulation,
\texttt{dnorm(0,100)} specified precision 100, hence variance 0.01,
although the manuscript described variance 100. We corrected this discrepancy
and recomputed all three branches; affected earlier summaries are not reused.

We fit each model using four overdispersed chains in NIMBLE
\citep{devalpine2017programming}. Thinning is 1. Attempt~1 runs 12{,}000
iterations per chain and discards the first 4{,}000. A fit that fails the
prespecified convergence criteria is rerun once, from new chain seeds, for 120{,}000 iterations per chain with the first 60{,}000
discarded. We allow no third attempt. The convergence criteria require
$\widehat R\leq1.01$, bulk ESS at least 400, and tail ESS at least 400 for all
11 monitored variables, including $\alpha$, $K_J$, the grand mean, base-scale
parameters, active-atom summaries, and outcome summaries.
These diagnostics follow \citet{vehtari2021rank}.
For each MCMC draw $t$, the cluster count is computed as
\begin{equation}
K_J^{(t)} = \left|\{z_1^{(t)},\ldots,z_J^{(t)}\}\right|.
\label{eq:F7}
\end{equation}
Posterior summaries of $K_J$ (mean, mode, quantiles) are computed from $\{K_J^{(t)}\}$.

\subsection{Outcome Metrics}
\label{subsec:F5-metrics}

For each accepted fit, we retain the posterior mean and standard deviation
of $K_J$, its 2.5th and 97.5th percentiles,
$\Pr(K_J=1\mid\text{data})$, $\Pr(K_J\leq2\mid\text{data})$, the posterior
mean of $\alpha$, aggregate site-effect RMSE, and coverage of $K^*$ by the
reported interval. Cell summaries give between-replication variation, bias
relative to $K^*$, and the numbers of accepted, runnable, and planned fits.
Section~\ref{subsec:F-finite-site-spread} compares the estimated SD, CDF and individual effects for the 100 simulated sites with their known values.

\subsection{Completed Fits and Convergence}
\label{subsec:F6-results}

We retain fits that pass the criteria in Section~\ref{subsec:F4-mcmc}.
Of 64{,}400 planned fits, 2{,}800 were not attempted because prior calibration
failed, and 4{,}449 failed the final convergence check. The remaining 57{,}151
fits enter the summaries. Table~\ref{tab:F6-population} distinguishes these
exclusions and gives the per-cell denominators for the main comparison.

\begin{table}[!htbp]
\centering
\caption{Simulation fit accounting and main-comparison denominators}
\label{tab:F6-population}
\small
\begin{tabular}{@{}lr@{}}
\toprule
Disposition & Fits \\
\midrule
Planned & 64{,}400 \\
Calibration failure; not attempted & 2{,}800 \\
Final convergence failure & 4{,}449 \\
Converged; included in summaries & 57{,}151 \\
\bottomrule
\end{tabular}

\medskip
\noindent Main comparison: converged fits out of 200 attempts under $N(0,1)$.

\medskip
\begin{tabular}{@{}rrccccc@{}}
\toprule
$K^*$ & $I$ & Vague & DORO-Unif & TSMM & Dual-Anchor & SSI-Weight \\
\midrule
5 & 0.2 & 200/200 & 200/200 & 200/200 & 200/200 & 200/200 \\
5 & 0.5 & 200/200 & 200/200 & 200/200 & 200/200 & 200/200 \\
5 & 0.8 & 183/200 & 182/200 & 185/200 & 185/200 & 182/200 \\
10 & 0.2 & 200/200 & 200/200 & 200/200 & 200/200 & 200/200 \\
10 & 0.5 & 200/200 & 200/200 & 199/200 & 200/200 & 200/200 \\
10 & 0.8 & 185/200 & 178/200 & 188/200 & 183/200 & 183/200 \\
30 & 0.2 & 200/200 & 195/200 & 167/200 & 169/200 & 200/200 \\
30 & 0.5 & 200/200 & 189/200 & 196/200 & 193/200 & 199/200 \\
30 & 0.8 & 188/200 & 155/200 & 168/200 & 170/200 & 185/200 \\

\bottomrule
\end{tabular}

\medskip
\parbox{\textwidth}{\small\textit{Note.} The upper panel covers all simulation
conditions. The lower panel gives the number of converged fits per prior in the main simulation (8{,}607 of 9{,}000 attempts). Main-text Table~\ref{tab:sim-results} uses the low-information rows. Figure~\ref{fig:3}A and~B use the replications in which all five priors converged (124--200 per design). Figure~\ref{fig:3}C uses the replications in which both $\mathrm{Gamma}(1,1)$ and Dual-Anchor converged (165--200 per design). Tables in
Section~\ref{subsec:F11-extended} give the denominators for the extended
comparisons. Calibration failures are separate from convergence failures.}

\end{table}

The two batches contribute 43{,}559 and 13{,}592 converged fits. The
convergence failure rates among attempted fits are 8.5 and 2.9 percent,
respectively; the second batch concentrates on $K^*\leq10$. All reported
tables retain the convergence exclusions and give the number of fits
behind each result.

\subsection{Primary Results and Sensitivity to the Grand-Mean Prior}
\label{subsec:F12-original}

We compared the three grand-mean priors in pairs, using the same simulated data for both fits. Each pairwise comparison included at least 111 replications in which both fits converged in each of the 45 conditions (Figure~\ref{fig:F-grandmean}). For the primary $N(0,1)$ analysis, Table~\ref{tab:F-main-full} adds the bias and RMSE of the estimated cluster count to the low-information results in Table~\ref{tab:sim-results}.

\begin{figure}[!htbp]
\centering
\includegraphics[width=\textwidth]{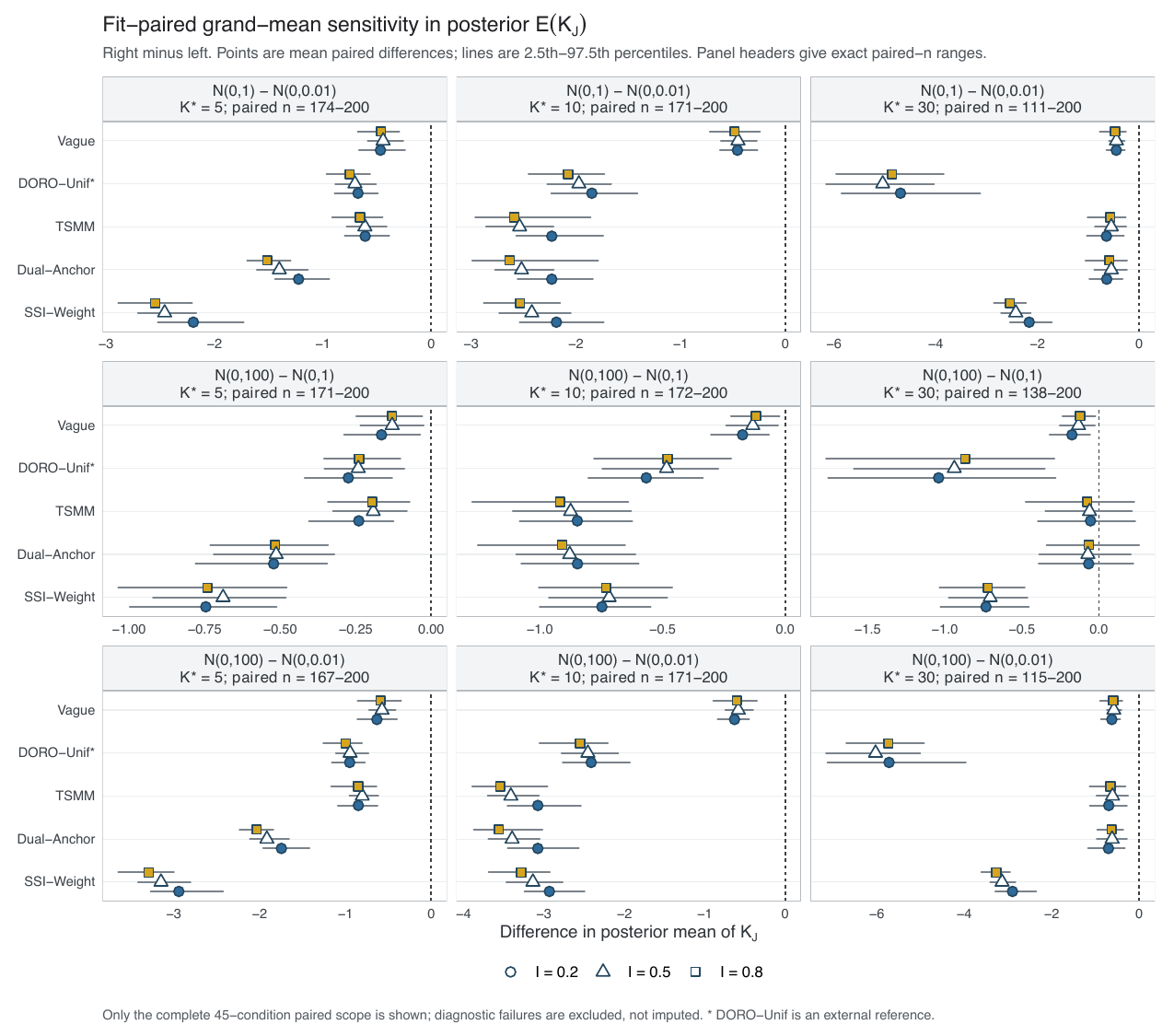}
\caption{Paired changes in posterior mean $K_J$ for all 45 main-comparison
conditions. Columns distinguish $K^*=5,10,30$; rows show the three
changes in grand-mean prior. Point shapes distinguish $I=0.2,0.5,0.8$.
Each point is the mean paired difference, with its 2.5--97.5 percentile
range over shared replications (at least 111 pairs per condition).
Panel headings specify the direction of each contrast.}
\label{fig:F-grandmean}
\end{figure}

Averaged over the 45 conditions, changing the grand-mean prior from
$N(0,0.01)$ to $N(0,1)$ changes posterior mean $K_J$ by $-1.61$ and posterior
$\Pr(K_J=1\mid\text{data})$ by $+0.157$. Changing it from $N(0,0.01)$ to
$N(0,100)$ gives corresponding averages of $-2.08$ and $+0.249$. All 45
condition-specific posterior-mean-$K_J$ contrasts have those directions.
Average paired changes in aggregate site-effect RMSE are approximately
$2.0\times10^{-4}$ and $1.2\times10^{-4}$, respectively. Thus the grand-mean prior substantially affects posterior counts in this
model, while aggregate site-effect RMSE changes little.

\begin{table}[!htbp]
\centering
\caption{Posterior cluster-count summaries at $I = 0.2$ under five hyperpriors, with converged-fit counts and interval widths (full form of main-text Table~\ref{tab:sim-results}).}
\label{tab:F-main-full}
\begingroup
\setstretch{1}
\fontsize{10}{12}\selectfont
\renewcommand{\arraystretch}{1.18}
\setlength{\tabcolsep}{3pt}
\begin{tabular}{@{}clccccccc@{}}
\toprule
$K^*$ & Hyperprior & $n/N$ & \shortstack{Mean $K_J$\\(SD)} & Bias & RMSE & \shortstack{Collapse\\prob.} & Coverage & \shortstack{Interval\\width} \\
\midrule
5 & Vague & 200/200 & 1.50 (0.11) & $-3.50$ & 3.50 & 0.725 & 0.535 & 3.60 \\
 & DORO-Unif\textsuperscript{$\dagger$} & 200/200 & 1.94 (0.17) & $-3.06$ & 3.07 & 0.536 & 1.000 & 5.07 \\
 & TSMM & 200/200 & 1.75 (0.14) & $-3.25$ & 3.25 & 0.616 & 0.965 & 4.57 \\
 & Dual-Anchor & 200/200 & 3.00 (0.27) & $-2.00$ & 2.02 & 0.315 & 1.000 & 8.29 \\
 & SSI-Weight & 200/200 & 3.41 (0.37) & $-1.59$ & 1.63 & 0.383 & 1.000 & 12.79 \\
\midrule
10 & Vague & 200/200 & 1.51 (0.12) & $-8.49$ & 8.49 & 0.723 & 0.000 & 3.67 \\
 & DORO-Unif\textsuperscript{$\dagger$} & 200/200 & 2.60 (0.31) & $-7.40$ & 7.40 & 0.516 & 0.980 & 10.85 \\
 & TSMM & 200/200 & 4.26 (0.40) & $-5.74$ & 5.76 & 0.269 & 1.000 & 13.70 \\
 & Dual-Anchor & 200/200 & 4.26 (0.39) & $-5.74$ & 5.75 & 0.269 & 1.000 & 13.68 \\
 & SSI-Weight & 200/200 & 3.42 (0.38) & $-6.58$ & 6.59 & 0.382 & 1.000 & 12.83 \\
\midrule
30 & Vague & 200/200 & 1.51 (0.13) & $-28.49$ & 28.49 & 0.723 & 0.000 & 3.70 \\
 & DORO-Unif\textsuperscript{$\dagger$} & 195/200 & 3.57 (0.57) & $-26.43$ & 26.44 & 0.600 & 0.097 & 24.89 \\
 & TSMM & 167/200 & 25.17 (0.32) & $-4.83$ & 4.84 & 0.001 & 1.000 & 32.94 \\
 & Dual-Anchor & 169/200 & 25.18 (0.31) & $-4.82$ & 4.83 & 0.001 & 1.000 & 32.92 \\
 & SSI-Weight & 200/200 & 3.39 (0.41) & $-26.61$ & 26.61 & 0.386 & 0.000 & 12.79 \\

\bottomrule
\end{tabular}
\endgroup
\par\medskip
\begin{minipage}{\linewidth}\normalsize \textit{Note.} Entries are means (between-replication SDs) over the converged fits of 200 replications per cell; $n/N$ is the number of converged fits over the number attempted. Coverage is the proportion of converged fits for which the 95 percent posterior interval for $K_J$ contains $K^*$. Width is the mean length of these intervals. $^{\dagger}$DORO-Unif is an external reference only.\end{minipage}
\end{table}

\subsection{Extended Designs: Misspecified Elicitation, Design Size, and Unequal Weights}
\label{subsec:F11-extended}

\makeatletter\let\primitiveinput=\@@input\makeatother
We extend the simulation to examine misspecified count judgments, design
size, unequal component sizes, and the trade-off weight. Each cell reports
means over fits that passed the four-chain convergence criteria, with the
number retained shown as $n/200$. We do not impute excluded fits. We calibrate
informativeness separately for each $J$, assign equal numbers of sites to the generating components
unless stated otherwise, and assign the grand mean an $N(0,100)$ prior.
Calibration failed before fitting in four planned cells (Dual-Anchor and
DORO-Unif at $\mu_K=2.5$); these cells are omitted. Aggregate $\tau$ RMSE
compares posterior mean site effects with their true values and can conceal
changes at individual sites.

\subsubsection{Misspecified Elicitation}
\label{subsubsec:F111-misspec}

The grid crosses $K^*\in\{5,10,30\}$ and $I\in\{0.2,0.8\}$ with
$r=\mu_K/K^*\in\{0.5,1,2\}$ at $J=100$. Twelve trajectories are complete: all
six medium-confidence ($\mathrm{VIF}=2.5$) TSMM trajectories, four
medium-confidence Dual-Anchor trajectories ($K^*\in\{10,30\}$), and two
low-confidence ($\mathrm{VIF}=5$) Dual-Anchor trajectories. The medium-confidence Dual-Anchor trajectories at $K^*=5$ are incomplete
because calibration failed at $r=0.5$, where $\mu_K=2.5$.
Table~\ref{tab:F11-A} gives all 36 cells. Figure~\ref{fig:5} shows the 30 medium-confidence cells for which all three points of a trajectory are available.

\begin{figure}[!tp]
  \centering
  \includegraphics[width=\textwidth]{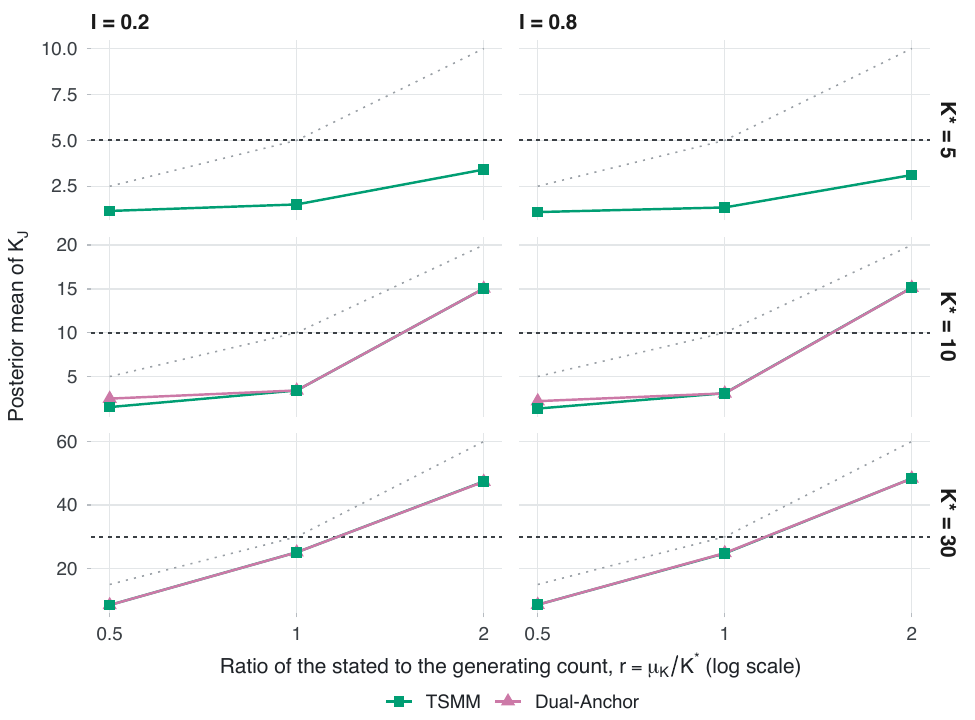}
  \caption{Sensitivity of posterior mean cluster count $K_J$ to the stated mean $\mu_K$ under TSMM and Dual-Anchor ($J = 100$, $N(0, 100)$ grand-mean prior). Rows give generating counts $K^*$ and columns give informativeness $I$. The horizontal axis shows $r = \mu_K / K^*$ on a log scale; dashed lines mark $K^*$ and dotted lines the stated mean $r K^*$. \textit{Note.} Converged fits range from 152 to 200 of 200 per cell (OSM Table~\ref{tab:F11-A}). The figure includes only complete three-point trajectories. Dual-Anchor is omitted at $K^* = 5$ because calibration failed at $r = 0.5$.}
  \label{fig:5}
\end{figure}

\begin{table}[p]
\centering
\caption{Completed Elicited-Mean Misspecification Grid}
\label{tab:F11-A}
\begingroup
\setstretch{1}
\fontsize{10}{12}\selectfont
\renewcommand{\arraystretch}{1.08}
\setlength{\tabcolsep}{4.5pt}
\begin{tabular}{@{}rrrrcrrrrr@{}}
\toprule
$K^*$ & $I$ & VIF & $r$ & $\lambda$ & $n/200$ &
\shortstack{Mean\\$K_J$} & \shortstack{Collapse\\prob.} & \shortstack{$\tau$\\RMSE} & Coverage \\
\midrule
\multicolumn{10}{@{}l}{\textit{TSMM} (VIF $=2.5$)} \\[3pt]
5 & 0.2 & 2.5 & 0.5 & --- & 200/200 & 1.156 & 0.882 & 0.23403 & 0.000 \\
5 & 0.2 & 2.5 & 1.0 & --- & 200/200 & 1.510 & 0.709 & 0.23444 & 0.550 \\
5 & 0.2 & 2.5 & 2.0 & --- & 200/200 & 3.410 & 0.391 & 0.23559 & 1.000 \\
\addlinespace[2pt]
5 & 0.8 & 2.5 & 0.5 & --- & 182/200 & 1.095 & 0.924 & 0.11894 & 0.000 \\
5 & 0.8 & 2.5 & 1.0 & --- & 186/200 & 1.347 & 0.795 & 0.11868 & 0.113 \\
5 & 0.8 & 2.5 & 2.0 & --- & 180/200 & 3.113 & 0.496 & 0.11878 & 1.000 \\
\addlinespace[2pt]
10 & 0.2 & 2.5 & 0.5 & --- & 200/200 & 1.513 & 0.709 & 0.23404 & 0.000 \\
10 & 0.2 & 2.5 & 1.0 & --- & 200/200 & 3.410 & 0.391 & 0.23517 & 1.000 \\
10 & 0.2 & 2.5 & 2.0 & --- & 194/200 & 15.033 & 0.029 & 0.23669 & 1.000 \\
\addlinespace[2pt]
10 & 0.8 & 2.5 & 0.5 & --- & 182/200 & 1.346 & 0.797 & 0.11769 & 0.000 \\
10 & 0.8 & 2.5 & 1.0 & --- & 183/200 & 3.099 & 0.495 & 0.11750 & 1.000 \\
10 & 0.8 & 2.5 & 2.0 & --- & 163/200 & 15.140 & 0.039 & 0.11742 & 1.000 \\
\addlinespace[2pt]
30 & 0.2 & 2.5 & 0.5 & --- & 198/200 & 8.526 & 0.137 & 0.23703 & 0.000 \\
30 & 0.2 & 2.5 & 1.0 & --- & 177/200 & 25.124 & 0.002 & 0.23790 & 1.000 \\
30 & 0.2 & 2.5 & 2.0 & --- & 156/200 & 47.428 & 0.000 & 0.23967 & 1.000 \\
\addlinespace[2pt]
30 & 0.8 & 2.5 & 0.5 & --- & 184/200 & 8.631 & 0.175 & 0.11682 & 0.000 \\
30 & 0.8 & 2.5 & 1.0 & --- & 168/200 & 24.784 & 0.002 & 0.11721 & 1.000 \\
30 & 0.8 & 2.5 & 2.0 & --- & 155/200 & 48.374 & 0.001 & 0.11757 & 1.000 \\
\addlinespace[5pt]
\multicolumn{10}{@{}l}{\textit{Dual-Anchor} (VIF $=2.5$)} \\[3pt]
10 & 0.2 & 2.5 & 0.5 & 0.16 & 200/200 & 2.478 & 0.432 & 0.23498 & 0.015 \\
10 & 0.2 & 2.5 & 1.0 & 1.00 & 200/200 & 3.412 & 0.390 & 0.23518 & 1.000 \\
10 & 0.2 & 2.5 & 2.0 & 1.00 & 193/200 & 15.035 & 0.029 & 0.23654 & 1.000 \\
\addlinespace[2pt]
10 & 0.8 & 2.5 & 0.5 & 0.16 & 187/200 & 2.191 & 0.545 & 0.11763 & 0.005 \\
10 & 0.8 & 2.5 & 1.0 & 1.00 & 182/200 & 3.078 & 0.498 & 0.11766 & 1.000 \\
10 & 0.8 & 2.5 & 2.0 & 1.00 & 169/200 & 15.126 & 0.037 & 0.11726 & 1.000 \\
\addlinespace[2pt]
30 & 0.2 & 2.5 & 0.5 & 1.00 & 200/200 & 8.519 & 0.138 & 0.23709 & 0.000 \\
30 & 0.2 & 2.5 & 1.0 & 1.00 & 171/200 & 25.115 & 0.002 & 0.23828 & 1.000 \\
30 & 0.2 & 2.5 & 2.0 & 1.00 & 152/200 & 47.363 & 0.000 & 0.23944 & 1.000 \\
\addlinespace[2pt]
30 & 0.8 & 2.5 & 0.5 & 1.00 & 185/200 & 8.612 & 0.177 & 0.11736 & 0.000 \\
30 & 0.8 & 2.5 & 1.0 & 1.00 & 154/200 & 24.884 & 0.002 & 0.11787 & 1.000 \\
30 & 0.8 & 2.5 & 2.0 & 1.00 & 158/200 & 48.368 & 0.000 & 0.11797 & 1.000 \\
\addlinespace[5pt]
\multicolumn{10}{@{}l}{\textit{Dual-Anchor} (VIF $=5.0$)} \\[3pt]
5 & 0.8 & 5.0 & 0.5 & 0.02 & 184/200 & 3.351 & 0.272 & 0.11831 & 1.000 \\
5 & 0.8 & 5.0 & 1.0 & 0.12 & 182/200 & 1.605 & 0.733 & 0.11872 & 0.967 \\
5 & 0.8 & 5.0 & 2.0 & 1.00 & 179/200 & 1.522 & 0.798 & 0.11880 & 0.883 \\
\addlinespace[2pt]
10 & 0.8 & 5.0 & 0.5 & 0.12 & 180/200 & 1.615 & 0.730 & 0.11771 & 0.000 \\
10 & 0.8 & 5.0 & 1.0 & 1.00 & 182/200 & 1.534 & 0.795 & 0.11786 & 0.005 \\
10 & 0.8 & 5.0 & 2.0 & 1.00 & 127/200 & 4.739 & 0.534 & 0.11747 & 1.000 \\

\bottomrule
\end{tabular}
\endgroup
\end{table}

Within every complete TSMM trajectory, increasing the stated mean increases
posterior mean $K_J$ and reduces the collapse probability. Cells with the
same $\mu_K$ but different $K^*$ give nearly identical posterior summaries
(for example, $K^*=5$, $r=2$ and $K^*=10$, $r=1$). For the trajectories examined here, posterior mean $K_J$ is more sensitive to the stated mean than to the generating count.

The weight anchor affects only cells with small stated means. At $\mu_K=5$,
Dual-Anchor increases posterior mean $K_J$ and reduces collapse relative to
TSMM. At $\mu_K\ge 10$, the trigger does not fire and the priors coincide.
Low-confidence Dual-Anchor trajectories are not uniformly monotone because
the fixed policy selects different $\lambda$ values as the target changes.

\FloatBarrier
\subsubsection{Design Size}
\label{subsubsec:F112-designsize}

The design-size grid crosses the five priors with
$K^*\in\{5,10\}$, $I\in\{0.2,0.8\}$, and $J\in\{25,50,100,200\}$, giving 20
complete trajectories. The stated targets use $r=1$ and medium confidence where applicable.
Figure~\ref{fig:6} shows how the posterior results vary with design size.
Tables~\ref{tab:F11-B} and \ref{tab:F11-C} report posterior mean $K_J$,
collapse probability, and the number of retained fits.

\begin{figure}[!tp]
  \centering
  \includegraphics[width=\textwidth]{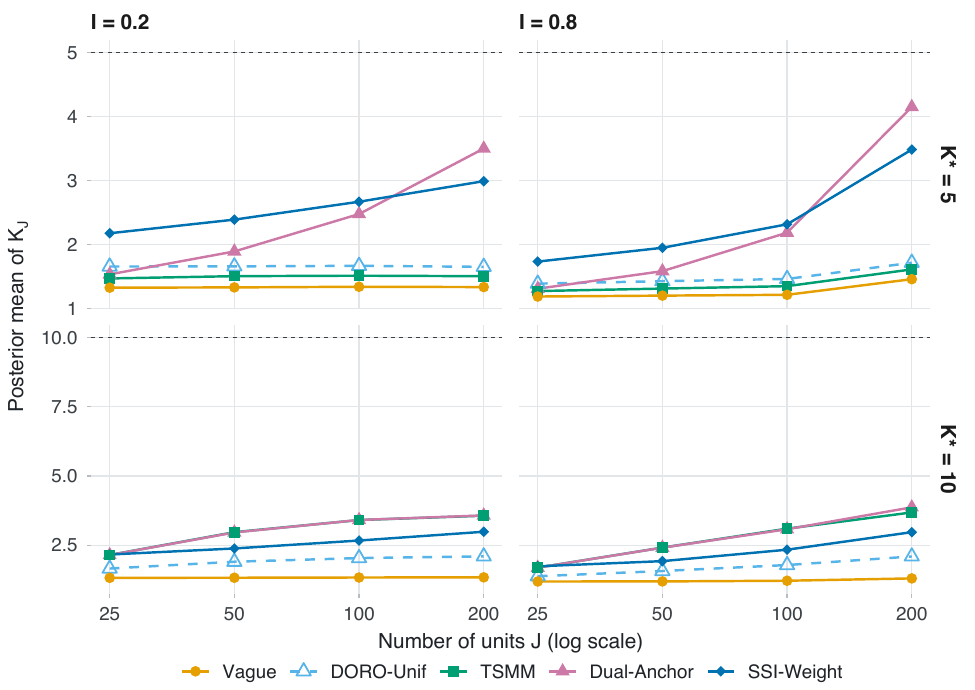}
  \caption{Posterior mean cluster count $K_J$ across design sizes $J = 25$ to 200 under five hyperpriors and the $N(0, 100)$ grand-mean prior. Rows give generating counts $K^*$ and columns give informativeness $I$; dashed lines mark $K^*$. \textit{Note.} At $I = 0.2$, all 200 fits converge through $J = 100$, and 191--200 converge at $J = 200$. At $I = 0.8$, convergence counts are 200, 198--200, 173--186 and 30--93 at $J = 25$, 50, 100 and 200, respectively (exact counts in OSM Appendix~\ref{app:simulation}). DORO-Unif is an external reference only.}
  \label{fig:6}
\end{figure}

\begin{table}[!htbp]
\centering
\caption{Design Size: Posterior Mean $K_J$ by $J$ (number of converged fits in parentheses)}
\label{tab:F11-B}
\begingroup
\small
\setstretch{1}
\renewcommand{\arraystretch}{1.15}
\setlength{\tabcolsep}{3.5pt}
\begin{tabular}{@{}lrrrrrr@{}}
\toprule
Prior & $K^*$ & $I$ & $J=25$ & $J=50$ & $J=100$ & $J=200$ \\
\midrule
Vague & 5 & 0.2 & 1.32 (200) & 1.33 (200) & 1.34 (200) & 1.33 (200) \\
 & 5 & 0.8 & 1.18 (200) & 1.20 (200) & 1.21 (183) & 1.46 (60) \\
 & 10 & 0.2 & 1.32 (200) & 1.33 (200) & 1.33 (200) & 1.34 (200) \\
 & 10 & 0.8 & 1.19 (200) & 1.19 (200) & 1.22 (181) & 1.30 (78) \\
\addlinespace[5pt]
DORO-Unif & 5 & 0.2 & 1.65 (200) & 1.66 (200) & 1.66 (200) & 1.65 (200) \\
 & 5 & 0.8 & 1.39 (200) & 1.42 (200) & 1.46 (184) & 1.71 (67) \\
 & 10 & 0.2 & 1.66 (200) & 1.90 (200) & 2.04 (200) & 2.10 (199) \\
 & 10 & 0.8 & 1.38 (200) & 1.57 (199) & 1.78 (181) & 2.09 (57) \\
\addlinespace[5pt]
TSMM & 5 & 0.2 & 1.47 (200) & 1.50 (200) & 1.51 (200) & 1.50 (200) \\
 & 5 & 0.8 & 1.27 (200) & 1.31 (199) & 1.35 (186) & 1.61 (66) \\
 & 10 & 0.2 & 2.15 (200) & 2.97 (200) & 3.41 (200) & 3.57 (196) \\
 & 10 & 0.8 & 1.69 (200) & 2.41 (198) & 3.10 (183) & 3.68 (58) \\
\addlinespace[5pt]
Dual-Anchor & 5 & 0.2 & 1.53 (200) & 1.89 (200) & 2.47 (200) & 3.50 (199) \\
 & 5 & 0.8 & 1.31 (200) & 1.58 (200) & 2.18 (183) & 4.15 (93) \\
 & 10 & 0.2 & 2.14 (200) & 2.96 (200) & 3.41 (200) & 3.58 (196) \\
 & 10 & 0.8 & 1.69 (200) & 2.41 (198) & 3.08 (182) & 3.87 (53) \\
\addlinespace[5pt]
SSI-Weight & 5 & 0.2 & 2.17 (200) & 2.39 (200) & 2.67 (200) & 2.99 (191) \\
 & 5 & 0.8 & 1.73 (200) & 1.95 (199) & 2.31 (173) & 3.48 (30) \\
 & 10 & 0.2 & 2.16 (200) & 2.38 (200) & 2.67 (200) & 2.99 (196) \\
 & 10 & 0.8 & 1.74 (200) & 1.93 (200) & 2.34 (178) & 2.97 (45) \\

\bottomrule
\end{tabular}
\endgroup
\end{table}

\begin{table}[!htbp]
\centering
\caption{Design Size: Posterior $\Pr(K_J=1\mid y)$ by $J$ (number of converged fits in parentheses)}
\label{tab:F11-C}
\begingroup
\small
\setstretch{1}
\renewcommand{\arraystretch}{1.15}
\setlength{\tabcolsep}{3.5pt}
\begin{tabular}{@{}lrrrrrr@{}}
\toprule
Prior & $K^*$ & $I$ & $J=25$ & $J=50$ & $J=100$ & $J=200$ \\
\midrule
Vague & 5 & 0.2 & 0.789 (200) & 0.796 (200) & 0.797 (200) & 0.802 (200) \\
 & 5 & 0.8 & 0.866 (200) & 0.865 (200) & 0.866 (183) & 0.765 (60) \\
 & 10 & 0.2 & 0.789 (200) & 0.796 (200) & 0.796 (200) & 0.798 (200) \\
 & 10 & 0.8 & 0.864 (200) & 0.869 (200) & 0.864 (181) & 0.831 (78) \\
\addlinespace[5pt]
DORO-Unif & 5 & 0.2 & 0.652 (200) & 0.646 (200) & 0.638 (200) & 0.641 (200) \\
 & 5 & 0.8 & 0.769 (200) & 0.752 (200) & 0.739 (184) & 0.641 (67) \\
 & 10 & 0.2 & 0.739 (200) & 0.667 (200) & 0.633 (200) & 0.617 (199) \\
 & 10 & 0.8 & 0.833 (200) & 0.776 (199) & 0.732 (181) & 0.660 (57) \\
\addlinespace[5pt]
TSMM & 5 & 0.2 & 0.741 (200) & 0.718 (200) & 0.709 (200) & 0.708 (200) \\
 & 5 & 0.8 & 0.833 (200) & 0.810 (199) & 0.795 (186) & 0.685 (66) \\
 & 10 & 0.2 & 0.617 (200) & 0.459 (200) & 0.391 (200) & 0.363 (196) \\
 & 10 & 0.8 & 0.745 (200) & 0.598 (198) & 0.495 (183) & 0.418 (58) \\
\addlinespace[5pt]
Dual-Anchor & 5 & 0.2 & 0.716 (200) & 0.583 (200) & 0.434 (200) & 0.271 (199) \\
 & 5 & 0.8 & 0.815 (200) & 0.704 (200) & 0.546 (183) & 0.228 (93) \\
 & 10 & 0.2 & 0.617 (200) & 0.461 (200) & 0.390 (200) & 0.362 (196) \\
 & 10 & 0.8 & 0.745 (200) & 0.597 (198) & 0.498 (182) & 0.396 (53) \\
\addlinespace[5pt]
SSI-Weight & 5 & 0.2 & 0.514 (200) & 0.517 (200) & 0.511 (200) & 0.503 (191) \\
 & 5 & 0.8 & 0.660 (200) & 0.648 (199) & 0.627 (173) & 0.442 (30) \\
 & 10 & 0.2 & 0.518 (200) & 0.519 (200) & 0.509 (200) & 0.500 (196) \\
 & 10 & 0.8 & 0.657 (200) & 0.654 (200) & 0.621 (178) & 0.554 (45) \\

\bottomrule
\end{tabular}
\endgroup
\end{table}

Under the Vague prior, the posterior probability of $K_J=1$ remains high at all four design sizes. At $K^*=5$, the Dual-Anchor posterior mean approaches the reference count as design size increases. The selected $\lambda$ also changes, taking values of 0.70, 0.30, 0.16 and 0.09 in turn. The change in the mean thus reflects both design size and the choice of weight. At $K^*=10$ the trigger is never exceeded, so Dual-Anchor and TSMM use the same prior. The SSI-Weight prior is the same across design sizes and generating counts.
At $J=200$ and $I=0.8$, the posterior mean counts were 3.48 with five generating components and 2.97 with ten, based on 30 and 45 converged fits, respectively.
Across all five priors in this setting, only 30--93 of 200 fits per condition pass the convergence
criteria. The reported endpoints describe those subsets; we do not substitute
summaries that include failed fits.

\FloatBarrier
\subsubsection{Dominant-Component Location}
\label{subsubsec:F113-unequal}

The DGP assigns exactly 60 of 100 sites to one component. The primary
balanced-random rule rotates this component across ordered component means;
the location sensitivity fixes it at the central mean. We compare all 16 prior-by-$K^*$-by-$I$ conditions under both rules
(Table~\ref{tab:F11-D}).

\begin{table}[!htbp]
\centering
\caption{Completed Dominant-Component Location Grid}
\label{tab:F11-D}
\begingroup
\setstretch{1}
\fontsize{10}{12}\selectfont
\renewcommand{\arraystretch}{1.18}
\setlength{\tabcolsep}{3pt}
\begin{tabular}{@{}lrr rrr rrr@{}}
\toprule
& & & \multicolumn{3}{c}{Balanced random} & \multicolumn{3}{c}{Central} \\
\cmidrule(lr){4-6}\cmidrule(lr){7-9}
Prior & $K^*$ & $I$ & $n/200$ & \shortstack{Mean\\$K_J$} & \shortstack{Collapse\\prob.} &
$n/200$ & \shortstack{Mean\\$K_J$} & \shortstack{Collapse\\prob.} \\
\midrule
Vague & 5 & 0.2 & 200/200 & 1.335 & 0.797 & 200/200 & 1.347 & 0.792 \\
 & 5 & 0.8 & 163/200 & 2.255 & 0.306 & 193/200 & 1.267 & 0.837 \\
 & 10 & 0.2 & 200/200 & 1.343 & 0.793 & 200/200 & 1.350 & 0.791 \\
 & 10 & 0.8 & 178/200 & 2.383 & 0.250 & 184/200 & 1.512 & 0.728 \\
\addlinespace[5pt]
DORO-Unif & 5 & 0.2 & 200/200 & 1.661 & 0.639 & 200/200 & 1.670 & 0.634 \\
 & 5 & 0.8 & 182/200 & 2.612 & 0.207 & 189/200 & 1.548 & 0.698 \\
 & 10 & 0.2 & 200/200 & 2.046 & 0.630 & 200/200 & 2.085 & 0.624 \\
 & 10 & 0.8 & 181/200 & 3.223 & 0.171 & 183/200 & 2.424 & 0.544 \\
\addlinespace[5pt]
TSMM & 5 & 0.2 & 200/200 & 1.516 & 0.707 & 200/200 & 1.528 & 0.702 \\
 & 5 & 0.8 & 173/200 & 2.466 & 0.247 & 189/200 & 1.413 & 0.763 \\
 & 10 & 0.2 & 200/200 & 3.413 & 0.388 & 200/200 & 3.456 & 0.382 \\
 & 10 & 0.8 & 181/200 & 4.284 & 0.092 & 178/200 & 3.899 & 0.309 \\
\addlinespace[5pt]
Dual-Anchor & 5 & 0.2 & 200/200 & 2.460 & 0.434 & 200/200 & 2.488 & 0.427 \\
 & 5 & 0.8 & 178/200 & 3.394 & 0.131 & 187/200 & 2.330 & 0.496 \\
 & 10 & 0.2 & 200/200 & 3.410 & 0.388 & 200/200 & 3.452 & 0.382 \\
 & 10 & 0.8 & 178/200 & 4.251 & 0.090 & 175/200 & 3.923 & 0.307 \\

\bottomrule
\end{tabular}
\endgroup
\end{table}

At $I=0.2$, the two rules for locating the dominant component give posterior means of $K_J$ that differ by less than 0.05 in every condition. At $I=0.8$, central placement lowers posterior mean $K_J$ by 0.33
to 1.06 and raises the collapse probability sharply for every prior. Thus, even at the higher
informativeness level, placing the dominant component at the center shifts
the posterior toward fewer clusters. 

\subsubsection{Sensitivity to the Trade-Off Weight}
\label{subsubsec:F114-lambda}

We compare the selected $\lambda$ with fixed values at $J=100$, $K^*=5$,
medium confidence, equal counts, and grand-mean prior $N(0,100)$.

\begin{table}[!htbp]
\centering
\caption{Selected and Fixed Trade-Off Weights}
\label{tab:F11-E}
\small
\setlength{\tabcolsep}{4pt}
\begin{tabular}{@{}rclrrrr@{}}
\toprule
$I$ & $\lambda$ & Role & $n/200$ & $\E(K_J\mid y)$ &
$\Pr(K_J=1\mid y)$ & $\tau$ RMSE \\
\midrule
0.2 & 0.16 & Selected policy & 200/200 & 2.474 & 0.434 & 0.23533 \\
0.2 & 0.50 & Fixed audit     & 200/200 & 1.761 & 0.617 & 0.23469 \\
0.2 & 0.90 & Fixed audit     & 200/200 & 1.539 & 0.698 & 0.23448 \\
0.8 & 0.16 & Selected policy & 183/200 & 2.183 & 0.546 & 0.11889 \\
0.8 & 0.50 & Fixed audit     & 183/200 & 1.534 & 0.722 & 0.11878 \\
0.8 & 0.90 & Fixed audit     & 183/200 & 1.365 & 0.788 & 0.11894 \\
\bottomrule
\end{tabular}
\end{table}

For these six specifications, the selected $\lambda=0.16$ yields a larger
posterior mean $K_J$ and a smaller posterior collapse probability than the two
fixed values. Aggregate $\tau$ RMSE differs by less than 0.001 within
either informativeness level. 

\FloatBarrier
\subsection{Recovery of the Site Effects and Their Distribution}
\label{subsec:F-finite-site-spread}

We compared the estimated effects with the 100 true effects in each simulated data set. A posterior draw gives a value for every latent site effect. We computed the SD across sites within each draw using denominator $J$, then averaged it over draws. The true finite-site SD, computed with the same denominator, was $0.25\sqrt{99/100}\simeq0.24875$ in every replication. The site RMSE is the root mean square difference between each site's posterior mean and its true effect.

We gave each site equal weight when estimating the finite-site CDF. Writing $m_j$ and $v_j$ for the posterior mean and variance of site $j$, the empirical CDF of the posterior means (PM) puts mass $1/J$ at each $m_j$. Shrinkage toward the overall mean can compress this distribution. \citet{louis1984estimating} introduced constrained Bayes (CB) estimation for normal means, and \citet{ghosh1992constrained} gave a general formulation. We rescale the deviations from $\bar m$ using the marginal posterior variances of the site effects, omitting posterior covariances between sites. The empirical CDF of the adjusted estimates puts mass $1/J$ at each
\[
  \bar m + (m_j-\bar m)\sqrt{1+\bar v/s_m^2},
  \qquad
  \bar v=J^{-1}\sum_{j=1}^J v_j,\qquad
  s_m^2=(J-1)^{-1}\sum_{j=1}^J(m_j-\bar m)^2 .
\]
For each replication we measured the error in each estimated CDF $\widehat F$ against the empirical CDF $F_J$ of the true effects. ISEL is $\int_{\mathbb R}\{\widehat F(t)-F_J(t)\}^2\,dt$ and the Kolmogorov--Smirnov distance (KS) is $\sup_t|\widehat F(t)-F_J(t)|$. We evaluated both over all jumps of the two empirical CDFs and averaged the errors over replications.

Figure~\ref{fig:F-cluster-counts} gives the posterior count estimates for the 45 main conditions, and Tables~\ref{tab:F-finite-site-spread}--\ref{tab:F-effect-recovery-high} give the site and distribution estimates using the converged fits under each prior. For the comparisons across priors in Figure~\ref{fig:3}A and B, we restricted all five means within a design to the same replications. With $K^*=5$ and TSMM, the counts of common replications were 200, 200 and 173 at $I=0.2$, 0.5 and 0.8, respectively. On this basis the three mean CB ISELs were 0.00721, 0.00406 and 0.00247. Under each of the three generating counts, the mean CB ISEL and the mean site RMSE decreased at each successive level of informativeness for all five priors. The CB KS error also decreased as informativeness rose. Its relative range \emph{across priors} did not consistently exceed that of site RMSE. At $K^*=10$ and $I=0.2$, for instance, the two ranges were 0.50 and 0.71 percent.

\begin{figure}[H]
\centering
\includegraphics[width=.90\textwidth]{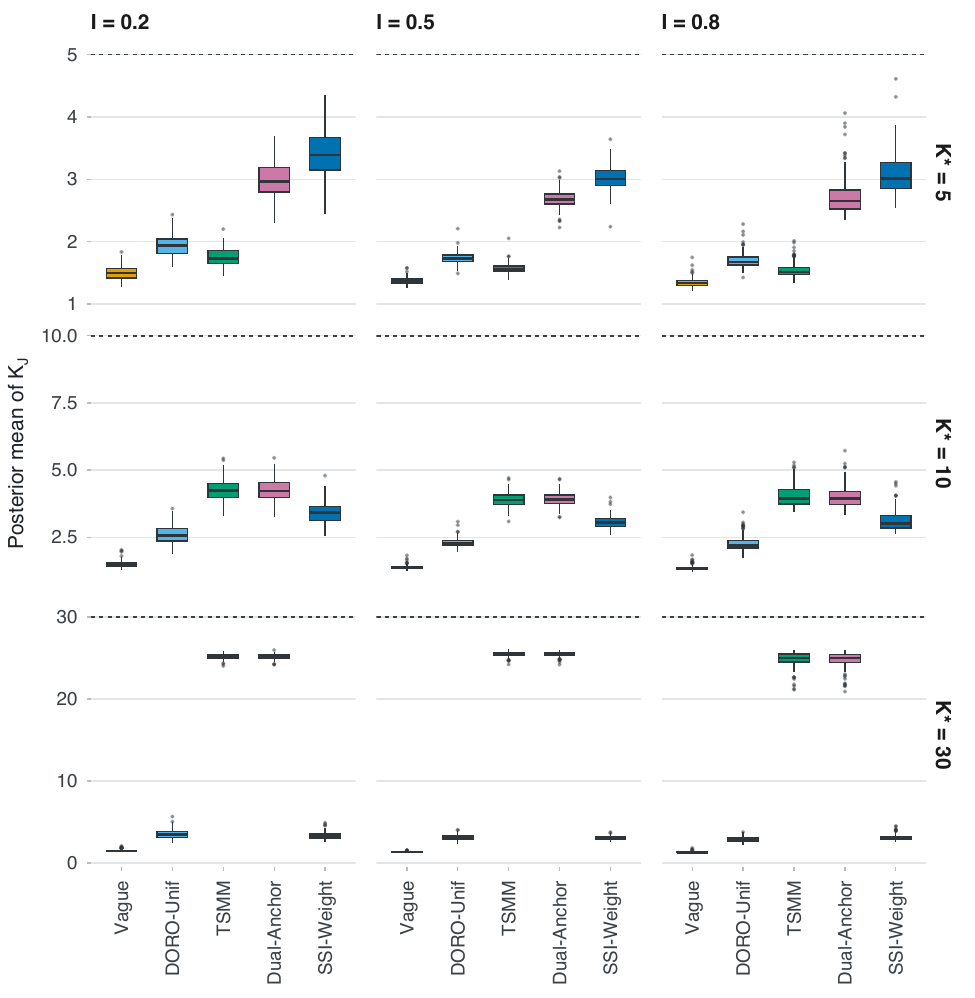}
\caption{Posterior mean cluster count for each generating count $K^*$ and level of site informativeness $I$. The dashed lines mark $K^*$. Each box summarizes the converged fits under the indicated prior (155--200 of 200 replications). Vague denotes $\mathrm{Gamma}(1,1)$ and DORO-Unif is an external reference.}
\label{fig:F-cluster-counts}
\end{figure}

\begin{table}[H]
\centering
\caption{Error in the estimated site effects and their distribution at $I=0.2$.}
\label{tab:F-finite-site-spread}
\begingroup
\setstretch{1}
\fontsize{9}{11}\selectfont
\setlength{\tabcolsep}{2.3pt}
\renewcommand{\arraystretch}{1.00}
\begin{tabular*}{\textwidth}{@{\extracolsep{\fill}}rlrrrrrrrrr@{}}
\toprule
& & & Mean & Finite & SD & Site & \multicolumn{2}{c}{ISEL} & \multicolumn{2}{c}{KS} \\
\cmidrule(lr){8-9}\cmidrule(l){10-11}
$K^*$ & Prior & $n$ & count & SD & bias & RMSE & PM & CB & PM & CB \\
\midrule
5 & $\mathrm{Gamma}(1,1)$ & 200 & 1.50 & 0.2942 & 0.0455 & 0.2344 & 0.02061 & 0.00719 & 0.2922 & 0.1653 \\
 & DORO-Unif & 200 & 1.94 & 0.2970 & 0.0483 & 0.2349 & 0.02036 & 0.00724 & 0.2912 & 0.1650 \\
 & TSMM & 200 & 1.75 & 0.2958 & 0.0470 & 0.2346 & 0.02047 & 0.00721 & 0.2916 & 0.1644 \\
 & Dual-Anchor & 200 & 3.00 & 0.3020 & 0.0532 & 0.2357 & 0.01982 & 0.00736 & 0.2878 & 0.1650 \\
 & SSI-Weight & 200 & 3.41 & 0.3009 & 0.0522 & 0.2355 & 0.01988 & 0.00733 & 0.2882 & 0.1650 \\
\addlinespace
10 & $\mathrm{Gamma}(1,1)$ & 200 & 1.51 & 0.2939 & 0.0451 & 0.2339 & 0.01885 & 0.00586 & 0.2600 & 0.1402 \\
 & DORO-Unif & 200 & 2.60 & 0.2981 & 0.0493 & 0.2346 & 0.01839 & 0.00598 & 0.2580 & 0.1402 \\
 & TSMM & 200 & 4.26 & 0.3036 & 0.0549 & 0.2355 & 0.01775 & 0.00612 & 0.2544 & 0.1399 \\
 & Dual-Anchor & 200 & 4.26 & 0.3036 & 0.0549 & 0.2355 & 0.01774 & 0.00612 & 0.2543 & 0.1403 \\
 & SSI-Weight & 200 & 3.42 & 0.3010 & 0.0523 & 0.2351 & 0.01806 & 0.00606 & 0.2561 & 0.1406 \\
\addlinespace
30 & $\mathrm{Gamma}(1,1)$ & 200 & 1.51 & 0.2917 & 0.0430 & 0.2350 & 0.01983 & 0.00615 & 0.2644 & 0.1394 \\
 & DORO-Unif & 195 & 3.57 & 0.2948 & 0.0461 & 0.2354 & 0.01938 & 0.00618 & 0.2621 & 0.1392 \\
 & TSMM & 167 & 25.17 & 0.3162 & 0.0674 & 0.2379 & 0.01580 & 0.00646 & 0.2387 & 0.1374 \\
 & Dual-Anchor & 169 & 25.18 & 0.3167 & 0.0680 & 0.2378 & 0.01562 & 0.00649 & 0.2388 & 0.1383 \\
 & SSI-Weight & 200 & 3.39 & 0.2985 & 0.0497 & 0.2361 & 0.01905 & 0.00628 & 0.2600 & 0.1398 \\
\bottomrule
\end{tabular*}
\endgroup
\smallskip
\noindent\parbox{\textwidth}{\small\textit{Note.} Entries are averages over the $n$ converged fits out of 200 replications for each prior. PM and CB denote empirical CDFs of the posterior means and constrained Bayes estimates. SD bias is relative to the true finite-site SD of 0.24875.}
\end{table}

\begin{table}[H]
\centering
\caption{Error in the estimated site effects and their distribution at $I=0.5$.}
\label{tab:F-effect-recovery-middle}
\begingroup
\setstretch{1}
\fontsize{9}{11}\selectfont
\setlength{\tabcolsep}{2.3pt}
\renewcommand{\arraystretch}{1.00}
\begin{tabular*}{\textwidth}{@{\extracolsep{\fill}}rlrrrrrrrrr@{}}
\toprule
& & & Mean & Finite & SD & Site & \multicolumn{2}{c}{ISEL} & \multicolumn{2}{c}{KS} \\
\cmidrule(lr){8-9}\cmidrule(l){10-11}
$K^*$ & Prior & $n$ & count & SD & bias & RMSE & PM & CB & PM & CB \\
\midrule
5 & $\mathrm{Gamma}(1,1)$ & 200 & 1.37 & 0.2561 & 0.0073 & 0.1783 & 0.01023 & 0.00405 & 0.2139 & 0.1416 \\
 & DORO-Unif & 200 & 1.73 & 0.2567 & 0.0080 & 0.1784 & 0.01021 & 0.00408 & 0.2141 & 0.1418 \\
 & TSMM & 200 & 1.57 & 0.2565 & 0.0077 & 0.1783 & 0.01021 & 0.00406 & 0.2140 & 0.1418 \\
 & Dual-Anchor & 200 & 2.69 & 0.2577 & 0.0089 & 0.1786 & 0.01017 & 0.00412 & 0.2140 & 0.1424 \\
 & SSI-Weight & 200 & 3.02 & 0.2574 & 0.0086 & 0.1785 & 0.01017 & 0.00410 & 0.2139 & 0.1424 \\
\addlinespace
10 & $\mathrm{Gamma}(1,1)$ & 200 & 1.38 & 0.2565 & 0.0078 & 0.1788 & 0.00870 & 0.00275 & 0.1839 & 0.1126 \\
 & DORO-Unif & 200 & 2.28 & 0.2573 & 0.0086 & 0.1789 & 0.00868 & 0.00279 & 0.1839 & 0.1129 \\
 & TSMM & 199 & 3.90 & 0.2584 & 0.0097 & 0.1793 & 0.00864 & 0.00282 & 0.1833 & 0.1135 \\
 & Dual-Anchor & 200 & 3.91 & 0.2585 & 0.0097 & 0.1791 & 0.00864 & 0.00283 & 0.1833 & 0.1134 \\
 & SSI-Weight & 200 & 3.05 & 0.2579 & 0.0091 & 0.1790 & 0.00866 & 0.00281 & 0.1838 & 0.1133 \\
\addlinespace
30 & $\mathrm{Gamma}(1,1)$ & 200 & 1.38 & 0.2580 & 0.0092 & 0.1795 & 0.00844 & 0.00272 & 0.1803 & 0.1112 \\
 & DORO-Unif & 189 & 3.14 & 0.2584 & 0.0096 & 0.1798 & 0.00845 & 0.00274 & 0.1805 & 0.1115 \\
 & TSMM & 196 & 25.49 & 0.2612 & 0.0124 & 0.1799 & 0.00823 & 0.00277 & 0.1787 & 0.1114 \\
 & Dual-Anchor & 193 & 25.48 & 0.2620 & 0.0132 & 0.1798 & 0.00810 & 0.00276 & 0.1777 & 0.1109 \\
 & SSI-Weight & 199 & 3.06 & 0.2593 & 0.0106 & 0.1798 & 0.00842 & 0.00279 & 0.1804 & 0.1124 \\
\bottomrule
\end{tabular*}
\endgroup
\smallskip
\noindent\parbox{\textwidth}{\small\textit{Note.} Entries are averages over the $n$ converged fits out of 200 replications for each prior. PM and CB denote empirical CDFs of the posterior means and constrained Bayes estimates. SD bias is relative to the true finite-site SD of 0.24875.}
\end{table}

\begin{table}[H]
\centering
\caption{Error in the estimated site effects and their distribution at $I=0.8$.}
\label{tab:F-effect-recovery-high}
\begingroup
\setstretch{1}
\fontsize{9}{11}\selectfont
\setlength{\tabcolsep}{2.3pt}
\renewcommand{\arraystretch}{1.00}
\begin{tabular*}{\textwidth}{@{\extracolsep{\fill}}rlrrrrrrrrr@{}}
\toprule
& & & Mean & Finite & SD & Site & \multicolumn{2}{c}{ISEL} & \multicolumn{2}{c}{KS} \\
\cmidrule(lr){8-9}\cmidrule(l){10-11}
$K^*$ & Prior & $n$ & count & SD & bias & RMSE & PM & CB & PM & CB \\
\midrule
5 & $\mathrm{Gamma}(1,1)$ & 183 & 1.34 & 0.2503 & 0.0015 & 0.1190 & 0.00375 & 0.00245 & 0.1445 & 0.1178 \\
 & DORO-Unif & 182 & 1.70 & 0.2506 & 0.0019 & 0.1188 & 0.00372 & 0.00245 & 0.1443 & 0.1175 \\
 & TSMM & 185 & 1.54 & 0.2505 & 0.0017 & 0.1187 & 0.00373 & 0.00245 & 0.1444 & 0.1178 \\
 & Dual-Anchor & 185 & 2.72 & 0.2508 & 0.0021 & 0.1186 & 0.00374 & 0.00249 & 0.1452 & 0.1192 \\
 & SSI-Weight & 182 & 3.09 & 0.2508 & 0.0020 & 0.1189 & 0.00372 & 0.00247 & 0.1446 & 0.1186 \\
\addlinespace
10 & $\mathrm{Gamma}(1,1)$ & 185 & 1.34 & 0.2503 & 0.0016 & 0.1179 & 0.00244 & 0.00128 & 0.1134 & 0.0851 \\
 & DORO-Unif & 178 & 2.26 & 0.2507 & 0.0020 & 0.1178 & 0.00242 & 0.00129 & 0.1131 & 0.0849 \\
 & TSMM & 188 & 4.05 & 0.2514 & 0.0026 & 0.1174 & 0.00238 & 0.00132 & 0.1132 & 0.0865 \\
 & Dual-Anchor & 183 & 4.02 & 0.2514 & 0.0027 & 0.1175 & 0.00238 & 0.00132 & 0.1132 & 0.0860 \\
 & SSI-Weight & 183 & 3.10 & 0.2510 & 0.0023 & 0.1177 & 0.00240 & 0.00130 & 0.1138 & 0.0855 \\
\addlinespace
30 & $\mathrm{Gamma}(1,1)$ & 188 & 1.34 & 0.2514 & 0.0027 & 0.1176 & 0.00234 & 0.00123 & 0.1114 & 0.0854 \\
 & DORO-Unif & 155 & 2.95 & 0.2516 & 0.0028 & 0.1179 & 0.00237 & 0.00126 & 0.1114 & 0.0854 \\
 & TSMM & 168 & 24.85 & 0.2524 & 0.0037 & 0.1178 & 0.00232 & 0.00126 & 0.1105 & 0.0861 \\
 & Dual-Anchor & 170 & 24.81 & 0.2528 & 0.0040 & 0.1177 & 0.00228 & 0.00124 & 0.1102 & 0.0859 \\
 & SSI-Weight & 185 & 3.08 & 0.2517 & 0.0030 & 0.1175 & 0.00233 & 0.00125 & 0.1117 & 0.0863 \\
\bottomrule
\end{tabular*}
\endgroup
\smallskip
\noindent\parbox{\textwidth}{\small\textit{Note.} Entries are averages over the $n$ converged fits out of 200 replications for each prior. PM and CB denote empirical CDFs of the posterior means and constrained Bayes estimates. SD bias is relative to the true finite-site SD of 0.24875.}
\end{table}

\subsubsection{Posterior Mean CDFs from Joint Draws}
\label{subsubsec:F-cdf-replay}

The fit records for the 45 conditions retain each site's posterior mean and variance. For a direct comparison with the finite-population CDF used in the applications, we recovered joint posterior draws for two low-information data sets. At $K^*=5$ we compared TSMM with Dual-Anchor. At $K^*=30$ we compared $\mathrm{Gamma}(1,1)$ with TSMM. For each pair we selected the first replication in which both fits converged and reran their four original chains with the same data and seeds. The saved site means, site SDs and cluster-count means were reproduced to within $10^{-12}$.

We averaged the empirical CDF of the 100 latent effects over joint posterior draws to obtain the posterior mean finite-site CDF. Following \citet{lee2025improving}, we also took its $(k-0.5)/J$ quantiles for $k=1,\ldots,J$ and placed mass $1/J$ at each quantile to obtain the triple-goal (GR) empirical CDF. Figure~\ref{fig:F-replayed-cdf} shows the CDFs. Table~\ref{tab:F-replayed-cdf} compares each estimate with the true CDF. The maximum absolute difference between the two posterior mean CDFs fitted to the same data was 0.281 percentage points at $K^*=5$ and 1.002 percentage points at $K^*=30$. These comparisons describe one replication at each count.

\begin{figure}[H]
\centering
\includegraphics[width=\textwidth]{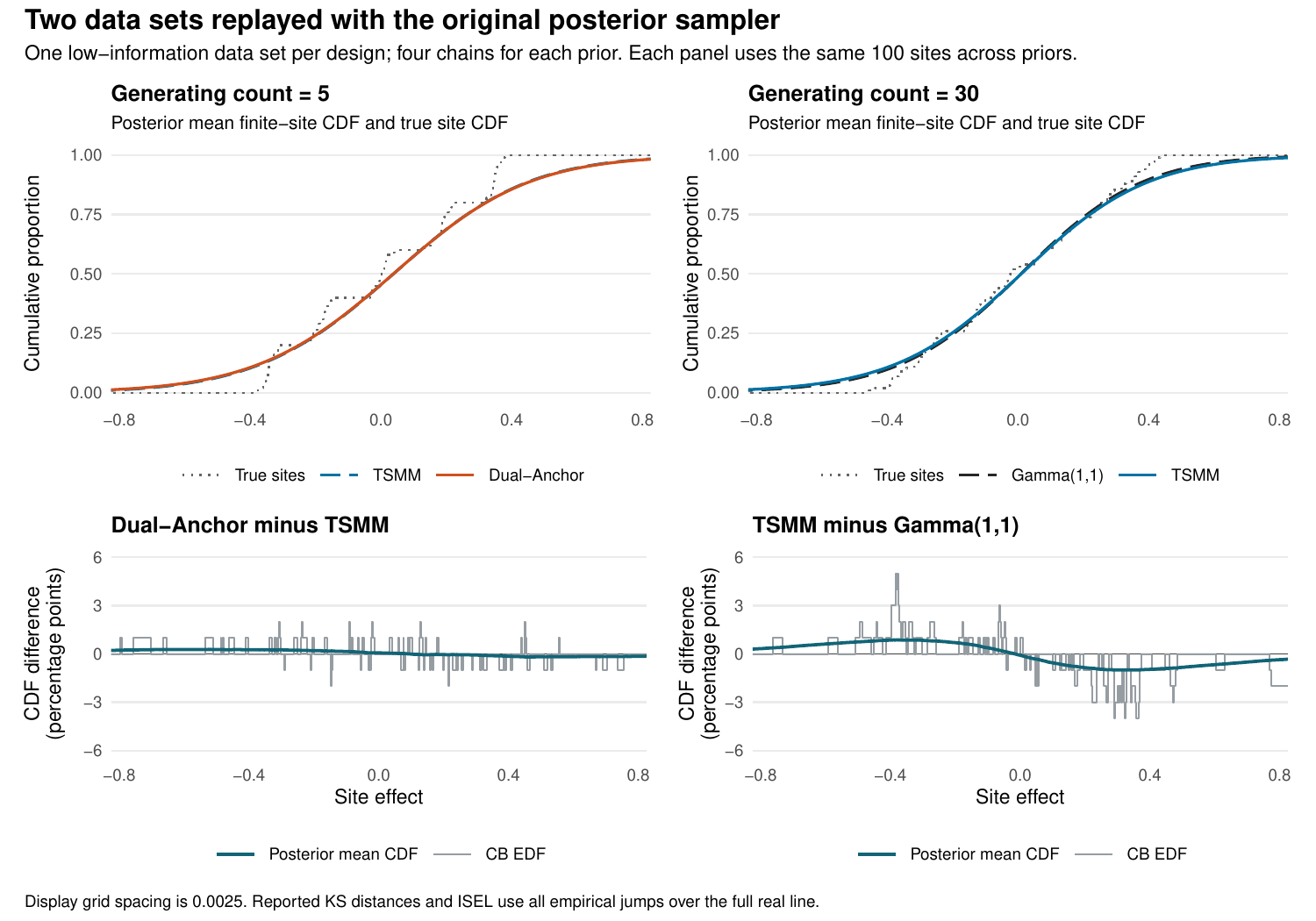}
\caption{Estimates of the finite-site CDF in two low-information data sets. Each column compares two priors fitted to the same 100 sites. The lower panels give the difference between the two estimates at each effect value.}
\label{fig:F-replayed-cdf}
\end{figure}

\begin{table}[H]
\centering
\caption{Errors relative to the true finite-site CDF in the two data sets in Figure~\ref{fig:F-replayed-cdf}.}
\label{tab:F-replayed-cdf}
\begingroup
\small
\setstretch{1}
\begin{tabular}{llrrrr}
\toprule
 & & \multicolumn{2}{c}{Posterior mean CDF} & \multicolumn{2}{c}{GR empirical CDF} \\
\cmidrule(lr){3-4}\cmidrule(l){5-6}
$K^*$ & Prior & ISEL & KS & ISEL & KS \\
\midrule
5 & TSMM & 0.006453 & 0.1557 & 0.006472 & 0.1600 \\
5 & Dual-Anchor & 0.006661 & 0.1570 & 0.006679 & 0.1600 \\
30 & $\mathrm{Gamma}(1,1)$ & 0.002162 & 0.0858 & 0.002180 & 0.0900 \\
30 & TSMM & 0.002805 & 0.0949 & 0.002823 & 0.0900 \\
\bottomrule
\end{tabular}
\endgroup
\end{table}

\subsubsection{Percentiles of the Finite-Site Effect Distribution}
\label{subsubsec:F-cb-percentiles}

For each simulated data set, we sorted the 100 true site effects and, separately, the 100 CB estimates under each prior. For $p=1,5,25,50,75,95,99$, we took the $p$th smallest value as the $p$th percentile, as given by the inverse empirical CDF. We compared the two priors on the same data sets, retaining the 1,681 of 1,800 replications in which both $\mathrm{Gamma}(1,1)$ and Dual-Anchor converged (165--200 per design). In the applications we obtain percentiles by inverting the posterior mean finite-population CDF (Figure~\ref{fig:app-effect-changes}).

We first asked how far the CB percentile moved when we changed the prior. Within each replication we subtracted the estimate under $\mathrm{Gamma}(1,1)$ from that under Dual-Anchor. Figure~\ref{fig:F-cb-percentile-movement} gives both the mean difference and the mean absolute difference. The latter was larger at the 99th percentile than at the median in all nine designs. At $I=0.8$, however, the 99th percentile moved by only 0.0027--0.0040 effect units on average in absolute terms, and the mean absolute changes at the 95th percentile and the median differed by less than 0.0001. The lower tail also moved. At $I=0.2$, the mean absolute change at the 5th percentile exceeded that at the median for each generating count.

\begin{figure}[H]
\centering
\includegraphics[width=\textwidth]{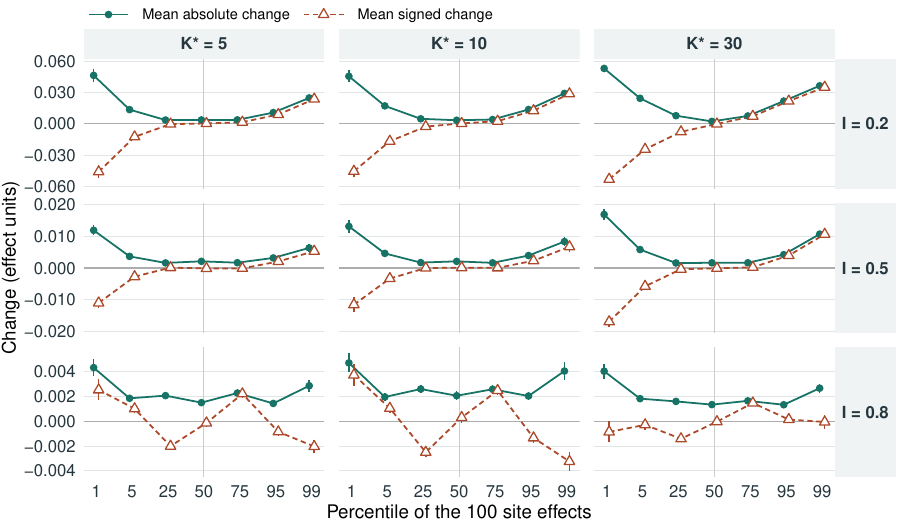}
\caption{Change in estimates of seven finite-site percentiles under Dual-Anchor relative to $\mathrm{Gamma}(1,1)$. Solid lines give the mean absolute change and dashed lines give the mean signed change. Whiskers show pointwise 95 percent Monte Carlo intervals for these means across paired replications. Vertical scales differ by information level. Percentiles are equally spaced for display.}
\label{fig:F-cb-percentile-movement}
\end{figure}

To assess recovery, we compared each prior's CB estimate with the true percentile in the same replication (Figure~\ref{fig:F-cb-percentile-recovery} and Tables~\ref{tab:F-cb-percentiles-low}--\ref{tab:F-cb-percentiles-high}). At $I=0.2$ and 0.5, both priors overestimated the 99th percentile on average at all three generating counts. Dual-Anchor gave a larger positive bias and a larger RMSE than $\mathrm{Gamma}(1,1)$ in each of these six designs. At $I=0.8$ the RMSEs for the 99th percentile under the two priors differed by less than 0.002 in each design.

\begin{figure}[H]
\centering
\includegraphics[width=\textwidth]{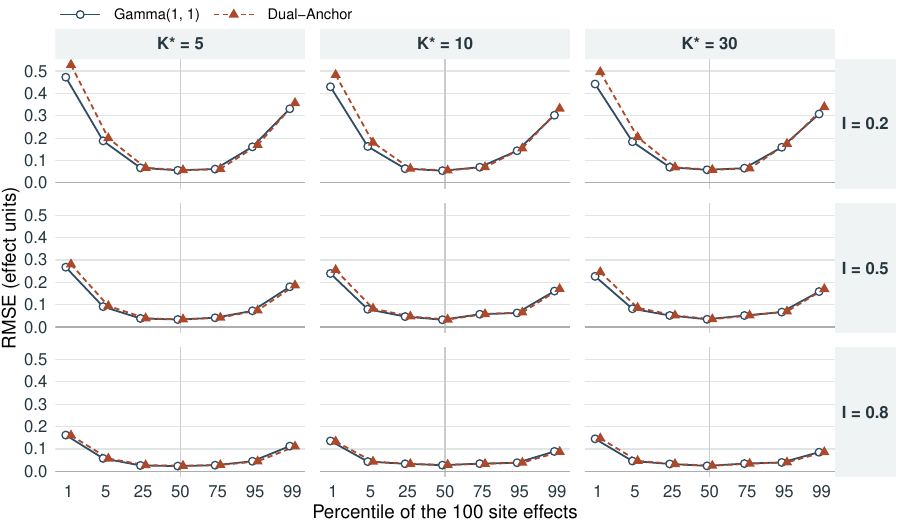}
\caption{Root mean squared error of the CB percentiles under $\mathrm{Gamma}(1,1)$ and Dual-Anchor. Each error is relative to the true percentile of the same 100 sites. All panels use the same vertical scale.}
\label{fig:F-cb-percentile-recovery}
\end{figure}

Monte Carlo variation within a fit contributes to the observed changes. For $K^*=10$ and 30, TSMM and Dual-Anchor specify the same prior, so their separate fits provide a check. At $I=0.8$, among the replications in which all three fits converged, the mean absolute change at the 99th percentile was 0.00371 for $\mathrm{Gamma}(1,1)$ versus Dual-Anchor and 0.00062 for TSMM versus Dual-Anchor when $K^*=10$ ($n=167$). The corresponding values for $K^*=30$ were 0.00269 and 0.00058 ($n=147$).

\begin{table}[H]
\centering
\caption{Recovery of finite-site percentiles and their sensitivity to the prior at $I=0.2$.}
\label{tab:F-cb-percentiles-low}
\begingroup
\setstretch{1}
\fontsize{9}{11}\selectfont
\setlength{\tabcolsep}{2.2pt}
\renewcommand{\arraystretch}{1.08}
\begin{tabular*}{\textwidth}{@{\extracolsep{\fill}}rrrrrrrrrr@{}}
\toprule
 & Mean & \multicolumn{3}{c}{$\mathrm{Gamma}(1,1)$} & \multicolumn{3}{c}{Dual-Anchor} & \multicolumn{2}{c}{Change in estimate} \\
\cmidrule(lr){3-5}\cmidrule(lr){6-8}\cmidrule(l){9-10}
Percentile & truth & Bias & MAE & RMSE & Bias & MAE & RMSE & Signed & Absolute \\
\midrule
\multicolumn{10}{l}{$K^*=5$, $n=200$} \\
1 & $-0.3905$ & $-0.4240$ & 0.4242 & 0.4724 & $-0.4700$ & 0.4700 & 0.5275 & $-0.04602$ & 0.04646 \\
5 & $-0.3662$ & $-0.1539$ & 0.1597 & 0.1875 & $-0.1664$ & 0.1713 & 0.1996 & $-0.01249$ & 0.01361 \\
25 & $-0.1908$ & $-0.0119$ & 0.0531 & 0.0669 & $-0.0122$ & 0.0530 & 0.0670 & $-0.00024$ & 0.00365 \\
50 & $-0.0010$ & 0.0005 & 0.0436 & 0.0559 & 0.0008 & 0.0442 & 0.0565 & 0.00031 & 0.00372 \\
75 & 0.1870 & 0.0055 & 0.0481 & 0.0613 & 0.0071 & 0.0486 & 0.0619 & 0.00160 & 0.00379 \\
95 & 0.3631 & 0.1219 & 0.1329 & 0.1609 & 0.1309 & 0.1407 & 0.1693 & 0.00894 & 0.01079 \\
99 & 0.3800 & 0.2981 & 0.2988 & 0.3311 & 0.3219 & 0.3223 & 0.3567 & 0.02373 & 0.02484 \\
\addlinespace[3pt]
\multicolumn{10}{l}{$K^*=10$, $n=200$} \\
1 & $-0.4226$ & $-0.3810$ & 0.3816 & 0.4296 & $-0.4265$ & 0.4267 & 0.4827 & $-0.04549$ & 0.04562 \\
5 & $-0.3902$ & $-0.1259$ & 0.1343 & 0.1625 & $-0.1426$ & 0.1494 & 0.1799 & $-0.01664$ & 0.01722 \\
25 & $-0.2178$ & 0.0130 & 0.0509 & 0.0632 & 0.0104 & 0.0507 & 0.0630 & $-0.00261$ & 0.00482 \\
50 & $-0.0098$ & 0.0050 & 0.0438 & 0.0542 & 0.0054 & 0.0447 & 0.0554 & 0.00036 & 0.00343 \\
75 & 0.2131 & $-0.0164$ & 0.0558 & 0.0698 & $-0.0138$ & 0.0557 & 0.0698 & 0.00258 & 0.00414 \\
95 & 0.3853 & 0.0966 & 0.1128 & 0.1434 & 0.1093 & 0.1225 & 0.1541 & 0.01265 & 0.01374 \\
99 & 0.4121 & 0.2653 & 0.2662 & 0.3022 & 0.2941 & 0.2945 & 0.3319 & 0.02875 & 0.02922 \\
\addlinespace[3pt]
\multicolumn{10}{l}{$K^*=30$, $n=169$} \\
1 & $-0.4415$ & $-0.3915$ & 0.3915 & 0.4418 & $-0.4446$ & 0.4446 & 0.4954 & $-0.05314$ & 0.05314 \\
5 & $-0.3916$ & $-0.1448$ & 0.1521 & 0.1837 & $-0.1692$ & 0.1733 & 0.2038 & $-0.02444$ & 0.02445 \\
25 & $-0.2193$ & 0.0055 & 0.0554 & 0.0695 & $-0.0022$ & 0.0555 & 0.0694 & $-0.00766$ & 0.00781 \\
50 & $-0.0047$ & $-0.0074$ & 0.0475 & 0.0580 & $-0.0075$ & 0.0477 & 0.0583 & $-0.00012$ & 0.00232 \\
75 & 0.2099 & $-0.0158$ & 0.0524 & 0.0655 & $-0.0085$ & 0.0522 & 0.0642 & 0.00733 & 0.00755 \\
95 & 0.3814 & 0.1142 & 0.1245 & 0.1589 & 0.1361 & 0.1416 & 0.1745 & 0.02183 & 0.02183 \\
99 & 0.4209 & 0.2637 & 0.2637 & 0.3078 & 0.2989 & 0.2989 & 0.3386 & 0.03514 & 0.03651 \\
\bottomrule
\end{tabular*}
\endgroup
\smallskip
\noindent\parbox{\textwidth}{\small\textit{Note.} The $n$ replications had converged fits under both priors. Mean truth is the average of the true finite-site percentile over these replications. Bias, mean absolute error (MAE) and RMSE compare each estimate with its true value. The last two columns summarize the paired difference (Dual-Anchor minus $\mathrm{Gamma}(1,1)$) in site-effect units. Absolute differences were taken before averaging.}
\end{table}

\begin{table}[H]
\centering
\caption{Recovery of finite-site percentiles and their sensitivity to the prior at $I=0.5$.}
\label{tab:F-cb-percentiles-middle}
\begingroup
\setstretch{1}
\fontsize{9}{11}\selectfont
\setlength{\tabcolsep}{2.2pt}
\renewcommand{\arraystretch}{1.08}
\begin{tabular*}{\textwidth}{@{\extracolsep{\fill}}rrrrrrrrrr@{}}
\toprule
 & Mean & \multicolumn{3}{c}{$\mathrm{Gamma}(1,1)$} & \multicolumn{3}{c}{Dual-Anchor} & \multicolumn{2}{c}{Change in estimate} \\
\cmidrule(lr){3-5}\cmidrule(lr){6-8}\cmidrule(l){9-10}
Percentile & truth & Bias & MAE & RMSE & Bias & MAE & RMSE & Signed & Absolute \\
\midrule
\multicolumn{10}{l}{$K^*=5$, $n=200$} \\
1 & $-0.3901$ & $-0.2487$ & 0.2489 & 0.2677 & $-0.2596$ & 0.2598 & 0.2806 & $-0.01094$ & 0.01189 \\
5 & $-0.3669$ & $-0.0718$ & 0.0768 & 0.0913 & $-0.0745$ & 0.0792 & 0.0944 & $-0.00268$ & 0.00368 \\
25 & $-0.1911$ & 0.0022 & 0.0300 & 0.0385 & 0.0024 & 0.0305 & 0.0393 & 0.00015 & 0.00163 \\
50 & $-0.0015$ & $-0.0039$ & 0.0259 & 0.0337 & $-0.0040$ & 0.0264 & 0.0344 & $-0.00011$ & 0.00213 \\
75 & 0.1874 & $-0.0126$ & 0.0328 & 0.0415 & $-0.0126$ & 0.0334 & 0.0421 & $-0.00008$ & 0.00170 \\
95 & 0.3634 & 0.0470 & 0.0582 & 0.0725 & 0.0491 & 0.0597 & 0.0745 & 0.00208 & 0.00319 \\
99 & 0.3803 & 0.1628 & 0.1640 & 0.1801 & 0.1681 & 0.1694 & 0.1866 & 0.00530 & 0.00635 \\
\addlinespace[3pt]
\multicolumn{10}{l}{$K^*=10$, $n=200$} \\
1 & $-0.4243$ & $-0.2096$ & 0.2100 & 0.2396 & $-0.2211$ & 0.2214 & 0.2554 & $-0.01149$ & 0.01312 \\
5 & $-0.3910$ & $-0.0491$ & 0.0644 & 0.0796 & $-0.0524$ & 0.0672 & 0.0829 & $-0.00329$ & 0.00462 \\
25 & $-0.2177$ & 0.0267 & 0.0380 & 0.0469 & 0.0268 & 0.0384 & 0.0473 & 0.00003 & 0.00173 \\
50 & $-0.0087$ & 0.0027 & 0.0265 & 0.0327 & 0.0029 & 0.0273 & 0.0337 & 0.00013 & 0.00212 \\
75 & 0.2126 & $-0.0405$ & 0.0477 & 0.0574 & $-0.0404$ & 0.0481 & 0.0579 & 0.00007 & 0.00165 \\
95 & 0.3847 & 0.0257 & 0.0498 & 0.0631 & 0.0280 & 0.0514 & 0.0652 & 0.00232 & 0.00395 \\
99 & 0.4113 & 0.1349 & 0.1366 & 0.1609 & 0.1417 & 0.1433 & 0.1705 & 0.00671 & 0.00834 \\
\addlinespace[3pt]
\multicolumn{10}{l}{$K^*=30$, $n=193$} \\
1 & $-0.4396$ & $-0.1929$ & 0.1956 & 0.2268 & $-0.2098$ & 0.2119 & 0.2455 & $-0.01686$ & 0.01687 \\
5 & $-0.3901$ & $-0.0524$ & 0.0651 & 0.0817 & $-0.0582$ & 0.0695 & 0.0862 & $-0.00573$ & 0.00582 \\
25 & $-0.2195$ & 0.0330 & 0.0429 & 0.0517 & 0.0327 & 0.0431 & 0.0518 & $-0.00034$ & 0.00161 \\
50 & $-0.0047$ & 0.0003 & 0.0271 & 0.0347 & 0.0002 & 0.0274 & 0.0351 & $-0.00007$ & 0.00166 \\
75 & 0.2096 & $-0.0321$ & 0.0426 & 0.0517 & $-0.0318$ & 0.0426 & 0.0518 & 0.00025 & 0.00167 \\
95 & 0.3822 & 0.0347 & 0.0522 & 0.0670 & 0.0387 & 0.0543 & 0.0697 & 0.00398 & 0.00422 \\
99 & 0.4231 & 0.1325 & 0.1350 & 0.1586 & 0.1432 & 0.1454 & 0.1701 & 0.01064 & 0.01069 \\
\bottomrule
\end{tabular*}
\endgroup
\smallskip
\noindent\parbox{\textwidth}{\small\textit{Note.} The $n$ replications had converged fits under both priors. Mean truth is the average of the true finite-site percentile over these replications. Bias, mean absolute error (MAE) and RMSE compare each estimate with its true value. The last two columns summarize the paired difference (Dual-Anchor minus $\mathrm{Gamma}(1,1)$) in site-effect units. Absolute differences were taken before averaging.}
\end{table}

\begin{table}[H]
\centering
\caption{Recovery of finite-site percentiles and their sensitivity to the prior at $I=0.8$.}
\label{tab:F-cb-percentiles-high}
\begingroup
\setstretch{1}
\fontsize{9}{11}\selectfont
\setlength{\tabcolsep}{2.2pt}
\renewcommand{\arraystretch}{1.08}
\begin{tabular*}{\textwidth}{@{\extracolsep{\fill}}rrrrrrrrrr@{}}
\toprule
 & Mean & \multicolumn{3}{c}{$\mathrm{Gamma}(1,1)$} & \multicolumn{3}{c}{Dual-Anchor} & \multicolumn{2}{c}{Change in estimate} \\
\cmidrule(lr){3-5}\cmidrule(lr){6-8}\cmidrule(l){9-10}
Percentile & truth & Bias & MAE & RMSE & Bias & MAE & RMSE & Signed & Absolute \\
\midrule
\multicolumn{10}{l}{$K^*=5$, $n=179$} \\
1 & $-0.3909$ & $-0.1513$ & 0.1513 & 0.1621 & $-0.1488$ & 0.1488 & 0.1605 & 0.00252 & 0.00431 \\
5 & $-0.3671$ & $-0.0458$ & 0.0490 & 0.0577 & $-0.0448$ & 0.0483 & 0.0572 & 0.00100 & 0.00184 \\
25 & $-0.1912$ & $-0.0059$ & 0.0210 & 0.0262 & $-0.0079$ & 0.0218 & 0.0273 & $-0.00201$ & 0.00206 \\
50 & $-0.0009$ & $-0.0027$ & 0.0182 & 0.0235 & $-0.0028$ & 0.0193 & 0.0248 & $-0.00014$ & 0.00149 \\
75 & 0.1874 & 0.0030 & 0.0225 & 0.0273 & 0.0052 & 0.0234 & 0.0283 & 0.00220 & 0.00226 \\
95 & 0.3632 & 0.0317 & 0.0375 & 0.0450 & 0.0308 & 0.0369 & 0.0445 & $-0.00086$ & 0.00143 \\
99 & 0.3803 & 0.1035 & 0.1036 & 0.1127 & 0.1015 & 0.1016 & 0.1112 & $-0.00203$ & 0.00285 \\
\addlinespace[3pt]
\multicolumn{10}{l}{$K^*=10$, $n=175$} \\
1 & $-0.4240$ & $-0.1217$ & 0.1217 & 0.1358 & $-0.1180$ & 0.1180 & 0.1332 & 0.00372 & 0.00469 \\
5 & $-0.3908$ & $-0.0237$ & 0.0363 & 0.0434 & $-0.0227$ & 0.0361 & 0.0433 & 0.00104 & 0.00194 \\
25 & $-0.2186$ & 0.0223 & 0.0272 & 0.0337 & 0.0198 & 0.0264 & 0.0329 & $-0.00250$ & 0.00258 \\
50 & $-0.0100$ & 0.0057 & 0.0223 & 0.0274 & 0.0060 & 0.0236 & 0.0293 & 0.00030 & 0.00205 \\
75 & 0.2120 & $-0.0217$ & 0.0280 & 0.0343 & $-0.0192$ & 0.0267 & 0.0332 & 0.00247 & 0.00257 \\
95 & 0.3855 & 0.0094 & 0.0305 & 0.0384 & 0.0081 & 0.0307 & 0.0387 & $-0.00134$ & 0.00202 \\
99 & 0.4108 & 0.0725 & 0.0746 & 0.0889 & 0.0692 & 0.0721 & 0.0873 & $-0.00324$ & 0.00403 \\
\addlinespace[3pt]
\multicolumn{10}{l}{$K^*=30$, $n=165$} \\
1 & $-0.4386$ & $-0.1240$ & 0.1256 & 0.1453 & $-0.1249$ & 0.1265 & 0.1475 & $-0.00086$ & 0.00402 \\
5 & $-0.3885$ & $-0.0271$ & 0.0371 & 0.0460 & $-0.0274$ & 0.0376 & 0.0468 & $-0.00029$ & 0.00181 \\
25 & $-0.2192$ & 0.0184 & 0.0268 & 0.0329 & 0.0170 & 0.0265 & 0.0326 & $-0.00141$ & 0.00159 \\
50 & $-0.0041$ & 0.0008 & 0.0191 & 0.0243 & 0.0007 & 0.0198 & 0.0249 & $-0.00004$ & 0.00133 \\
75 & 0.2109 & $-0.0204$ & 0.0287 & 0.0345 & $-0.0190$ & 0.0284 & 0.0342 & 0.00146 & 0.00163 \\
95 & 0.3808 & 0.0163 & 0.0304 & 0.0396 & 0.0165 & 0.0308 & 0.0400 & 0.00013 & 0.00132 \\
99 & 0.4217 & 0.0673 & 0.0727 & 0.0854 & 0.0672 & 0.0729 & 0.0859 & $-0.00006$ & 0.00265 \\
\bottomrule
\end{tabular*}
\endgroup
\smallskip
\noindent\parbox{\textwidth}{\small\textit{Note.} The $n$ replications had converged fits under both priors. Mean truth is the average of the true finite-site percentile over these replications. Bias, mean absolute error (MAE) and RMSE compare each estimate with its true value. The last two columns summarize the paired difference (Dual-Anchor minus $\mathrm{Gamma}(1,1)$) in site-effect units. Absolute differences were taken before averaging.}
\end{table}


\section{Software Implementation}
\label{app:software}

\subsection{Purpose and Verification}
\label{subsec:G1-purpose}

We use \texttt{DPprior} 2.0.0 to calibrate Gamma hyperpriors and check their
count and weight implications. Parameters use the shape--rate convention.
The code below implements the example in main-text Section~\ref{subsec:worked-example}. The public
replication package provides the package source, analysis code and installation
instructions; Section~\ref{subsec:G7-repository} gives the repository and
rebuild commands.

The software uses the fixed-scale Dual-Anchor objective in
Section~\ref{subsec:dual-anchor} and Appendix~\ref{app:weights}. We recomputed
all affected calibrations after correcting a discrepancy between the earlier
software and supplement. Count and soft calibrations must converge and pass
verification. For the hard sensitivity, we also allow a verified solution
at the constraint boundary. We recompute moments at a higher quadrature order. The example
uses 160 quadrature nodes for fitting and 320 for verification. A failed
calibration is reported as such and is not replaced by an approximate result.

\subsection{A Complete Calibration Example}
\label{subsec:G2-overview}

For 50 sites we specify a count mean of five and a variance of ten. Run the code below from the replication repository root (Section~\ref{subsec:G7-repository}) with \texttt{DPprior} 2.0.0 installed. The code verifies the count calibration as in Algorithm~\ref{alg:tsmm}. If the size-biased majority probability under this prior exceeds 0.40, the wrapper evaluates 29 trade-off weights using the 0.25 soft target. Among the candidates that pass the checks in Algorithm~\ref{alg:dual-anchor} and give a majority probability of at most 0.40, the wrapper selects the largest $\lambda$. If none qualifies, the code stops. We fit the separate hard-bound prior, with an upper bound of 0.25, starting from the count-calibrated fit.

\begin{lstlisting}[literate=]
library(DPprior)
source("calibration/policy_helpers.R")
policy <- pcv_load_policy("calibration/policy.yml")
spec <- list(spec_id = "worked_J50_mu5_v10",
             J = 50L, mu_K = 5, var_K = 10)

fit <- DPprior_fit(J = spec$J, mu_K = spec$mu_K, var_K = spec$var_K,
                  method = "A2-MN", M = 160,
                  check_diagnostics = TRUE)
stopifnot(identical(fit$status, "converged"),
          isTRUE(fit$usable), isTRUE(fit$verified))
p_majority <- prob_wsb_exceeds(
  0.5, fit$parameters$a, fit$parameters$b)

if (p_majority > policy$dual$trigger_above) {
  selected <- pcv_soft_curve_select(fit, spec, policy)
  stopifnot(is.na(selected$rejection_code))
  primary <- selected$fit
  print(selected$curve)
} else {
  primary <- fit
}
primary$parameters[c("a", "b")]

hard <- pcv_fit_hard_sensitivity(fit, spec, policy)
print(hard$gate)
stopifnot(isTRUE(hard$gate$ok))
\end{lstlisting}

Use the returned shape and rate in the chosen DP mixture model. Before
fitting that model, report the achieved count moments and both size-biased
and largest-weight diagnostics (Appendix~\ref{app:weights}). Repeat calibration and
selection when varying the count or its uncertainty. A qualitative
confidence level maps to variance as described in Section~\ref{subsec:elicit}; it is an
input to examine, not a property learned from the data.

For this count target, we timed \texttt{DPprior\_a2\_newton} in \texttt{DPprior} 2.0.0 with R 4.6.0 on arm64 macOS (Darwin 25.6.0). We used 160 quadrature nodes for calibration and 320 for verification, a solver tolerance of $10^{-10}$ and a limit of 30 Newton iterations. Across three fresh R sessions, the median elapsed time for the first call after loading the package was 0.261 seconds (range 0.249--0.263). We then repeated the call three times in each session. Across these nine calls, the median was 0.097 seconds (range 0.089--0.113). We timed only the count calibration, excluding R startup, package loading, Dual-Anchor refinement and posterior model fitting.

\subsection{Worked Results and the Trade-Off Curve}
\label{subsec:G5-example}

Table~\ref{tab:G-worked} summarizes the three priors. The count-calibrated
prior gives a majority probability of 0.497, so the weight check triggers
refinement. The selected $\lambda=0.30$ lowers this probability to 0.398
and raises the prior expected count by about one. The hard bound requires
a larger change in the count. The two weight requirements therefore have
different consequences even though both analyses start from the same
count judgment.

\begin{table}[H]
\centering
\caption{Worked example: count calibration, Dual-Anchor and the hard sensitivity}
\label{tab:G-worked}
\small
\setlength{\tabcolsep}{4pt}
\renewcommand{\arraystretch}{1.15}
\begin{tabular}{@{}lrrrrr@{}}
\toprule
Prior & $a$ & $b$ & $\E(K_{50})$ & $\Var(K_{50})$ & $\Pr(W_{\mathrm{SB}}>.5)$ \\
\midrule
TSMM & 1.4082 & 1.0770 & 5.000 & 10.000 & 0.497 \\
Dual-Anchor & 2.3158 & 1.4204 & 5.884 & 9.738 & 0.398 \\
Dual-Anchor (hard) & 5.2745 & 2.3059 & 7.488 & 9.051 & 0.250 \\
\bottomrule
\end{tabular}
\end{table}

Figure~\ref{fig:G4-tradeoff} shows all 29 grid values: 0.01--0.20 in
increments of 0.01, then 0.25, 0.30 and 0.40--1.00 in increments of 0.10.
All points in this example pass numerical verification. Points above the
0.40 line are ineligible under the selection rule. The 0.25 target remains
unmet at the selected point; the hard solution in Table~\ref{tab:G-worked}
is a separate sensitivity analysis.

\begin{figure}[H]
\centering
\includegraphics[width=\textwidth]{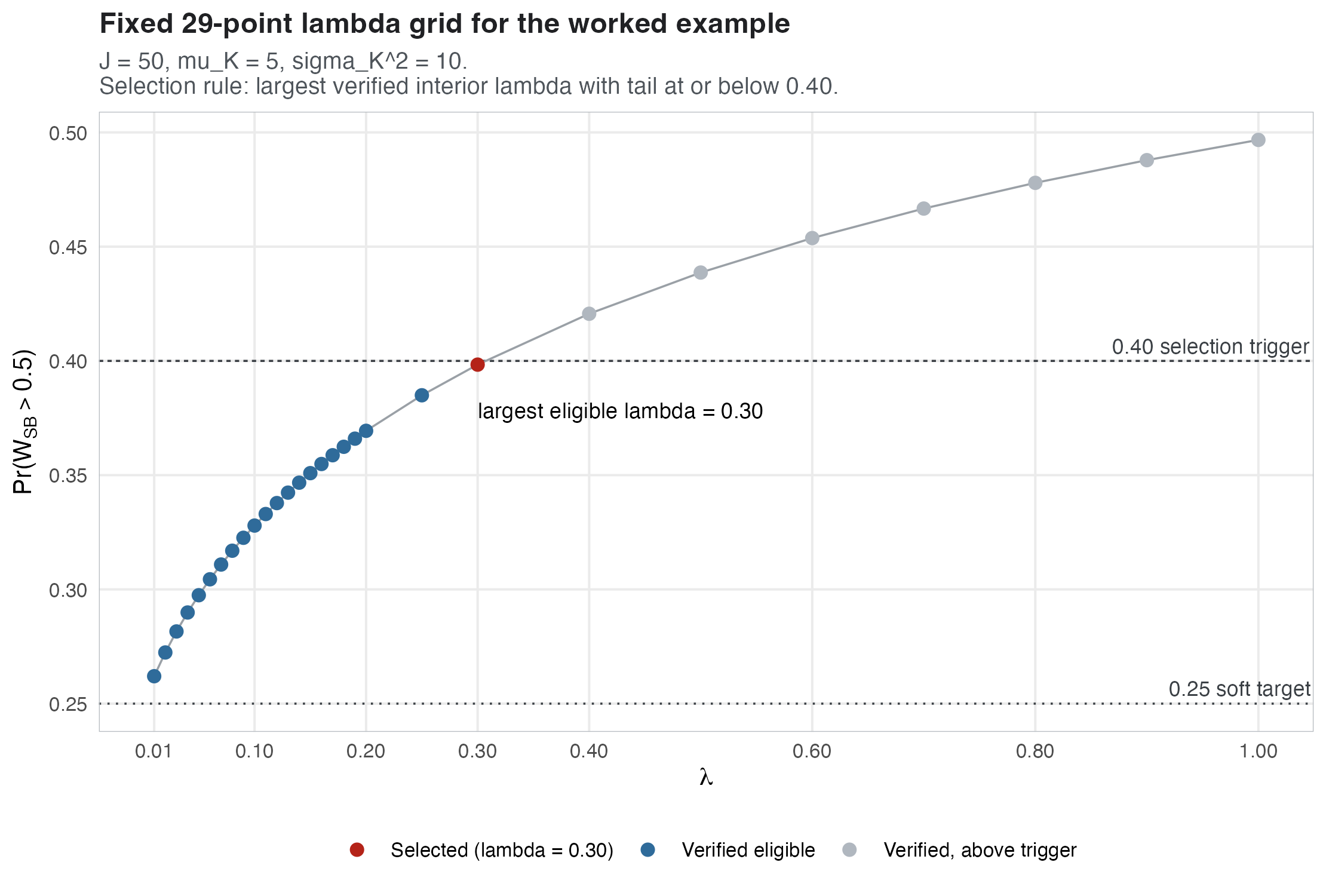}
\caption{All 29 soft Dual-Anchor solutions for $J=50$, $\mu_K=5$ and
$\sigma_K^2=10$. The selected point, $\lambda=0.30$, is the largest verified
interior grid value with majority probability at most 0.40. The lower line
marks the soft target of 0.25. Grey points are verified but exceed the
selection threshold.}
\label{fig:G4-tradeoff}
\end{figure}

\subsection{Browser Interface}
\label{subsec:G6-studio}

Applied researchers can use \emph{DPprior Studio} without installing R:
\begin{center}
\url{https://joonho.shinyapps.io/dpprior-studio/}.
\end{center}
The app takes the design size, a cluster-count judgment, and its uncertainty.
It fits the Gamma hyperprior and displays the induced cluster-count distribution,
size-biased weight and co-clustering diagnostics. It also runs the Dual-Anchor
sensitivity analysis and shows how the count and weight summaries change. The
app defaults are instructional starting points; an analysis should report the
settings used.

\subsection{Replication Package}
\label{subsec:G7-repository}

The replication package and its Quarto guide are available at
\begin{center}
\url{https://github.com/joonho112/DPprior-paper-replication},
\end{center}
with a rendered guide at
\begin{center}
\url{https://joonho112.github.io/DPprior-paper-replication/}.
\end{center}
For the simulation study, calibration examples and applications, the release at
\href{\replicationcommiturl}{commit \replicationcommit}
lists the inputs, scripts and outputs for the tables and figures included in that release. It records software versions, seeds and calibration diagnostics, and accounts for the planned simulation fits.

\Needspace{8\baselineskip}
To rebuild the \ifreplicationreleased tables and figures in the current release\else 41 tables and figures included in commit \replicationcommit\fi, run the following commands from the repository root.
\begin{lstlisting}[language={}]
Rscript 00_setup.R --profile rebuild
Rscript exhibits/00_build_all.R
Rscript verify_reproduction.R
\end{lstlisting}
\ifreplicationreleased
These commands rebuild the main-text and supplemental displays from the supplied summaries and verify them against the stored results. The index identifies the scripts, inputs and outputs for each display.
\else
These commands rebuild and check the 41 tables and figures in that release, including 68 selected numerical comparisons. The release predates the simulation comparisons of distributional recovery and percentile sensitivity (Section~\ref{subsec:sim-main}) and the comparisons of latent distributions and individual estimates in Figure~\ref{fig:app-effect-changes}. The school and study co-clustering and shrinkage comparisons in Figure~\ref{fig:app-tail-unit-pooling} and OSM Section~\ref{subsubsec:H-tail-pooling} were also added after this release. The revised displays also use the interval probabilities in Table~\ref{tab:E-elicitation-worksheet} and six additional density runs for the schools and studies (Section~\ref{subsec:H-distribution}).
\fi

The rebuild does not rerun the production posterior fits. Separate
drivers prepare source data and refit a selected simulation, application or
Rasch model. We tested the drivers with short runs. A production refit uses the original schedules and requires a fresh assessment of convergence.

The package obtains Project STAR through \texttt{mlmRev} and the
writing-to-learn data through \texttt{metadat}. \citet{air2024ocrs} released
OCRS through openICPSR as a public-use file (ICPSR 202181, version V1); the package
provides a script that converts the downloaded impact-analysis file into the
24-school input used here. The source OCRS file and documentation are obtained
from the official deposit rather than duplicated in the repository. The Rasch
refit similarly requires the user to obtain the response data through the Item
Response Warehouse. The supplied aggregate summaries suffice both to rebuild the results and to calculate the conditional WLEs in Section~\ref{subsubsec:H-irt-repro}.

\section{Applications: Data, Models, and Additional Results}
\label{app:applications}

\subsection{Overview}
\label{subsec:H1-purpose}

This appendix gives the data, model specifications, diagnostics and sensitivity
analyses for Section~\ref{sec:applications}. Sections~\ref{subsec:H2-data}--\ref{subsec:H5-results} concern schools and studies;
\Cref{subsec:H-irt} gives the Rasch analysis. We also report the Online
Credit Recovery Study (OCRS) as a descriptive example with little estimated
heterogeneity.

The main comparison is among $\mathrm{Gamma}(1,1)$, TSMM and Dual-Anchor. In the application-specific tables below Vague denotes $\mathrm{Gamma}(1,1)$. For the schools and studies we use DORO-Unif as an external reference and Dual-Anchor (hard) to examine sensitivity to the weight bound. For STAR, writing-to-learn and OCRS, we fitted five priors under each of two scale specifications (30 fits).
We also fitted the five priors to the STAR estimates from schoolwise OLS with HC3 standard errors. All 35 passed the
convergence criteria after at most one rerun. The table below gives the corresponding comparisons for the Rasch analysis.

\begin{center}
\small
\begin{tabular}{@{}>{\raggedright\arraybackslash}p{0.23\textwidth}>{\raggedright\arraybackslash}p{0.37\textwidth}>{\raggedright\arraybackslash}p{0.33\textwidth}@{}}
\toprule
Role & STAR and meta-analysis & Rasch \\
\midrule
Main comparison & Vague, TSMM, Dual-Anchor & Vague, TSMM, Dual-Anchor \\
External reference & DORO-Unif & Gamma(1, 3) reference \\
Sensitivity & Dual-Anchor (hard); wider scale priors; HC3 first stage (STAR) & Dual-Anchor (hard); stated count of 25 \\
Benchmark (different model) & Normal random effects (REML) & Rasch model with normally distributed abilities \\
\bottomrule
\end{tabular}
\end{center}

\subsection{Data and First-Stage Effect Estimates}
\label{subsec:H2-data}

\subsubsection{Project STAR}

The Tennessee Student/Teacher Achievement Ratio experiment randomized
kindergarten students within schools to small or regular-size classes
\citep{krueger1999experimental,nye2000effects}. The public source,
\texttt{mlmRev::star}, version 1.0-9 \citep{bates2026mlmrev}, contains
6,325 kindergarten records. We require complete records for class type,
school, reading score, mathematics score and teacher identifier. Excluding
539 records leaves 5,786 students in 337 classrooms and $J=79$ schools,
each with both treatment arms. We pool the regular and regular-with-aide
arms throughout.

For the students with complete records we standardize the reading and mathematics scores separately, average the two standardized scores and standardize the average. A positive contrast indicates higher achievement in small classes than in the pooled regular classes. We retain classroom dependence in the primary
school-specific contrasts by fitting
\[
 y_{isc}=\sum_{j=1}^{79}\beta_j\mathbf 1(s=j)
 +\sum_{j=1}^{79}\tau_j\mathbf 1(s=j)\,\mathrm{small}_{isc}
 +u_c+\epsilon_{isc},
\]
where $u_c$ is a classroom random intercept. We fit the model by REML with a fixed intercept and a fixed small-class coefficient for each school. Each school contrast is its estimated small-class coefficient. Its standard error is the square root of the corresponding diagonal entry of the fixed-effect covariance matrix.

In the second stage these 79 pairs supply $\widehat\tau_j$ and $\widehat{se}_j$ in Equation~\eqref{eq:model}. Conditional on the latent school effects $\tau_j$, we model the estimates as independent normal observations centered on those effects, with variances $\widehat{se}_j^2$. The DPM is the model for the school effects. We treat the first-stage standard errors as fixed, so the analysis neither fits the student outcomes jointly nor propagates uncertainty in the first-stage variance components. The mean school contrast is 0.191 standard deviations and the mean standard error is 0.428. The informativeness index is $I=0.107$, compared with 0.721 for
schoolwise OLS and 0.675 for OLS with HC3 standard errors. For the HC3 sensitivity analysis we replace both the school contrasts and their standard errors with schoolwise OLS estimates and HC3 standard errors. This analysis ignores classroom dependence.

\subsubsection{Bangert--Drowns and OCRS}

The writing-to-learn meta-analysis includes $J=48$ studies and 5,576
participants \citep{bangertdrowns2004effects}. The studies compare academic achievement under instruction emphasizing writing tasks and under conventional instruction. We take the bias-corrected standardized mean differences \texttt{yi} and their sampling variances \texttt{vi} from \texttt{metadat::}\allowbreak\texttt{dat.bangertdrowns2004}, version 1.6-0 \citep{viechtbauer2026metadat}. A positive effect favors the intervention. As the package documentation explains, the bias-corrected values differ slightly from the effects reported in the article. Only total sample sizes were reported, so the sampling variances assume equally sized groups. We use these estimates and variances as supplied, without reconstructing individual study records.

Each study supplies one estimate $\widehat\tau_j=\texttt{yi}_j$ and sampling variance $\widehat{se}_j^2=\texttt{vi}_j$ to the normal likelihood in Equation~\eqref{eq:model}. We treat the sampling variances as fixed and use a DPM for the latent study effects. The likelihood does not involve individual participant outcomes. We fit the separate normal random-effects benchmark by REML using \texttt{metafor} \citep{viechtbauer2010conducting}. Grade and intervention length are omitted from the likelihood, so the profiles in Section~\ref{subsec:H5-results} describe associations with these variables.

OCRS contributes $J=24$ complete-case school contrasts based on 1,063
subject-level observations. \citet{air2024ocrs} released the impact-analysis
data through openICPSR as a public-use file (ICPSR 202181, version V1).
The schoolwise contrast is a descriptive reconstruction,
not the trial's published covariate- and block-adjusted estimand. The
replication package supplies the preparation code and resulting school-level
summaries; the source file remains available from the official deposit.

\begin{table}[!htbp]
\centering
\caption{Normal random-effects benchmarks for the application inputs.}
\label{tab:H1-data}
\small
\begin{tabular}{@{}lrcccc@{}}
\toprule
Data & $J$ & Pooled effect [95\% CI] &
$\widehat\tau^2$ & $I^2$ (\%) & 95\% prediction interval \\
\midrule
STAR & 79 & 0.194 [0.097, 0.291] & 0.0215 & 11.2 & [-0.110, 0.497] \\
Bangert--Drowns & 48 & 0.222 [0.132, 0.312] & 0.0499 & 58.4 & [-0.225, 0.669] \\
OCRS & 24 & 0.037 [-0.077, 0.152] & 0.0000 & 0.0 & [-0.077, 0.152] \\
\bottomrule
\end{tabular}
\end{table}

Table~\ref{tab:H1-data} describes continuous heterogeneity under a normal
random-effects model; it does not estimate a cluster count.

\subsection{Priors and Additional Primary-Scale Results}
\label{subsec:H3-priors}
\leavevmode\par

\begin{table}[H]
\centering
\caption{Additional posterior summaries for STAR and Bangert--Drowns.}
\label{tab:H2-priors}
\small
\setlength{\tabcolsep}{7pt}
\begin{tabular}{@{}lccc@{}}
\toprule
Hyperprior & $\Pr(W_{\max}>.5\mid y)$ & $\E(.5^\alpha\mid y)$ & \shortstack{Intervals\\excluding zero} \\
\midrule
\multicolumn{4}{@{}l}{\textit{STAR}} \\
Vague & 0.963 & 0.797 & 3/79 \\
DORO-Unif & 0.935 & 0.679 & 3/79 \\
TSMM & 0.910 & 0.588 & 2/79 \\
Dual-Anchor & 0.804 & 0.419 & 3/79 \\
Dual-Anchor (hard) & 0.601 & 0.250 & 3/79 \\
\multicolumn{4}{@{}l}{\textit{Bangert--Drowns}} \\
Vague & 0.963 & 0.785 & 13/48 \\
DORO-Unif & 0.941 & 0.676 & 13/48 \\
TSMM & 0.958 & 0.771 & 15/48 \\
Dual-Anchor & 0.863 & 0.520 & 13/48 \\
Dual-Anchor (hard) & 0.650 & 0.294 & 13/48 \\
\bottomrule
\end{tabular}
\par\medskip
\begin{minipage}{\textwidth}\small\textit{Note.}
$W_{\max}$ is the largest population atom weight, including unoccupied mass.
The last column counts strict equal-tail 95 percent unit intervals.
DORO-Unif is an external reference and Dual-Anchor (hard) a sensitivity
analysis. Each fit contributes 16,000 retained draws.
\end{minipage}
\end{table}

The primary DPM gives the grand mean a $\mathcal N(0,1)$ prior and the
base-measure variance of component means a uniform prior on $(0,1)$.
Within-component variances have separate inverse-gamma priors,
$\sigma_k^2\sim\mathrm{InvGamma}(s_1,s_2)$, with
$s_1,s_2\sim\mathrm{InvGamma}(2,1)$, as in \Cref{subsec:F4-mcmc}.
For STAR we expected about five clusters and assigned the range from 2 to 8 a probability of 80 percent.
For Bangert--Drowns and OCRS we expected about four clusters and assigned the range from 1 to 8 a probability of 80 percent. The fixed-grid
policy selects $\lambda=0.20$, 0.13 and 0.25, respectively. The STAR hard
sensitivity has $\alpha$ effectively fixed at 2.

\begin{table}[H]
\centering
\caption{Prior cluster counts and size-biased weight probabilities in the three applications}
\label{tab:H-all-priors}
\small
\setlength{\tabcolsep}{4pt}
\begin{adjustbox}{max width=\linewidth}
\begin{tabular}{lcccc}
\toprule
Prior & $(a,b)$ & $\E(K_J)$ [90\% interval] & $\Pr(W_{\mathrm{SB}}>0.5)$ & $\Pr(W_{\mathrm{SB}}>0.9)$ \\
\midrule
\addlinespace[5pt]\multicolumn{5}{l}{\textit{Project STAR (79 schools)}}\\
$\mathrm{Gamma}(1,1)$ & (1.00, 1.00) & 4.61 [1, 11] & 0.591 & 0.303 \\
TSMM & (5.13, 4.94) & 5.00 [2, 9] & 0.510 & 0.140 \\
Dual-Anchor & (15.77, 11.11) & 6.24 [3, 10] & 0.385 & 0.051 \\
DORO-Unif & (2.33, 2.18) & 4.99 [1, 10] & 0.527 & 0.187 \\
Dual-Anchor (hard) & $\alpha \approx 2.00$ & 7.93 [4, 12] & 0.250 & 0.010 \\
\addlinespace[5pt]\multicolumn{5}{l}{\textit{Writing-to-learn (48 studies)}}\\
$\mathrm{Gamma}(1,1)$ & (1.00, 1.00) & 4.12 [1, 10] & 0.591 & 0.303 \\
TSMM & (1.15, 1.22) & 4.00 [1, 9] & 0.597 & 0.297 \\
Dual-Anchor & (3.78, 2.50) & 5.66 [2, 10] & 0.395 & 0.084 \\
DORO-Unif & (2.81, 3.19) & 3.99 [1, 8] & 0.576 & 0.217 \\
Dual-Anchor (hard) & (11.56, 5.44) & 7.15 [3, 12] & 0.250 & 0.017 \\
\addlinespace[5pt]\multicolumn{5}{l}{\textit{Vocabulary checklist (500 respondents)}}\\
$\mathrm{Gamma}(1,1)$ & (1.00, 1.00) & 6.44 [1, 17] & 0.591 & 0.303 \\
TSMM & (2.11, 3.08) & 5.00 [1, 11] & 0.652 & 0.308 \\
Dual-Anchor & (20.30, 14.51) & 8.75 [4, 14] & 0.388 & 0.050 \\
$\mathrm{Gamma}(1,3)$ reference & (1.00, 3.00) & 3.03 [1, 8] & 0.812 & 0.566 \\
Dual-Anchor (hard) & $\alpha \approx 2.00$ & 11.59 [7, 17] & 0.250 & 0.010 \\
\bottomrule
\end{tabular}
\end{adjustbox}
\par\medskip
\begin{minipage}{\linewidth}\small\textit{Note.} The Gamma parameters are shape and rate. A randomly selected unit belongs to a component of population weight $W_{\mathrm{SB}}$. The selected Dual-Anchor weights $\lambda$ are 0.20, 0.13 and 0.05 for STAR, writing-to-learn and vocabulary. Dual-Anchor (hard) imposes a separate bound of 0.25 on the majority probability.\end{minipage}
\end{table}

We supplied the illustrative count judgments for STAR and Bangert--Drowns. There is no independent record of expert elicitation or of when the judgments were made relative to data inspection, so we cannot describe them as data-blind preregistered beliefs. We fixed the grid and selection policy before posterior simulation. OSM Section~\ref{subsec:E6-guidelines} gives the count moments and interval probabilities attained and distinguishes the 0.40 trigger, the 0.25 soft target and the separate 0.25 hard bound.

Tables~\ref{tab:H-all-priors} and~\ref{tab:H-all-posteriors} extend main-text Table~\ref{tab:applications} to the external references and hard-bound priors. Table~\ref{tab:H2-priors} adds the posterior probability
of population dominance and the alpha-only functional. The latter averages
the conditional size-biased tail over posterior draws of $\alpha$; it omits
the fitted allocations and residual mass, and is therefore not a posterior
probability of either observed or population dominance.

\begin{table}[H]
\centering
\caption{Posterior cluster counts, weights and individual estimates under all five priors}
\label{tab:H-all-posteriors}
\small
\setlength{\tabcolsep}{4pt}
\begin{adjustbox}{max width=\linewidth}
\begin{tabular}{lccccc}
\toprule
Prior & \shortstack{Cluster count\\Mean [95\% interval]} & \shortstack{$\Pr(K_J=1$\\$\mid y)$} & \shortstack{Largest\\observed share} & \shortstack{Largest\\pop.\ weight} & \shortstack{Unit\\RMS} \\
\midrule
\addlinespace[5pt]\multicolumn{6}{l}{\textit{Project STAR (79 schools)}}\\
$\mathrm{Gamma}(1,1)$ & 1.96 [1, 7] & 0.560 & 0.920 & 0.917 & 0.0000 \\
TSMM & 3.12 [1, 8] & 0.200 & 0.825 & 0.818 & 0.0071 \\
Dual-Anchor & 4.47 [1, 9] & 0.064 & 0.724 & 0.715 & 0.0088 \\
DORO-Unif & 2.57 [1, 8] & 0.354 & 0.868 & 0.863 & 0.0032 \\
Dual-Anchor (hard) & 6.57 [2, 12] & 0.006 & 0.590 & 0.579 & 0.0114 \\
\addlinespace[5pt]\multicolumn{6}{l}{\textit{Writing-to-learn (48 studies)}}\\
$\mathrm{Gamma}(1,1)$ & 1.89 [1, 6] & 0.548 & 0.910 & 0.904 & 0.0000 \\
TSMM & 1.96 [1, 6] & 0.525 & 0.907 & 0.901 & 0.0025 \\
Dual-Anchor & 3.31 [1, 8] & 0.190 & 0.785 & 0.773 & 0.0071 \\
DORO-Unif & 2.33 [1, 6] & 0.361 & 0.868 & 0.859 & 0.0030 \\
Dual-Anchor (hard) & 5.29 [1, 11] & 0.034 & 0.632 & 0.615 & 0.0103 \\
\addlinespace[5pt]\multicolumn{6}{l}{\textit{Vocabulary checklist (500 respondents)}}\\
$\mathrm{Gamma}(1,1)$ & 4.43 [2, 10] & 0.006 & 0.860 & 0.859 & 0.0000 \\
TSMM & 4.30 [2, 9] & 0.003 & 0.867 & 0.866 & 0.0132 \\
Dual-Anchor & 6.73 [3, 12] & 0.000 & 0.797 & 0.796 & 0.0123 \\
$\mathrm{Gamma}(1,3)$ reference & 3.51 [2, 8] & 0.010 & 0.886 & 0.885 & 0.0067 \\
Dual-Anchor (hard) & 8.99 [5, 14] & 0.000 & 0.748 & 0.745 & 0.0142 \\
\bottomrule
\end{tabular}
\end{adjustbox}
\par\medskip
\begin{minipage}{\linewidth}\small\textit{Note.} The two share and weight columns give posterior means. The largest observed share is the fraction of sampled units in the largest cluster. The largest population weight includes the possibility of unseen components. Under the selected vocabulary Dual-Anchor prior, no retained draw assigned all 500 respondents to one cluster. Unit RMS gives the root mean squared change in individual posterior means from $\mathrm{Gamma}(1,1)$, in achievement SDs for STAR, standardized mean differences for writing-to-learn and logits for vocabulary. Vocabulary changes use raw-score means weighted by the numbers of respondents.\end{minipage}
\end{table}

Among the primary-scale fits for STAR, writing-to-learn and OCRS, the largest root mean square change in individual posterior means
from $\mathrm{Gamma}(1,1)$ was 0.0114 achievement SDs in STAR. In each comparison, at most two units changed
in whether their 95 percent posterior interval excluded zero. These small changes in unit estimates
coexist with larger changes in posterior cluster counts.

\subsection{MCMC Schedule and Convergence}
\label{subsec:H4-mcmc}

For the schools and studies, we fit the mixtures in NIMBLE \citep{devalpine2017programming}, sampling the allocations under the Chinese restaurant process representation.
Each fit starts with four chains of 12,000 iterations, including 4,000
warmup. A fit that fails the criteria in \Cref{subsec:F4-mcmc} receives one
rerun of 120,000 iterations per chain with 60,000 warmup. We retain 4,000
draws per chain for the reported summaries. Twenty-five fits passed
initially and ten after rerunning; none was excluded.
Table~\ref{tab:H3-conv} reports diagnostics for the 15 primary-scale fits.

\begin{table}[!htbp]
\centering
\caption{Primary-scale convergence diagnostics.}
\label{tab:H3-conv}
\small
\begin{tabular}{@{}lrrrr@{}}
\toprule
Hyperprior & Final attempt & Max $\widehat R$ & Min bulk ESS & Min tail ESS \\
\midrule
\multicolumn{5}{@{}l}{\textit{STAR}} \\
Vague & 2 & 1.0016 & 4315 & 3884 \\
DORO-Unif & 2 & 1.0009 & 4868 & 3254 \\
TSMM & 1 & 1.0034 & 575 & 561 \\
Dual-Anchor & 2 & 1.0016 & 3223 & 2306 \\
Dual-Anchor (hard) & 2 & 1.0018 & 1841 & 1011 \\
\multicolumn{5}{@{}l}{\textit{Bangert--Drowns}} \\
Vague & 1 & 1.0038 & 1010 & 1195 \\
DORO-Unif & 2 & 1.0008 & 7353 & 4005 \\
TSMM & 2 & 1.0003 & 7363 & 5338 \\
Dual-Anchor & 1 & 1.0073 & 783 & 574 \\
Dual-Anchor (hard) & 1 & 1.0093 & 875 & 869 \\
\multicolumn{5}{@{}l}{\textit{OCRS}} \\
Vague & 1 & 1.0038 & 892 & 680 \\
DORO-Unif & 1 & 1.0041 & 777 & 526 \\
TSMM & 1 & 1.0037 & 1127 & 766 \\
Dual-Anchor & 1 & 1.0092 & 856 & 556 \\
Dual-Anchor (hard) & 2 & 1.0019 & 3730 & 1951 \\
\bottomrule
\end{tabular}
\par\medskip
\begin{minipage}{\textwidth}\small\textit{Note.}
Extrema are over all quantities checked for each fit: 90 for STAR, 59 for
Bangert--Drowns and 35 for OCRS. All entries pass the prespecified thresholds;
they are not diagnostics for $\alpha$ or $K_J$ alone.
\end{minipage}
\end{table}

\FloatBarrier
\subsection{Sensitivity Analyses and Descriptive Profiles}
\label{subsec:H5-results}

\subsubsection{Model Scale and STAR First-Stage Uncertainty}

The wider-scale specification jointly changes the grand-mean prior to
$\mathcal N(0,100)$ (variance 100) and the uniform prior on the
base-measure variance of component means to $(0,10)$. This shifts the
posterior toward fewer clusters in all 15 matched comparisons
(Table~\ref{tab:H4-sens}). The absolute change in posterior mean $K_J$
reaches 1.834, although changes in unit estimates remain small.

\begin{table}[!htbp]
\centering
\caption{Changes under the wider model scale and STAR HC3 first stage.}
\label{tab:H4-sens}
\small
\begin{tabular}{@{}lrrrr@{}}
\toprule
Hyperprior & $\Delta\E(K_J)$ & $\Delta\Pr(K_J=1)$ & \shortstack{Unit-mean\\RMSD} & \shortstack{Changed\\intervals/$J$} \\
\midrule
\multicolumn{5}{@{}l}{\textit{Wide scale: STAR}} \\
Vague & -0.744 & +0.282 & 0.0043 & 2/79 \\
DORO-Unif & -1.067 & +0.351 & 0.0036 & 0/79 \\
TSMM & -1.449 & +0.395 & 0.0069 & 1/79 \\
Dual-Anchor & -1.616 & +0.255 & 0.0044 & 1/79 \\
Dual-Anchor (hard) & -1.030 & +0.064 & 0.0029 & 0/79 \\
\multicolumn{5}{@{}l}{\textit{Wide scale: Bangert--Drowns}} \\
Vague & -0.667 & +0.291 & 0.0076 & 1/48 \\
DORO-Unif & -0.905 & +0.341 & 0.0052 & 0/48 \\
TSMM & -0.695 & +0.289 & 0.0053 & 1/48 \\
Dual-Anchor & -1.420 & +0.363 & 0.0065 & 0/48 \\
Dual-Anchor (hard) & -1.834 & +0.209 & 0.0032 & 0/48 \\
\multicolumn{5}{@{}l}{\textit{Wide scale: OCRS}} \\
Vague & -0.345 & +0.181 & 0.0021 & 0/24 \\
DORO-Unif & -0.700 & +0.314 & 0.0054 & 0/24 \\
TSMM & -0.378 & +0.194 & 0.0041 & 0/24 \\
Dual-Anchor & -0.771 & +0.304 & 0.0079 & 0/24 \\
Dual-Anchor (hard) & -1.569 & +0.408 & 0.0060 & 0/24 \\
\multicolumn{5}{@{}l}{\textit{STAR HC3: STAR}} \\
Vague & +0.239 & -0.053 & 0.1744 & 19/79 \\
DORO-Unif & +0.560 & -0.099 & 0.1738 & 19/79 \\
TSMM & +0.314 & -0.035 & 0.1712 & 20/79 \\
Dual-Anchor & +0.455 & -0.023 & 0.1709 & 19/79 \\
Dual-Anchor (hard) & +0.309 & -0.003 & 0.1692 & 19/79 \\
\bottomrule
\end{tabular}
\par\medskip
\begin{minipage}{\textwidth}\small\textit{Note.}
Each comparison is sensitivity minus primary for the same application and
hyperprior. RMSD compares unit posterior means between the two specifications;
the last column counts changes in whether the strict equal-tail 95 percent
unit interval excludes zero. Both sensitivities change inputs beyond the
concentration prior.
\end{minipage}
\end{table}

Using STAR's schoolwise OLS estimates with HC3 standard errors changes both
the school estimates and their uncertainty inputs. Across the five priors,
the root mean squared changes in unit posterior means are 0.169--0.174,
about twenty times the primary-scale Dual-Anchor--Vague difference.
With the schoolwise OLS estimates and HC3 standard errors as inputs, the 95 percent posterior intervals excluded zero for 22 of the 79 schools under each prior. Compared with the classroom random-intercept analysis, this changed which intervals excluded zero for 19 or 20 schools, depending on the prior.
Schoolwise HC3 standard errors do not account for classroom dependence, so
this comparison also changes the assumptions underlying the first-stage
uncertainty estimates.

\subsubsection{OCRS as a Descriptive Illustration}

\begin{table}[!htbp]
\centering
\caption{OCRS calibration and primary-scale posterior results.}
\label{tab:H-ocrs}
\small
\begin{tabular}{@{}lcccc@{}}
\toprule
Hyperprior & $(a,b)$ & $\E(K_{24}\mid y)$ [95\%] & $\Pr(K_{24}=1\mid y)$ & \shortstack{Unit-mean\\RMSD} \\
\midrule
Vague & (1.00, 1.00) & 1.47 [1, 4] & 0.717 & 0.0000 \\
DORO-Unif & (2.14, 1.78) & 1.92 [1, 5] & 0.519 & 0.0043 \\
TSMM & (0.88, 0.65) & 1.51 [1, 5] & 0.713 & 0.0027 \\
Dual-Anchor & (1.68, 0.94) & 2.07 [1, 7] & 0.503 & 0.0048 \\
Dual-Anchor (hard) & (3.48, 1.42) & 3.25 [1, 9] & 0.248 & 0.0077 \\
\bottomrule
\end{tabular}
\par\medskip
\begin{minipage}{\textwidth}\small\textit{Note.}
The Gamma hyperprior uses shape $a$ and rate $b$. RMSD compares unit
posterior means with Vague. These school contrasts describe the available
complete-case sample and do not reproduce the trial's adjusted estimand.
\end{minipage}
\end{table}

The OCRS benchmark estimates little heterogeneity: its pooled effect is
0.037 (95 percent CI $[-0.077,0.152]$), with $\widehat\tau^2$ and $I^2$
approximately zero. The DPM nevertheless gives posterior mean $K_{24}$
from 1.47 under Vague to 2.07 under Dual-Anchor and 3.25 under the hard
sensitivity (Table~\ref{tab:H-ocrs}). No unit interval excludes zero.
Thus cluster counts can remain sensitive to the concentration prior even
when the normal random-effects benchmark estimates little heterogeneity.

\subsubsection{Descriptive Profiles}

\begin{table}[!htbp]
\centering
\caption{STAR location profiles under Vague and Dual-Anchor.}
\label{tab:H5-location}
\small
\begin{tabular}{@{}lrrcc@{}}
\toprule
Location & Schools & Observed mean & Vague [95\%] & Dual-Anchor [95\%] \\
\midrule
Inner city & 16 & 0.230 & 0.206 [0.078, 0.336] & 0.208 [0.081, 0.341] \\
Rural & 38 & 0.188 & 0.188 [0.082, 0.292] & 0.188 [0.084, 0.293] \\
Suburban & 18 & 0.151 & 0.186 [0.062, 0.312] & 0.184 [0.057, 0.309] \\
Urban & 7 & 0.218 & 0.197 [0.030, 0.367] & 0.196 [0.025, 0.368] \\
\bottomrule
\end{tabular}
\end{table}

Table~\ref{tab:H5-location} compares STAR's location-specific mean effects
under Vague and Dual-Anchor. We average the school effects within each
location for every posterior draw. The location differences are small and
the intervals overlap substantially.

\begin{table}[!htbp]
\centering
\caption{Descriptive Bangert--Drowns slopes by grade and intervention length.}
\label{tab:H6-moderator}
\small
\begin{tabular}{@{}lcc@{}}
\toprule
Hyperprior & Grade slope [95\%] & Length slope [95\%] \\
\midrule
Vague & 0.0108 [-0.0278, 0.0520] & 0.0090 [0.0006, 0.0183] \\
DORO-Unif & 0.0107 [-0.0292, 0.0516] & 0.0090 [0.0005, 0.0181] \\
TSMM & 0.0107 [-0.0286, 0.0523] & 0.0089 [0.0005, 0.0180] \\
Dual-Anchor & 0.0110 [-0.0290, 0.0532] & 0.0091 [0.0006, 0.0182] \\
Dual-Anchor (hard) & 0.0115 [-0.0283, 0.0542] & 0.0092 [0.0006, 0.0183] \\
\bottomrule
\end{tabular}
\end{table}

For Bangert--Drowns, we regress the study effects in each posterior draw
separately on grade and intervention length. Table~\ref{tab:H6-moderator}
reports the resulting unadjusted slopes. All 48 studies have grade data;
46 have intervention length. The observed-effect slopes are 0.0423 and
0.0200, respectively. The posterior slope for length is positive in about
98 percent of draws under each prior. These are descriptive associations;
the covariates neither define the mixture components nor enter the fitted
DPM as moderators.

\FloatBarrier
\subsection{Rasch Analysis of a Vocabulary Checklist}
\label{subsec:H-irt}

We examine how the concentration prior affects the number and sizes of
clusters, the fitted latent density and respondent estimates.
Respondents in the same cluster share a component mean and variance,
but their latent traits need not be equal.

\subsubsection{Data}
\label{subsubsec:H-irt-data}

In 2016, the Open-Source Psychometrics Project administered the vocabulary
checklist online after an interactive version of the Generic Conspiracist
Beliefs Scale developed by \citet{brotherton2013measuring}
\citep{openpsychometrics2016gcbs}. The Item Response Warehouse distributes
the responses as \texttt{gcbs\_brotherton\_2013\_vcl} under a Creative
Commons Attribution 4.0 license \citep{domingue2025irw}.
The file contains 2,495 respondents and 16 entries with no missing responses.
A checked entry is coded 1 and an unchecked entry 0.

We draw 500 respondents at random with seed 81770971, then exclude the three
non-words \texttt{VCL6}, \texttt{VCL9} and \texttt{VCL12}, which the
codebook identifies as validity checks \citep{openpsychometrics2016gcbs}.
The sample is therefore unchanged by excluding these items. We analyze the
13 real words, retain scores of both 0 and 13, and impose no age restriction.
The sample size is the largest supported by the calibration tables used here;
the method itself has no such limit.

\begin{table}[!htbp]
\centering
\caption{Checklist entries and proportions checked.}
\label{tab:H-irt-items}
\small
\begin{tabular}{@{}lllrr@{}}
\toprule
Entry & Word & Status & Full file & Sample \\
\midrule
\texttt{VCL1} & boat & analyzed & 0.972 & 0.976 \\
\texttt{VCL2} & incoherent & analyzed & 0.942 & 0.940 \\
\texttt{VCL3} & pallid & analyzed & 0.549 & 0.552 \\
\texttt{VCL4} & robot & analyzed & 0.976 & 0.976 \\
\texttt{VCL5} & audible & analyzed & 0.953 & 0.950 \\
\texttt{VCL6} & cuivocal & excluded (non-word) & 0.105 & 0.098 \\
\texttt{VCL7} & paucity & analyzed & 0.273 & 0.270 \\
\texttt{VCL8} & epistemology & analyzed & 0.453 & 0.462 \\
\texttt{VCL9} & florted & excluded (non-word) & 0.059 & 0.052 \\
\texttt{VCL10} & decide & analyzed & 0.967 & 0.958 \\
\texttt{VCL11} & pastiche & analyzed & 0.313 & 0.312 \\
\texttt{VCL12} & verdid & excluded (non-word) & 0.167 & 0.144 \\
\texttt{VCL13} & abysmal & analyzed & 0.788 & 0.810 \\
\texttt{VCL14} & lucid & analyzed & 0.926 & 0.932 \\
\texttt{VCL15} & betray & analyzed & 0.971 & 0.972 \\
\texttt{VCL16} & funny & analyzed & 0.983 & 0.980 \\

\bottomrule
\end{tabular}
\par\medskip
\begin{minipage}{\textwidth}\small\textit{Note.}
The full file contains 2,495 respondents and the analysis sample 500.
The three excluded entries are non-words according to the instrument codebook.
\end{minipage}
\end{table}

Eight analyzed words are checked by more than 92 percent of the sample;
only paucity, pastiche and epistemology are checked by fewer than half
(Table~\ref{tab:H-irt-items}). Every raw score from 0 to 13 occurs.
The median is 10, with quartiles 9 and 12; 63 respondents check all 13 words
and three check none. The 13-item KR-20 reliability is 0.771. The WLE
reliability below measures person separation on the fitted Rasch scale.

\subsubsection{Model}
\label{subsubsec:H-irt-model}

For respondent $j=1,\ldots,500$ and item $i=1,\ldots,13$, we fit
\begin{equation}
\label{eq:H-irt-model}
\begin{aligned}
 y_{ij}\mid\eta_j,\beta_i
   &\sim\Bernoulli\{\mathrm{logit}^{-1}(\eta_j-\beta_i)\},\\
 \eta_j\mid\mu_j,\sigma_j^2&\sim\mathcal N(\mu_j,\sigma_j^2),\\
 (\mu_j,\sigma_j^2)\mid G&\overset{\mathrm{iid}}{\sim}G,\\
 G\mid\alpha&\sim\DP(\alpha,G_0),\qquad
 \alpha\sim\mathrm{Gamma}(a,b).
\end{aligned}
\end{equation}
The base measure is
$G_0=\mathcal N(0,2)\times\mathrm{InvGamma}(2.01,1.01)$ on the component
mean and variance. Item difficulties are
$\beta_i=\beta_i^\ast-\bar\beta^\ast$, with
$\beta_i^\ast\sim\mathcal N(0,3)$. Normal distributions are parameterized
by variance, and the Gamma hyperprior by shape and rate.

The sum-to-zero item constraint fixes location and the unit Rasch slope fixes
scale \citep{paganin2023computational}. These constraints do not resolve
the separate finite-item identification problem. With finitely many Rasch
items, only finitely many functionals of the latent distribution are
identified; identification of the full distribution requires infinitely
many items \citep[Theorems~5--6]{sanmartin2011bayesian}. We therefore
interpret $K_{500}$ as the number of clusters among these examinees.
The clusters need not correspond to substantive ability populations. All fits, including the Gaussian benchmark,
use the same item-anchored scale without rescaling.

The fits use DPMirt 0.2.0 and NIMBLE 1.4.2 under R 4.6.0
\citep{dpmirt2026package,devalpine2017programming}, with DPprior 2.0.0
for calibration. The component capacity is $M=100$, above the sensitivity
prior's 95th percentile of 39 clusters. Across the final fits,
the largest sampled count is 30, and no capacity warning occurs. All
hyperprior comparisons keep the base measure and item priors fixed.

\subsubsection{Weighted-Likelihood Reliability}
\label{subsubsec:H-irt-wle}

We fix the item difficulties at their posterior means from the prespecified
Dual-Anchor fit and compute Warm's weighted-likelihood estimates (WLEs)
\citep{warm1989weighted}. For raw score $r$, the estimate solves
\[
 r-\sum_i p_i(\theta)+\frac{I'(\theta)}{2I(\theta)}=0,
 \qquad
 I(\theta)=\sum_i p_i(\theta)\{1-p_i(\theta)\},
\]
where $p_i(\theta)=\mathrm{logit}^{-1}(\theta-\beta_i)$ and $I'$ is the
derivative of $I$. The standard error is
$\mathrm{SE}(\widehat\theta)=I(\widehat\theta)^{-1/2}$. We compute
person-separation reliability as
\[
 1-\frac{\operatorname{mean}_j\{\mathrm{SE}(\widehat\theta_j)^2\}}
         {\operatorname{var}_j(\widehat\theta_j)},
\]
using the unbiased sample variance and equal respondent weights
\citep{adams2005reliability}. We use \texttt{TAM::tam.mml.wle2} in
TAM 4.3-25 \citep{tam2025package}; a direct solution of the estimating
equation reproduces its estimates and standard errors to numerical precision.

Reliability is 0.7121: the mean error variance is 1.3120 and the WLE variance
4.5568. It ranges from 0.7121 to 0.7125 using difficulties from the other
main-arm fits, and is 0.7153 using the Gaussian benchmark. Thus estimated
measurement error is about 29 percent of the across-person WLE variance in
this sample. The calculation conditions on the item estimates and does not
include their uncertainty or measure the precision of $K_{500}$.

\begin{table}[!htbp]
\centering
\caption{Weighted-likelihood estimates by raw score.}
\label{tab:H-irt-wle}
\small
\begin{tabular}{@{}rrrr@{}}
\toprule
Raw score & Respondents & WLE & Standard error \\
\midrule
0 & 3 & -5.053 & 1.580 \\
1 & 1 & -3.783 & 0.981 \\
2 & 2 & -3.078 & 0.821 \\
3 & 4 & -2.521 & 0.756 \\
4 & 2 & -2.015 & 0.734 \\
5 & 3 & -1.515 & 0.741 \\
6 & 9 & -0.980 & 0.776 \\
7 & 17 & -0.353 & 0.849 \\
8 & 52 & 0.487 & 0.969 \\
9 & 70 & 1.685 & 1.066 \\
10 & 116 & 2.848 & 1.030 \\
11 & 89 & 3.806 & 1.025 \\
12 & 69 & 4.785 & 1.138 \\
13 & 63 & 6.246 & 1.733 \\

\bottomrule
\end{tabular}
\par\medskip
\begin{minipage}{\textwidth}\small\textit{Note.}
WLEs use posterior mean item difficulties from the Dual-Anchor fit.
Warm's correction gives finite estimates at both extreme scores; no score
adjustment or model refit is used. The 63 perfect-score respondents have a
standard error of 1.733 logits, indicating limited resolution at the top of
the scale.
\end{minipage}
\end{table}

\FloatBarrier
\subsubsection{Prior Calibration}
\label{subsubsec:H-irt-priors}

For the 500 examinees we expected about five clusters and assigned the range from 2 to 10
a probability of 80 percent.
This illustrative judgment was made after
inspecting the data and fixed before fitting. The interval-width conversion
in Section~\ref{subsec:elicit} gives target variance 9.742. Matching this variance does not fix the probability of the stated interval. Table~\ref{tab:H-irt-priors} gives the resulting probability of 2--10 clusters.
The sensitivity judgment specifies an expected count of 25 and variance 60
(variance inflation factor 2.5).

For Vague we use $\mathrm{Gamma}(1,1)$. The Gamma(1, 3) reference follows the concentration hyperprior used for the two real-data analyses of \citet{paganin2023computational}. We use the same base measure here as in our other Rasch DPM fits. Although those authors describe $b$ as a scale, the moments in their Table~2 correspond to shape 1 and rate 3, the parameterization used here.
The DORO-Unif projection failed the target-family adequacy check at $J=500$
and was not fitted.

For the main judgment, TSMM has a size-biased majority probability of 0.652,
above the 0.40 trigger. Under the fixed grid and selection rule in
Section~\ref{subsec:dual-anchor}, the eligible values are $\lambda=0.02$ through 0.05.
Selecting the largest gives Dual-Anchor with expected count 8.75 and majority
probability 0.388. Before fitting, we accepted the increase from the stated mean of five clusters and chose Dual-Anchor for the main analysis, retaining TSMM for comparison. Neither the
judgment nor the selected weight was changed after examining the posteriors.

For the hard bound of 0.25, the count discrepancy decreases as the Gamma
prior concentrates. The solver reaches its domain boundary with $\alpha$
effectively fixed at 2 (standard deviation 0.001), giving expected count
11.59. This is reported only as a sensitivity analysis. Under the count-25
judgment, the TSMM majority probability is 0.042: the soft trigger does not
fire and the hard constraint is inactive, so all three calibrated roles
coincide.

\begin{table}[!htbp]
\centering
\caption{The six distinct Rasch hyperpriors and their count and weight implications.}
\label{tab:H-irt-priors}
\small
\setlength{\tabcolsep}{4pt}
\begin{tabular}{@{}lcccc@{}}
\toprule
Hyperprior & $(a,b)$ & $\E(K_{500})$ [90\%] & \shortstack{$\Pr(2\le K_{500}$\\$\le10)$} & \shortstack{$\Pr(W_{\mathrm{SB}}$\\${}>.5)$} \\
\midrule
Vague & (1, 1) & 6.44 [1, 17] & 0.684 & 0.591 \\
Gamma(1, 3) reference & (1, 3) & 3.03 [1, 8] & 0.676 & 0.812 \\
TSMM & (2.1062, 3.0789) & 5.00 [1, 11] & 0.849 & 0.652 \\
Dual-Anchor & (20.2968, 14.5089) & 8.75 [4, 14] & 0.736 & 0.388 \\
Dual-Anchor (hard) & $\alpha\approx2$ & 11.59 [7, 17] & 0.373 & 0.250 \\
Count 25 & (8.9037, 1.6256) & 25.00 [13, 39] & 0.015 & 0.042 \\
\bottomrule
\end{tabular}
\par\medskip
\begin{minipage}{\textwidth}\small\textit{Note.}
The first five rows are the main-arm roles. ``Count 25'' is the common
TSMM, Dual-Anchor and Dual-Anchor (hard) prior in the sensitivity arm. The
hard main-arm row uses the solver's Gamma parameters,
$a=3.27\times10^6$ and $b=1.63\times10^6$ (rounded); the compact notation
describes their effectively fixed-$\alpha$ limit.
The count interval is $[q_{.05},q_{.95}]$. The probabilities for $2\leq K_{500}\leq10$ refer to the interval in the main judgment. The count-25 judgment specifies no such interval.
\end{minipage}
\end{table}

\subsubsection{MCMC Schedule and Convergence}
\label{subsubsec:H-irt-convergence}
\label{subsubsec:H-irt-replication}

The two sets of five labeled fits contain only six distinct priors. Both arms use the same Vague prior and the same Gamma(1, 3) reference. In the count-25 arm TSMM, Dual-Anchor and Dual-Anchor (hard) also give the same prior. We ran all ten labeled fits
with four chains each. The four repeated-prior fits were exempt from reruns
under the original schedule and retain their initial diagnostics in
Table~\ref{tab:H-irt-diagnostics}. They do not supply the main-text results
or the posterior summaries below, which use the six converged distinct-prior
fits.

The initial schedule used 12,000 iterations per chain, 4,000 warmup and
thinning 2. The prespecified single rerun used 36,000 iterations, 12,000
warmup and thinning 6. No distinct-prior fit passed initially; only main-arm
TSMM passed after this rerun. The other five still failed mainly on the
largest observed allocation share: bulk effective sample sizes were
270--403 against a threshold of 400, with $\widehat R$ up to 1.027. Some
density tail effective sample sizes also remained below 400.

We therefore ran those five fits for 120,000 iterations with 60,000 warmup
and thinning 15. This third schedule departed from the prespecified
single-rerun policy. The data, priors, fitting code and diagnostic thresholds
were unchanged, and draws from earlier attempts were not pooled with later
ones. All five passed; none of the six distinct-prior fits was excluded.
Every final fit provides 4,000 draws per chain, or 16,000 in total.

We require rank-normalized, folded $\widehat R\le1.01$ and bulk and tail
effective sample sizes of at least 400 \citep{vehtari2021rank}, computed with
\texttt{posterior} 1.7.0. Checked quantities are $\alpha$, $K_{500}$,
the largest observed allocation share, the finite nonconstant log-likelihood,
13 item difficulties, 14 raw-score-group means, and density values where
the pooled mean density is at least 0.01. Individual traits and labeled atom
parameters are recorded separately. Respondents with the same raw score have
the same marginal posterior under the Rasch model
\citep{sanmartin2011bayesian}, which motivates the score-group checks.
For discrete quantities with tied tail quantiles, a non-finite tail ESS is
replaced by the smaller bulk ESS of the 5 and 95 percent tail indicators;
constant draws receive $\widehat R=1$. We also require finite draws,
counts below the component capacity, and less than 0.001 posterior mean
mass outside the density grid.

\begin{table}[!htbp]
\centering
\caption{Final Rasch diagnostics, including the four fits exempt from rerunning.}
\label{tab:H-irt-diagnostics}
\small
\setlength{\tabcolsep}{4pt}
\begin{tabular}{@{}llrrrr@{}}
\toprule
Arm & Hyperprior & Attempt & Max $\widehat R$ & Min bulk ESS & Min tail ESS \\
\midrule
Main & Vague & 3 & 1.0034 & 796 & 759 \\
Main & Gamma(1, 3) reference & 3 & 1.0035 & 829 & 818 \\
Main & TSMM & 2 & 1.0042 & 460 & 473 \\
Main & Dual-Anchor & 3 & 1.0024 & 1,039 & 1,252 \\
Main & Dual-Anchor (hard) & 3 & 1.0030 & 1,057 & 1,460 \\
Sens. & Vague & 1 & 1.0478 & 135 & 110 \\
Sens. & Gamma(1, 3) reference & 1 & 1.0551 & 72 & 98 \\
Sens. & TSMM & 3 & 1.0029 & 975 & 1,487 \\
Sens. & Dual-Anchor & 1 & 1.0674 & 46 & 68 \\
Sens. & Dual-Anchor (hard) & 1 & 1.0223 & 123 & 110 \\
\bottomrule
\end{tabular}
\par\medskip
\begin{minipage}{\textwidth}\small\textit{Note.}
Every fit has four chains and 16,000 retained draws. Extrema are over all
checked quantities. The four sensitivity-arm rows at attempt~1 repeat
priors fitted elsewhere and do not pass the criteria; only their diagnostics
are retained here. All five main-arm fits and sensitivity-arm TSMM pass.
\end{minipage}
\end{table}

\subsubsection{Posterior Counts, Weights and Respondent Estimates}
\label{subsubsec:H-irt-posterior}

For each draw, $K_{500}$ counts clusters among the respondents and the largest observed
share is the size of the largest cluster divided by 500. To obtain the largest
population weight, we draw the occupied weights and residual mass from
their conditional Dirichlet distribution and expand the residual by
stick-breaking until no unseen atom can exceed the largest seen. This
retains the distinction between an observed share and a population weight.
Tables~\ref{tab:H-irt-posterior} and \ref{tab:H-irt-weights} give each
distinct-prior fit once.

Figure~\ref{fig:irt-count} compares the count distributions under the five main-arm priors. The tables also give results for the count-25 stress test.

\begin{figure}[!tp]
  \centering
  \includegraphics[width=\textwidth]{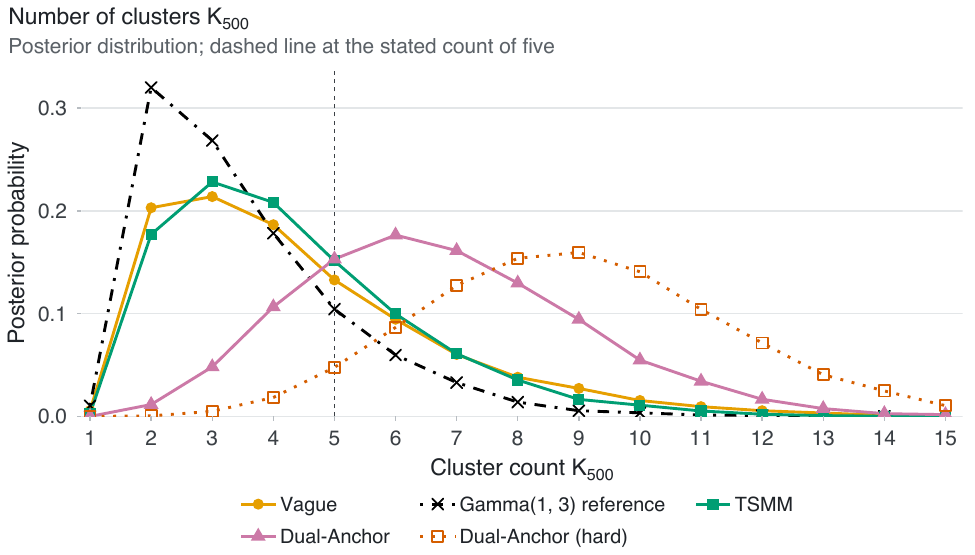}
  \caption{Posterior distribution of the cluster count $K_{500}$ for the vocabulary checklist under five hyperpriors. The dashed line marks the stated count of five. \textit{Note.} The Gamma(1, 3) reference uses the real-data setting of \citet{paganin2023computational}. Dual-Anchor is the count/weight compromise and Dual-Anchor (hard) the weight-bound sensitivity. The figure shows values with posterior probability above 0.005 under at least one prior; each distribution uses 16,000 retained draws.}
  \label{fig:irt-count}
\end{figure}

\begin{table}[!htbp]
\centering
\caption{Posterior count and concentration summaries for the distinct Rasch priors.}
\label{tab:H-irt-posterior}
\small
\setlength{\tabcolsep}{4pt}
\begin{tabular}{@{}lcccc@{}}
\toprule
Hyperprior & \shortstack{$\E(K\mid y)$\\(MCSE)} & 95\% interval & \shortstack{$\Pr(K=1\mid y)$\\(MCSE)} & $\E(\alpha\mid y)$ [95\%] \\
\midrule
Vague & 4.43 (0.048) & [2, 10] & 0.006 (0.001) & 0.65 [0.07, 1.91] \\
Gamma(1, 3) reference & 3.51 (0.031) & [2, 8] & 0.010 (0.002) & 0.38 [0.04, 1.12] \\
TSMM & 4.30 (0.045) & [2, 9] & 0.003 (0.001) & 0.59 [0.12, 1.44] \\
Dual-Anchor & 6.73 (0.033) & [3, 12] & 0.000 (0.000) & 1.29 [0.79, 1.91] \\
Dual-Anchor (hard) & 8.99 (0.034) & [5, 14] & 0.000 (0.000) & 2.00 [2.00, 2.00] \\
Count 25 & 11.72 (0.067) & [5, 20] & 0.000 (0.000) & 3.00 [1.37, 5.20] \\
\bottomrule
\end{tabular}
\par\medskip
\begin{minipage}{\textwidth}\small\textit{Note.}
$K=K_{500}$. All rows use 16,000 draws from converged fits. ``Count 25''
is sensitivity-arm TSMM, whose prior also applies to its two Dual-Anchor
roles. MCSEs are computed from the four-chain array; zero denotes rounding
to the displayed precision.
\end{minipage}
\end{table}

\begin{table}[!htbp]
\centering
\caption{Posterior observed allocation shares and population weights.}
\label{tab:H-irt-weights}
\small
\setlength{\tabcolsep}{5pt}
\begin{tabular}{@{}lccc@{}}
\toprule
Hyperprior & \shortstack{Largest observed\\share [95\%]} & \shortstack{Largest population\\weight [95\%]} & $\Pr(W_{\max}>.5\mid y)$ \\
\midrule
Vague & 0.860 [0.48, 0.98] & 0.859 [0.48, 0.99] & 0.969 \\
Gamma(1, 3) reference & 0.886 [0.52, 0.99] & 0.885 [0.52, 0.99] & 0.984 \\
TSMM & 0.867 [0.50, 0.98] & 0.866 [0.50, 0.99] & 0.974 \\
Dual-Anchor & 0.797 [0.44, 0.97] & 0.796 [0.44, 0.97] & 0.935 \\
Dual-Anchor (hard) & 0.748 [0.40, 0.96] & 0.745 [0.40, 0.96] & 0.896 \\
Count 25 & 0.691 [0.37, 0.95] & 0.687 [0.37, 0.94] & 0.819 \\
\bottomrule
\end{tabular}
\par\medskip
\begin{minipage}{\textwidth}\small\textit{Note.}
Share and weight columns give posterior means and equal-tail 95 percent
intervals. The population calculation includes residual unoccupied mass.
The final column is a posterior probability of population dominance.
\end{minipage}
\end{table}

With the count-25 prior, the posterior mean number of clusters is 11.72 (95 percent interval $[5,20]$). The corresponding mean for main-arm TSMM is 4.30. Under the count-25 prior, the largest population weight has posterior mean
0.687 and exceeds one half with posterior probability 0.819. Increasing the stated count thus raises the posterior cluster count while one
component still carries much of the population mass.

The hyperprior has much less effect on person estimates. We summarize them
by raw score because within-score differences in estimated posterior means
arise only from Monte Carlo error. In each of the six distinct-prior fits, the posterior means increase strictly over the 14 raw scores and satisfy our criterion for equality within a score group. For that check, the largest individual-mean deviation from its group
mean must be at most four times the largest individual MCSE in the group.

The root mean squared change from Vague to Dual-Anchor is 0.012 logits over
the 500 respondents, using score-group means weighted by group size.
The other main-arm priors give 0.007--0.014 logits. The largest
Dual-Anchor score-group shift is 0.031 logits at score 13, with MCSE 0.016;
the largest across the main arm is 0.045 at score 0 under Dual-Anchor
(hard), with MCSE 0.022. For each difference in posterior means, we add the squared MCSEs from the two independent fits
and take the square root. Neither shift meets our reporting
criterion, $\mathrm{MCSE}<|\mathrm{shift}|/3$. For each raw score we express the width of the 95 percent interval as a ratio to the width under Vague. The weighted averages of the 14 ratios, with weights proportional to the numbers of respondents, range from 0.996 to 0.999 over the four comparison priors. The differences are larger at some individual scores. For the three respondents at score 0 the interval under the hard sensitivity is 6.4 percent wider than under Vague. For the 63 respondents at score 13 the TSMM interval is 2.4 percent narrower.
We do not rank respondents within the perfect-score group: 63 respondents
tie at score 13, making a top-decile comparison depend on Monte Carlo noise.

\subsubsection{Distribution of the 500 Respondents and Fitted Population Density}
\label{subsubsec:H-irt-density}

Figures~\ref{fig:H-finite-cdf} and~\ref{fig:H-population-density} give distributions for the observed respondents and for a new exchangeable respondent on the item-anchored logit scale. For the 500 observed respondents, the fraction with ability at or below $x$ in a joint posterior draw is
\[
 F_{500}^{(s)}(x)=\frac{1}{500}\sum_{j=1}^{500}\mathbf{1}\{\eta_j^{(s)}\le x\}.
\]
Averaging these fractions estimates the posterior mean of the finite-population CDF, $500^{-1}\sum_j\Pr(\eta_j\le x\mid y)$ \citep{lee2025improving}. Each respondent has equal weight, including those with the same raw score. Figure~\ref{fig:H-finite-cdf}A shows the posterior mean CDF for each of the five priors in Table~\ref{tab:H-all-priors}, evaluated on $[-4,9]$ at intervals of 0.025 logits. The thin curves are the pointwise 2.5th and 97.5th posterior percentiles for each prior (inverse empirical CDF, R type~1). The interval for Vague is shaded light grey. Each fit contributes four chains of 4,000 retained draws.

For illustration, take a threshold of 3 logits. The posterior mean fraction of respondents with ability at or below this point is 0.473 under $\mathrm{Gamma}(1,1)$ (pointwise 95 percent interval $[0.414,0.534]$) and 0.475 under Dual-Anchor ($[0.414,0.540]$). If instead we assign each respondent a single score, the individual posterior mean, then 279 of 500 respondents (0.558) score at or below 3 under either prior. The empirical distribution of those 500 means takes just 14 values in the Rasch model, one per raw score. Shrinkage makes it narrower than the distribution of abilities \citep{lee2025improving}.

Of the 521 grid points, the largest difference between posterior mean CDFs occurs at 3.600 logits, where the Gamma(1, 3) reference gives 0.6073 and Dual-Anchor (hard) 0.6099, a difference of 0.26 percentage points. The largest difference from Vague, 0.19 percentage points, occurs for Dual-Anchor at 3.400 logits. At both locations the Monte Carlo standard error of the difference is about 0.09 percentage points. The standard errors of the CDF means are below 0.00085 throughout the grid. Of the 2,605 combinations of prior and grid point, 2,282 met the $\widehat R$ and bulk and tail ESS criteria. At the other 323, ties in the upper tail precluded calculation of tail ESS. The largest range of an interval endpoint across chains was 0.016. We did not calculate Monte Carlo errors for the endpoints or simultaneous bands for differences between CDFs.

Figure~\ref{fig:H-population-density} gives the ability density for a new exchangeable respondent. At each posterior draw we calculate the mixture density from the population weights, component means and variances, then average the densities over draws. We use the grid $[-10,14]$ with 0.025-logit spacing and show $[-4,9]$. In preliminary fits the planned grid $[-6,6]$ omitted over six percent of the mass. The wider grid contains at least 0.999 of the posterior mean mass in every production chain. For individual draws the omitted mass can be larger, with maxima of about 0.013--0.019 across DPM fits. We plot the densities without renormalizing them.

For the Gaussian comparison we use the same Rasch likelihood and item constraint, with
$\eta_j\sim\mathcal N(\mu,\sigma_\eta^2)$,
$\mu\sim\mathcal N(0,3)$ and
$\sigma_\eta^2\sim\mathrm{InvGamma}(2.01,1.01)$.
The prior variance of the location is 3 in the Gaussian model and 2 for each DPM component. The variance parameters also have different roles in a single normal distribution and a mixture. Differences between the fits cannot therefore be attributed solely to the restriction to one component. This Gaussian benchmark uses four chains of 8,000 iterations, 4,000 warmup and thinning 2, giving 8,000 retained draws. Its maximum $\widehat R$ is 1.0093 and minimum bulk and tail ESS are 636 and 1,350 over the checked parameters, score-group means and density values. It passed at the first attempt. The posterior mean location is 3.06 logits and the posterior mean latent standard deviation is 2.10. This latter quantity, $\E(\sigma_\eta\mid y)$, differs from the standard deviation of the posterior mean density, which also reflects uncertainty about location.

The five DPM priors give similar posterior mean density curves. The Gaussian curve is broader, with a lower peak. For Vague and Dual-Anchor the $L^1$ distance between the mean curves on $[-10,14]$, calculated by the trapezoidal rule, is 0.0061. The largest absolute difference at a grid point is 0.0014. For Vague and the Gaussian model the corresponding values are 0.1527 and 0.0338. These distances describe the mean curves on the grid. We did not pair draws across priors, calculate intervals for the differences or test for equivalence. The small differences among the DPM curves must be interpreted in light of their pointwise Monte Carlo errors and give no basis for ranking the priors.

The grey shading and thin limits in Figure~\ref{fig:H-population-density} give pointwise 95 percent posterior intervals. We checked convergence wherever the pooled posterior mean density was at least 0.01. Every such grid point met the $\widehat R$ and bulk and tail ESS criteria. We did not check the lower-density tails or store Monte Carlo errors for the interval endpoints. We used the residual-mass tolerance in Section~\ref{subsec:H-distribution} only for the school and study densities.

\subsubsection{Reproducibility}
\label{subsubsec:H-irt-repro}

The replication package supplies the raw-score frequencies, item summaries,
exported item difficulties and analysis code needed to reproduce the WLE
scoring and reliability calculation. The base-R verifier conditions on the
exported item difficulties: it neither refits the Rasch model nor propagates
uncertainty in those estimates. Response-level records and person-level
posterior draws are not redistributed. A full refit therefore requires the
user to obtain the source responses through the Item Response Warehouse and
to retain its license and attribution.

\subsection{Effect distributions among the observed units and in the population}
\label{subsec:H-distribution}

We use the original joint draws of the latent effects for the 79 schools, 48 studies and 500 respondents to compare the same observed units under different priors. For new exchangeable units we compare the fitted population densities. The additional STAR and writing-to-learn runs stored the mixture parameters under the three main priors. The original Rasch fits retained these parameters (Section~\ref{subsubsec:H-irt-density}).

\paragraph{Distribution among the observed units.}
For posterior draw $s$, the CDF and standard deviation of the effects across the fixed units are
\[
 F_J^{(s)}(x)=\frac{1}{J}\sum_{j=1}^J\mathbf{1}\{\tau_j^{(s)}\leq x\},
 \qquad
 S_J^{(s)}=\left\{\frac{1}{J}\sum_{j=1}^J
 (\tau_j^{(s)}-\bar\tau^{(s)})^2\right\}^{1/2}.
\]
The IQR within a draw is the difference between the 0.75 and 0.25 quantiles across units (inverse empirical CDF, R type 1). Table~\ref{tab:H-finite-spread} gives the posterior means of the SD and IQR and their equal-tailed 95 percent intervals (type 7 posterior quantiles). Both measure spread among the fixed units within a draw. Averaging the effects first would give the spread of the unit posterior means. Figure~\ref{fig:H-finite-cdf} shows the posterior mean CDFs and pointwise 95 percent intervals at the 161 grid points of $[-2,2]$ for schools and studies and on the 521-point vocabulary grid described in Section~\ref{subsubsec:H-irt-density}. The intervals use type 1 posterior quantiles because the CDF is discrete.

Within each fit we grouped the draws by cluster count. For draws with one cluster, the mean SD of the school effects ranges from 0.21 to 0.22 across the five priors, close to the range of 0.23 to 0.24 for draws with two or more. The corresponding means for the studies are 0.24 and 0.26.

Each school and study archive contains 4,000 approximately equally spaced draws after warmup from each of four chains. The spacing is two or three iterations, or 15 or 16, depending on the accepted schedule. Rank-normalized split-$\widehat R$, bulk and tail effective sample sizes and Monte Carlo standard errors (MCSEs) are calculated from the retained draws in their original order. They apply to these subsequences rather than the complete unthinned chains. We require $\widehat R\leq1.01$ and effective sample sizes of at least 400 for both bulk and tails. All 30 SD and IQR summaries for schools, studies and respondents met these criteria. Endpoint MCSEs use the package's default order-statistic approximation. For the two STAR TSMM lower endpoints, MCSEs of 0.0054 and 0.0081 leave the second decimal place uncertain.

We assess CDF convergence pointwise where values vary across draws. Of 805 combinations of prior and grid point for STAR, 456 met the criteria, 47 had a constant CDF and 302 lacked computable diagnostics. The corresponding counts for Bangert--Drowns are 438, 84 and 283. None of the available diagnostics failed. The constant and unchecked values occur in the tails, where the curves require caution. We calculated Monte Carlo errors only for the CDF means. The endpoints are discrete (type 1) quantiles. Differences between the mean CDFs are descriptive, with no pairing of draws across priors.

\begin{table}[!htbp]
\centering
\caption{Posterior spread of latent effects across the observed units under five priors}
\label{tab:H-finite-spread}
\small
\setlength{\tabcolsep}{4pt}
\begin{tabular*}{\linewidth}{@{\extracolsep{\fill}}llcc@{}}\toprule
Application & Prior & SD [95\% interval] & IQR [95\% interval] \\
\midrule
STAR & $\mathrm{Gamma}(1,1)$ & 0.226 [0.11, 0.34] & 0.301 [0.14, 0.47] \\
 & TSMM & 0.234 [0.13, 0.35]$^{*}$ & 0.305 [0.15, 0.48]$^{*}$ \\
 & Dual-Anchor & 0.236 [0.12, 0.35] & 0.305 [0.15, 0.48] \\
 & DORO-Unif & 0.228 [0.12, 0.34] & 0.301 [0.15, 0.47] \\
 & Dual-Anchor (hard) & 0.240 [0.13, 0.35] & 0.311 [0.16, 0.48] \\
\addlinespace
Writing-to-learn & $\mathrm{Gamma}(1,1)$ & 0.252 [0.17, 0.34] & 0.334 [0.21, 0.48] \\
 & TSMM & 0.250 [0.17, 0.34] & 0.334 [0.20, 0.48] \\
 & Dual-Anchor & 0.256 [0.17, 0.35] & 0.332 [0.20, 0.48] \\
 & DORO-Unif & 0.253 [0.17, 0.34] & 0.334 [0.20, 0.48] \\
 & Dual-Anchor (hard) & 0.259 [0.17, 0.35] & 0.334 [0.20, 0.48] \\
\addlinespace
Vocabulary & $\mathrm{Gamma}(1,1)$ & 2.03 [1.84, 2.31] & 2.42 [2.01, 2.77] \\
 & TSMM & 2.02 [1.83, 2.28] & 2.42 [2.04, 2.77] \\
 & Dual-Anchor & 2.02 [1.84, 2.29] & 2.41 [2.00, 2.78] \\
 & $\mathrm{Gamma}(1,3)$ reference & 2.03 [1.84, 2.28] & 2.43 [2.05, 2.76] \\
 & Dual-Anchor (hard) & 2.03 [1.84, 2.30] & 2.41 [1.99, 2.78] \\
\bottomrule\end{tabular*}
\par\medskip
\begin{minipage}{\linewidth}\small\textit{Note.} Both summaries are in achievement SDs for STAR, standardized mean differences for writing-to-learn, and logits for vocabulary. Entries give the posterior mean and a 95 percent interval for each within-draw summary. $^{*}$Monte Carlo SEs of the lower TSMM limits for STAR were 0.0054 (SD) and 0.0081 (IQR). For vocabulary, Monte Carlo SEs were 0.012--0.019 logits for the upper SD limits and 0.012--0.024 logits for the lower IQR limits.\end{minipage}
\end{table}

\begin{figure}[!htbp]
\centering
\includegraphics[width=\textwidth]{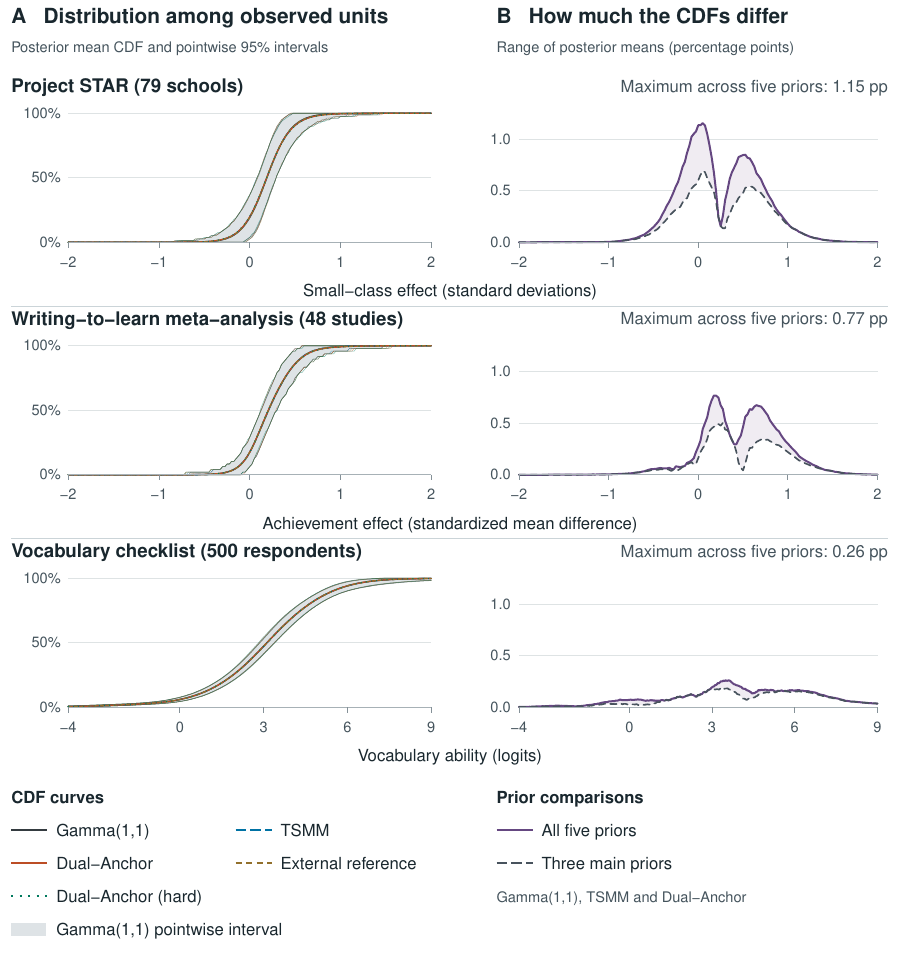}
\caption{Finite-sample distributions under five priors. A gives the posterior means and pointwise 95 percent intervals of $F_J(x)$. B gives the maximum minus the minimum posterior mean CDF at each $x$, in percentage points, for the three main priors and for all five priors. DORO-Unif is the external reference for the schools and studies, $\mathrm{Gamma}(1,3)$ for vocabulary. The grid spacing is 0.025 in each application.}
\label{fig:H-finite-cdf}
\end{figure}

\paragraph{Comparing quantiles and individual estimates.}
For each prior, pooling the 16,000 draws for each of the $J$ observed units with equal weight gives the posterior mean finite-population CDF. We sorted the pooled values from each fit and matched them by rank. The RMS of the differences over all ranks equals the 2-Wasserstein distance ($W_2$) between the two estimated distributions. The comparison uses every draw, including those outside the CDF display grid. Main-text Figure~\ref{fig:app-effect-changes}B joins the signed differences (Dual-Anchor minus $\mathrm{Gamma}(1,1)$) at 199 equally spaced percentiles from the 0.5th to the 99.5th. To compare individual estimates we matched posterior means by school or study. For respondents we used the differences between raw-score means, weighting each by the number of respondents with that score.

\begin{table}[!htbp]
\centering
\caption{Changes in individual estimates and distribution quantiles across the three main priors}
\label{tab:H-effect-changes}
\small
\setlength{\tabcolsep}{4pt}
\begin{adjustbox}{max width=\linewidth}
\begin{tabular}{lrrr}
\toprule
& \multicolumn{2}{c}{Distribution quantiles} & \shortstack{Individual}\\ \cmidrule(lr){2-3} Prior pair & All ($W_2$) & 5th--95th percentiles & estimates \\
\midrule
\multicolumn{4}{l}{\itshape Tennessee STAR \quad 79 schools, achievement SDs} \\
$\mathrm{Gamma}(1,1)$ and TSMM & 0.00891 & 0.00430 & 0.00714 \\
$\mathrm{Gamma}(1,1)$ and Dual-Anchor & 0.01167 & 0.00448 & 0.00879 \\
TSMM and Dual-Anchor & 0.00309 & 0.00068 & 0.00307 \\
\addlinespace[.9em]
\multicolumn{4}{l}{\itshape Writing-to-learn \quad 48 studies, standardized mean differences} \\
$\mathrm{Gamma}(1,1)$ and TSMM & 0.00192 & 0.00075 & 0.00255 \\
$\mathrm{Gamma}(1,1)$ and Dual-Anchor & 0.00833 & 0.00243 & 0.00707 \\
TSMM and Dual-Anchor & 0.00910 & 0.00264 & 0.00826 \\
\addlinespace[.9em]
\multicolumn{4}{l}{\itshape Vocabulary \quad 500 respondents, logits} \\
$\mathrm{Gamma}(1,1)$ and TSMM & 0.02209 & 0.00729 & 0.01323 \\
$\mathrm{Gamma}(1,1)$ and Dual-Anchor & 0.01811 & 0.00728 & 0.01227 \\
TSMM and Dual-Anchor & 0.01847 & 0.00471 & 0.00441 \\
\bottomrule
\end{tabular}
\end{adjustbox}
\par\medskip
\begin{minipage}{\linewidth}\small\textit{Note.} Entries are root mean square changes on the original effect scale within each application. The 5th--95th percentile column excludes the lowest and highest 5 percent of quantile ranks and averages over the remaining ranks. Each fit contributes 16,000 retained draws for every observed unit.\end{minipage}
\end{table}

Table~\ref{tab:H-effect-changes} gives all three pairwise comparisons. For $\mathrm{Gamma}(1,1)$ and Dual-Anchor the quantile shifts exceeded the individual shifts in RMS terms in all three applications. Restricting the calculation to ranks between the 5th and 95th percentiles reversed the ordering in all three. For $\mathrm{Gamma}(1,1)$ versus TSMM in the writing studies, even the full-distribution shift was smaller than the individual shift. In Table~\ref{tab:H-effect-splits} the vocabulary distance of 0.0181 logits across priors is close to the within-prior distances (0.0150 and 0.0167 logits). The latter use half as many draws per estimate and give a descriptive check on Monte Carlo variation.

\begin{table}[!htbp]
\centering
\caption{Sensitivity of estimated changes to the choice of chains}
\label{tab:H-effect-splits}
\small
\setlength{\tabcolsep}{4pt}
\begin{adjustbox}{max width=\linewidth}
\begin{tabular}{lrrrrr}
\toprule
 & \multicolumn{3}{c}{Between priors} & \multicolumn{2}{c}{Within a prior}\\ \cmidrule(lr){2-4}\cmidrule(l){5-6} RMS change & All chains & Chains 1--2 & Chains 3--4 & $\mathrm{Gamma}(1,1)$ & Dual-Anchor \\
\midrule
\multicolumn{6}{l}{\itshape Tennessee STAR \quad 79 schools, achievement SDs} \\
All quantiles ($W_2$) & 0.01167 & 0.01109 & 0.01234 & 0.00311 & 0.00185 \\
5th--95th percentiles & 0.00448 & 0.00395 & 0.00503 & 0.00193 & 0.00080 \\
Individual estimates & 0.00879 & 0.00835 & 0.00986 & 0.00423 & 0.00364 \\
\addlinespace[.9em]
\multicolumn{6}{l}{\itshape Writing-to-learn \quad 48 studies, standardized mean differences} \\
All quantiles ($W_2$) & 0.00833 & 0.00732 & 0.00947 & 0.00157 & 0.00267 \\
5th--95th percentiles & 0.00243 & 0.00262 & 0.00228 & 0.00082 & 0.00059 \\
Individual estimates & 0.00707 & 0.00609 & 0.00861 & 0.00280 & 0.00369 \\
\addlinespace[.9em]
\multicolumn{6}{l}{\itshape Vocabulary \quad 500 respondents, logits} \\
All quantiles ($W_2$) & 0.01811 & 0.02114 & 0.02356 & 0.01496 & 0.01671 \\
5th--95th percentiles & 0.00728 & 0.00620 & 0.00943 & 0.00786 & 0.00572 \\
Individual estimates & 0.01227 & 0.01253 & 0.01288 & 0.00754 & 0.00574 \\
\bottomrule
\end{tabular}
\end{adjustbox}
\par\medskip
\begin{minipage}{\linewidth}\small\textit{Note.} The first three columns compare $\mathrm{Gamma}(1,1)$ with the selected Dual-Anchor prior using the indicated chains from each fit. The last two compare chains 1--2 with chains 3--4 within a single fit. Each two-chain subset contributes 8,000 draws.\end{minipage}
\end{table}

For one school and two studies, switching from Vague to TSMM changes whether the 95 percent interval includes zero. In that order the lower endpoints are 0.0096, $-0.0049$ and $-0.0010$ under Vague, and $-0.0045$, 0.00001 and 0.00005 under TSMM. Their Monte Carlo standard errors are 0.0045, 0.0038 and 0.0040 under Vague, and 0.0050, 0.0037 and 0.0042 under TSMM. All three TSMM endpoints lie within one Monte Carlo standard error of zero. Figure~\ref{fig:H-unit-changes} shows these endpoints and the vocabulary score changes.

\begin{figure}[!htbp]
\centering
\includegraphics[width=\textwidth]{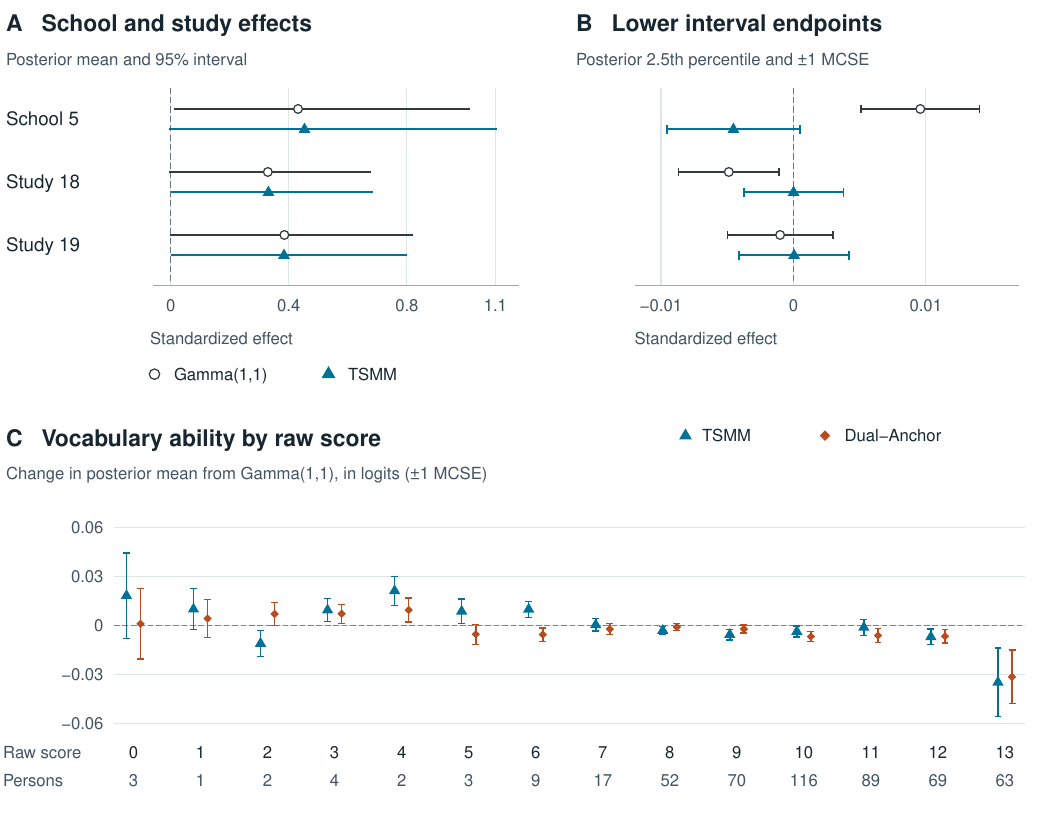}
\caption{Sensitivity of individual estimates. A shows the three units whose equal-tailed 95 percent intervals changed their relation to zero under $\mathrm{Gamma}(1,1)$ and TSMM. B shows the lower interval endpoints with one Monte Carlo SE on either side. C shows differences in posterior mean vocabulary ability from $\mathrm{Gamma}(1,1)$ at each raw score, also with one Monte Carlo SE. The differences use independent fits.}
\label{fig:H-unit-changes}
\end{figure}

\paragraph{Latent population density.}
For each of the school and study fits under $\mathrm{Gamma}(1,1)$, TSMM and the selected Dual-Anchor prior we estimate the population density
$f_G(x)=\int\phi(x;\mu,\sigma^2)\,dG(\mu,\sigma^2)$
without convolving it with observation error. The six additional runs retain the original data, priors, final schedules and seeds, and store the occupied labels, component parameters and hyperparameters at each draw. They reproduce the original count and unit summaries to floating-point precision, with the same interval decisions. The replication materials link each run to its original.

Conditional on a retained state, the occupied weights and the remaining mass have a Dirichlet distribution with parameters $(n_1,\ldots,n_{K_J},\alpha)$. Component labels match counts to atoms. For the residual DP we draw independent $\mathrm{Beta}(1,\alpha)$ sticks. Atoms come from the product base at that draw, with independent normal means and inverse-Gamma variances. The intervals therefore include conditional random-$G$ variation that would be lost by retaining only occupied atoms and renormalizing their weights. We keep the weights unnormalized. The unrepresented mass is at most $10^{-6}$ per draw. The bound concerns missing probability mass and gives no uniform pointwise density-error bound.

We evaluate all six densities on $[-4,4]$ at intervals of 0.025. In each chain the conservative upper bound on omitted grid mass averages at most 0.001, though individual draws can exceed it. The discrepancy between the mean trapezoidal integral and the mean of the analytically calculated Gaussian masses represented on the grid is at most $10^{-4}$. All 1,926 density values (321 grid points in each of six fits) meet the same $\widehat R$ and ESS criteria. For every curve, Monte Carlo standard errors of the mean or an endpoint exceed 0.005 at some grid points. The largest is about 0.145 for an upper endpoint under STAR TSMM. Near the mode the interval limits are less precise than their appearance suggests.

Figure~\ref{fig:H-population-density} shows the school and study densities on $[-2,2]$ and the vocabulary densities on $[-4,9]$ logits. Relative to Vague, the $L^1$ distances between posterior mean densities on the full grid (trapezoidal rule) are 0.02695 and 0.02907 for STAR and 0.00394 and 0.02477 for Bangert--Drowns, in each case for TSMM and Dual-Anchor, respectively. We did not average distances over posterior draws, estimate Monte Carlo errors for the distances or test for equivalence. The school and study density runs cover the three main priors. The vocabulary panel includes all five DPM priors and the Gaussian benchmark (Section~\ref{subsubsec:H-irt-density}).

\begin{figure}[!htbp]
\centering
\includegraphics[width=\textwidth]{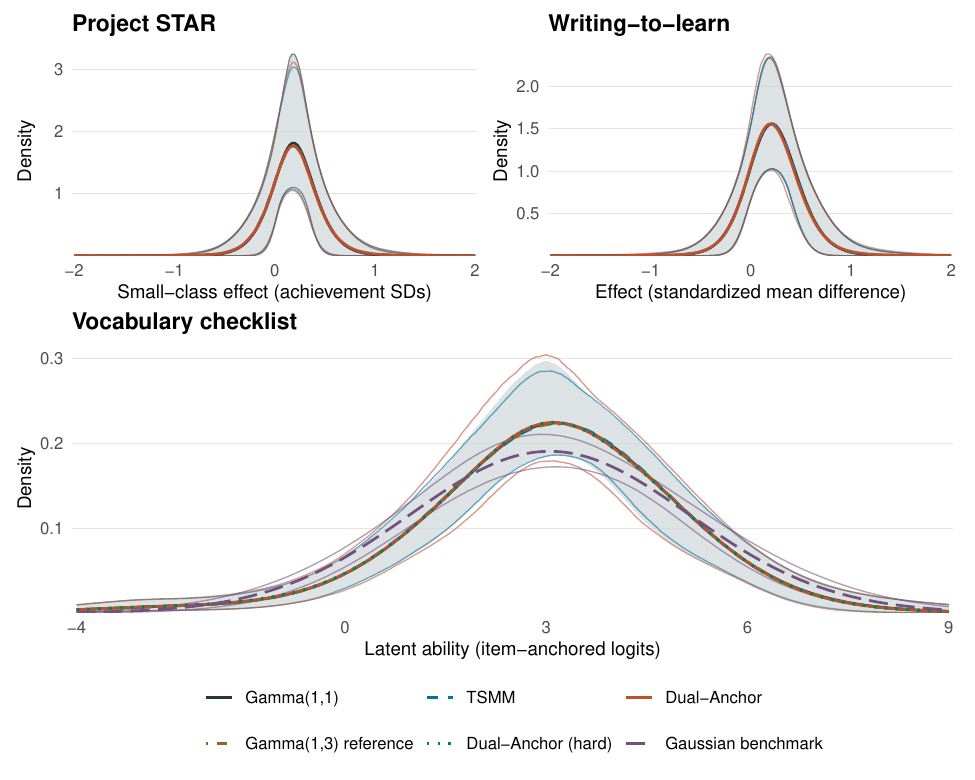}
\caption{Posterior mean densities for a new school, study or respondent. Grey shading gives pointwise 95 percent intervals under $\mathrm{Gamma}(1,1)$. Thin curves give interval limits for TSMM and Dual-Anchor and, for vocabulary, for the Gaussian benchmark. The Gaussian and DPM Rasch fits use different priors on location.}
\label{fig:H-population-density}
\end{figure}

\FloatBarrier

\subsubsection{Pooling of schools and studies with large observed effects}
\label{subsubsec:H-tail-pooling}

For the comparison in Figure~\ref{fig:app-tail-unit-pooling}, we ranked the 79 STAR schools by their first-stage estimates from the classroom random-intercept model and the 48 writing-to-learn studies by their observed standardized mean differences. We included units at or above the 90th percentile in the upper group and those from the 25th through the 75th percentile in the middle group (type~7 sample quantiles, including ties). We used the same groups under both priors. For each unit and each posterior draw, we calculated the fraction of units in the middle group assigned to its cluster, excluding the unit itself. We averaged over draws and then over units in each group. To measure shrinkage we calculated the absolute distance between each observed effect and its posterior mean and averaged these distances within each group. For every upper-tail unit the posterior mean lay between the observed effect and the sample median of the observed effects under both priors.

\begin{table}[!htbp]
\centering
\caption{Co-clustering with the middle group and distance between observed effects and posterior means}
\label{tab:H-tail-pooling}
\small
\setlength{\tabcolsep}{3.5pt}
\begin{adjustbox}{max width=\linewidth}
\begin{tabular}{@{}lrrrrrrr@{}}
\toprule
& & \multicolumn{3}{c}{Co-clustering probability (\%)} & \multicolumn{3}{c}{Mean absolute distance} \\
\cmidrule(lr){3-5}\cmidrule(l){6-8}
Group & $n$ & Gamma & Dual & Reduction & Gamma & Dual & Reduction \\
\midrule
\multicolumn{8}{l}{\itshape Project STAR (achievement SDs)} \\
Upper & 8 & 86.73 & 57.87 & 28.87 & 0.6633 & 0.6443 & 0.0190 \\
Middle & 39 & 89.14 & 63.13 & 26.01 & 0.1008 & 0.0988 & 0.0019 \\
Excess reduction & & & & 2.86 & & & 0.0170 \\
Monte Carlo SE & & & & 0.17 & & & 0.0018 \\
\addlinespace[.5em]
\multicolumn{8}{l}{\itshape Writing-to-learn (standardized mean differences)} \\
Upper & 5 & 83.37 & 63.80 & 19.57 & 0.4627 & 0.4506 & 0.0121 \\
Middle & 25 & 88.28 & 72.08 & 16.20 & 0.0567 & 0.0566 & 0.0001 \\
Excess reduction & & & & 3.37 & & & 0.0120 \\
Monte Carlo SE & & & & 0.39 & & & 0.0029 \\
\bottomrule
\end{tabular}
\end{adjustbox}
\par\medskip
\begin{minipage}{\linewidth}\small\textit{Note.} Gamma and Dual denote $\mathrm{Gamma}(1,1)$ and the selected Dual-Anchor prior. Here a reduction is $\mathrm{Gamma}(1,1)$ minus Dual-Anchor. A positive reduction in this table corresponds to a negative change in Figure~\ref{fig:app-tail-unit-pooling}. Excess reduction is the reduction in the upper group minus that in the middle group. Reductions in probability and their Monte Carlo SEs are in percentage points. Differences use the unrounded estimates.\end{minipage}
\end{table}

Each fit contributed four chains of 4,000 retained draws. We took the upper-minus-middle contrast within each fit before combining the Monte Carlo variances from the two independent fits. To obtain the Monte Carlo SEs in Table~\ref{tab:H-tail-pooling}, we linearized each absolute distance at its estimated posterior mean. Resampling blocks of 200 draws within each chain gave Monte Carlo SEs of 0.0020 in STAR and 0.0027 in writing-to-learn for the excess reduction in mean distance. In writing-to-learn, one effect from Willey (1988, study~46) accounted for 71 percent of the total reduction in distance among the five upper-tail studies. For three of the five the change was smaller than its Monte Carlo SE.

We repeated the comparison using the upper 5, 15 and 20 percent of the observed effects, keeping the middle group fixed. At these three cutoffs the excess reductions in co-clustering probability were 3.82, 2.27 and 1.83 percentage points for STAR and 4.63, 2.47 and 2.32 for writing-to-learn. The corresponding excess reductions in mean distance were 0.0252, 0.0127 and 0.0095 achievement SDs for STAR and 0.0192, 0.0073 and 0.0052 standardized mean differences for writing-to-learn. The upper 5 percent included only four schools and three studies.

Figure~\ref{fig:H-tail-pooling-count} compares the upper tenth and middle half within draws having the same cluster count. All pairs shared a cluster when $K_J=1$. With two clusters, the probability of sharing a cluster was higher for two middle units than for an upper and a middle unit. In STAR this difference was 4.63 percentage points under $\mathrm{Gamma}(1,1)$ and 4.43 under Dual-Anchor. In writing-to-learn the differences were 9.73 and 9.19 percentage points. Each count shown in the figure occurred in at least 174 retained draws per chain under both priors in both applications.

\begin{figure}[!htbp]
\centering
\includegraphics[width=\textwidth]{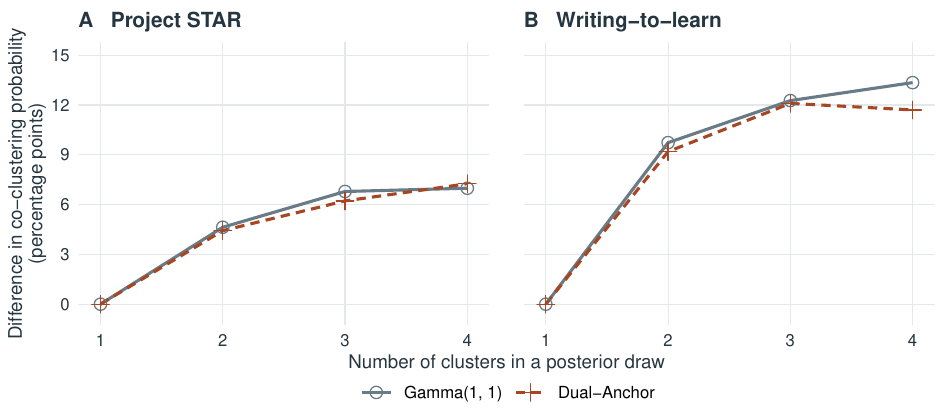}
\caption{Co-clustering by posterior cluster count. At each count, the vertical axis gives the probability that two distinct middle-group units shared a cluster minus the probability that an upper-tail unit shared a cluster with a middle-group unit, in percentage points. Groups are defined by the observed effects as in Table~\ref{tab:H-tail-pooling}.}
\label{fig:H-tail-pooling-count}
\end{figure}

\FloatBarrier

\subsection{Data and Code Availability}
\label{subsec:H6-repro}

Section~\ref{subsec:G7-repository} gives the data sources and instructions for rebuilding the application displays and refitting the models.
Section~\ref{subsubsec:H-irt-repro} describes the materials for reproducing the vocabulary scores.


\end{document}